\documentclass[journal,twoside,web]{ieeecolor}
\usepackage{generic}
\usepackage{cite}
\usepackage{amsmath,amssymb,amsfonts}
\usepackage{graphicx}
\graphicspath{{../figures/}}
\usepackage{enumerate}

\newcommand{\M}{\mathcal{M}}

\newcommand{\depth}{\mathrm{depth}}

\newtheorem{definition}{Definition}
\newtheorem{theorem}{Theorem}
\newtheorem{lemma}{Lemma}

\newtheorem{remark}{Remark}
\newtheorem{example}{Example}
\newtheorem{proposition}{Proposition}

\newtheorem{convention}{Convention}

\usepackage[dvipsnames]{xcolor}
\definecolor{dark_red_color}{RGB}{201,0,0.3}

\usepackage{algorithm}
\usepackage{algpseudocode}
\usepackage{algorithmicx}

\def\longversion{1} 

\usepackage{hyperref}
\hypersetup{hidelinks=true}
\usepackage{textcomp}
\def\BibTeX{{\rm B\kern-.05em{\sc i\kern-.025em b}\kern-.08em
    T\kern-.1667em\lower.7ex\hbox{E}\kern-.125emX}}
\begin{document}
\title{Efficient Planning in Large-scale Systems Using \\ Hierarchical Finite State Machines}
\author{Elis Stefansson and Karl H. Johansson 
\thanks{This work was supported in part by Swedish Research Council Distinguished Professor Grant 2017-01078, Knut and Alice Wallenberg Foundation Wallenberg Scholar Grant, and Swedish Strategic Research Foundation FUSS SUCCESS Grant. } 
\thanks{E. Stefansson and K.H. Johansson are with the School of Electrical Engineering and Computer Science, KTH Royal Institute of Technology, 11428, Stockholm Sweden (e-mail: \{elisst,kallej\}@kth.se). The authors are also affiliated with Digital Futures.}}

\maketitle

\begin{abstract}
We consider optimal planning in a large-scale system formalised as a hierarchical finite state machine (HFSM). A planning algorithm is proposed computing an optimal plan between any two states in the HFSM, consisting of two steps: A pre-processing step that computes optimal exit costs of the machines in the HFSM, with time complexity scaling with the number of machines; and a query step that efficiently computes an optimal plan by removing irrelevant subtrees of the HFSM using the optimal exit costs. The algorithm is reconfigurable in the sense that changes in the HFSM are handled with ease, where the pre-processing step recomputes only the optimal exit costs affected by the change. The algorithm can also exploit compact representations that groups together identical machines in the HFSM, where the algorithm only needs to compute the optimal exit costs for one of the identical machines within each group, thereby avoid unnecessary recomputations. We validate the algorithm on large systems with millions of states and a robotic application. It is shown that our approach outperforms Dijkstra's algorithm, Bidirectional Dijkstra and Contraction~Hierarchies.

\end{abstract}

\begin{IEEEkeywords}
Efficient planning, hierarchical finite state machine, large-scale system, modularity
\end{IEEEkeywords}

\section{Introduction}
\subsection{Motivation}
Large-scale systems are ubiquitous with examples such as power grid and intelligent transportation systems. To manage such large-scale systems, a common approach is to exploit their modularity by breaking up the system into components, separately analyse these components and then aggregate the~results. One theory that formalises modularity is the notion of a hierarchical finite state machine (HFSM). An HFSM is simply a collection of finite state machines (FSMs) nested into a hierarchy, where the hierarchy specifies which FSM is nested into which. Originally proposed by \cite{harel1987statecharts}, the intention was to model reactive systems, such as hand-watches. In this research, we are instead interested in using HFSM towards a formal modular theory for large-scale \emph{control systems} (where we additionally have to choose control inputs), and to design control algorithms that efficiently compute optimal plans for such systems by exploiting their modularity.




\subsection{An illustrative example}\label{illustrative_example} 
Consider the robot application, illustrated in Fig. \ref{fig:robot_example_detailed_overview}. Here, a robot is moving inside a house. At each location in the house, the robot can perform a task, in this case work at a lab desk. The robot can also go outside the house and move to other houses. This control system (where we control the robot) can be naturally modelled as an HFSM with three layers corresponding to the houses, the locations inside a house, and the lab desk. The first research question is \emph{how to exploit the hierarchical structure of such systems to design control algorithms that compute optimal plans efficiently.}

Consider now that one of the houses (House 2) gets some locations blocked due to maintenance work, as seen in Fig.~\ref{fig:robot_example_change_overview}. In this case, the robot might need to replan its route to account for the change. However, only a part of the system has changed and therefore, only a part of the planning should be updated. The second research question is \emph{how to construct a reconfigurable planning algorithm that handles such changes~efficiently.}

Finally, we note that there might be several houses that are identical. In this case, the identical houses should ideally be grouped together into a more compact representation, keeping just one house as a reference. Furthermore, this compact representation should be exploited by the planning algorithm to avoid unnecessary recomputations. The third research question is therefore \emph{how planning algorithms can exploit such a compact representation to obtain faster computations of optimal~plans.}

\begin{figure}
	\centering
  \includegraphics[width=0.4\textwidth]{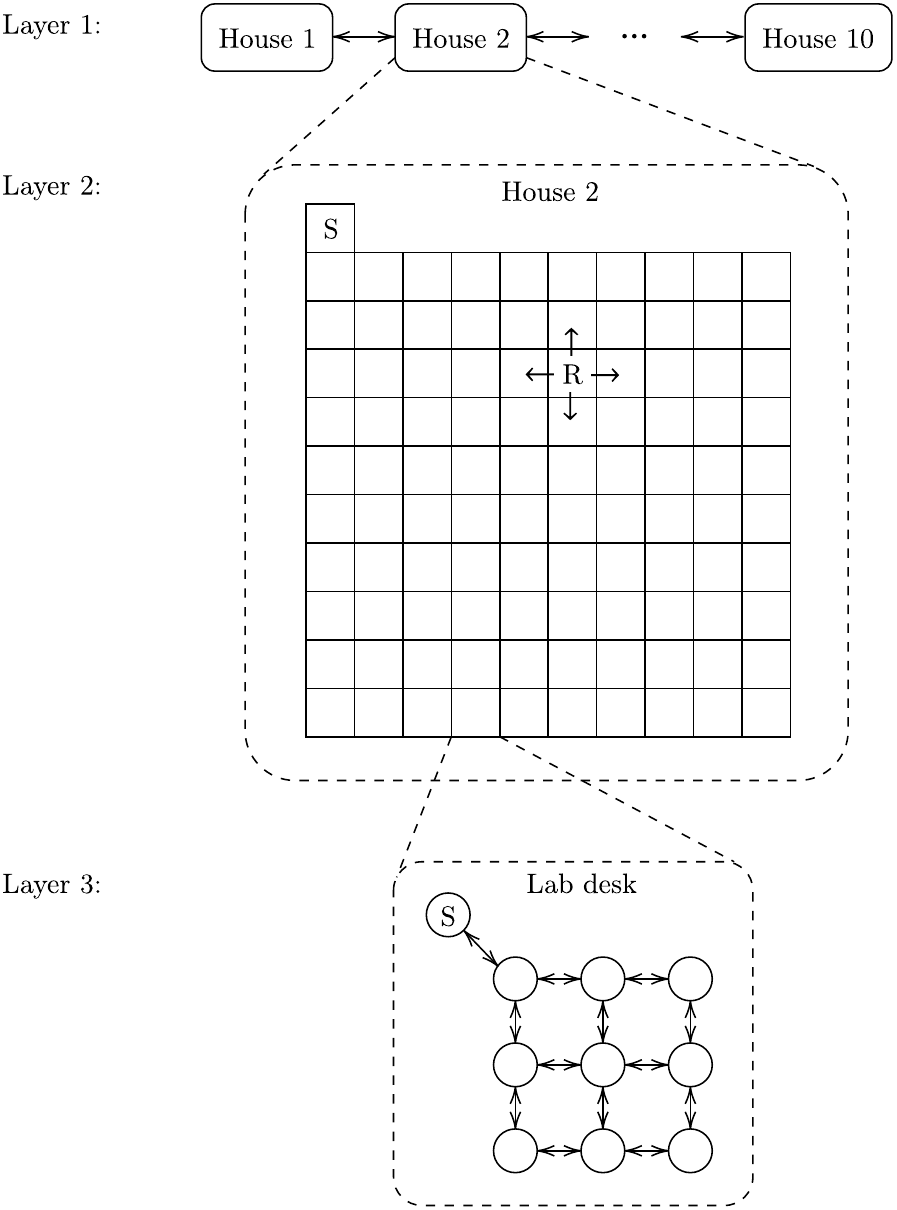}
  \caption{Robot application modelled as an HFSM with three~layers.}  
  \label{fig:robot_example_detailed_overview}
\end{figure}

\subsection{Contribution}
The main contribution of this paper is to formalise modular large-scale systems as hierarchical finite state machines and develop a planning algorithm that efficiently computes optimal plans for such systems. More precisely, our contribution is~three-fold:

First, we extend the HFSM theory in \cite{biggar2021modular} to machines with costs, i.e., Mealy machines (MMs). We call the resulting hierarchical machines for hierarchical Mealy machines (HiMMs). We justify this construction by proving that an HiMM is modular in a strict mathematical sense, extending results from~\cite{biggar2021modular}. We then introduce modifications in such systems and extend the setup to the case when machines may be reused at several parts of the HiMM. Finally, we show how the theory can be extended to compactly represent identical machines in the HiMM by grouping them together and saving just one machine as a~reference.

Second, we propose a planning algorithm for computing an optimal plan between any two states in an HiMM. The algorithm consists of two steps. The first step is a preprocessing step, the Optimal Exit Computer, that computes optimal exits costs of each MM in the HiMM, with time complexity $O(|X|)$ where $|X|$ is the number of MMs in the HiMM. This step needs to be computed only once for a fixed HiMM. The second step is a query step, the Online Planner, that efficiently computes an optimal plan between any two states by removing irrelevant subtrees of the HFSM using the optimal exit costs, with time complexity $O(\mathrm{depth}(Z) \log ( \mathrm{depth}(Z)))$ (compared to Dijkstra's algorithm with time complexity scaling exponentially with $\mathrm{depth}(Z)$ in the worst case), where $\mathrm{depth}(Z)$ is the number of layers in the HiMM. The algorithm is reconfigurable in the sense that a modification in the HiMM is handled with ease, because the Optimal Exit Computer only needs to update affected optimal exit costs in time $O(\mathrm{depth}(Z))$. We also extend the algorithm to exploit identical machines grouped together. Here, significant speed-ups can be achieved by the Optimal Exit Computer, in the best case reducing the time complexity from $O(|X|)$ to $O(\mathrm{depth}(Z))$. These time complexity expressions are simplified for brevity, see the remaining paper for the full expressions.

Third, we consider numerical evaluations with systems up to millions of states and a robot application. We compare our planning algorithm with Dijkstra's algorithm, Bidirectional Dijkstra and Contraction Hierarchies. Our planning algorithm is several times faster (at query time) than Dijkstra's algorithm and Bidirectional Dijkstra (but with a slower preprocessing step), while Contraction Hierarchies computes an optimal plan several times faster than ours but at the expense of a preprocessing step being several magnitudes slower than ours. Thus, our planning algorithm obtains a good balance between preprocessing time and query time. Furthermore, grouping together identical machines, we can compute optimal plans for enormous systems (such as an HiMM with equivalent MM having $10^{150}$ states), where the other methods become intractable.


Early versions of this paper have been presented at conferences as \cite{stefansson2023ecc,stefansson2023cdc}. This paper extends the conference papers considerably by considering HiMMs with identical machines with a corresponding hierarchical planning algorithm, and justifying the definition of an HiMM by proving it is modular in a strict mathematical sense. We also provide additional~simulations. 



\begin{figure}
	\centering
  \includegraphics[width=0.49\textwidth]{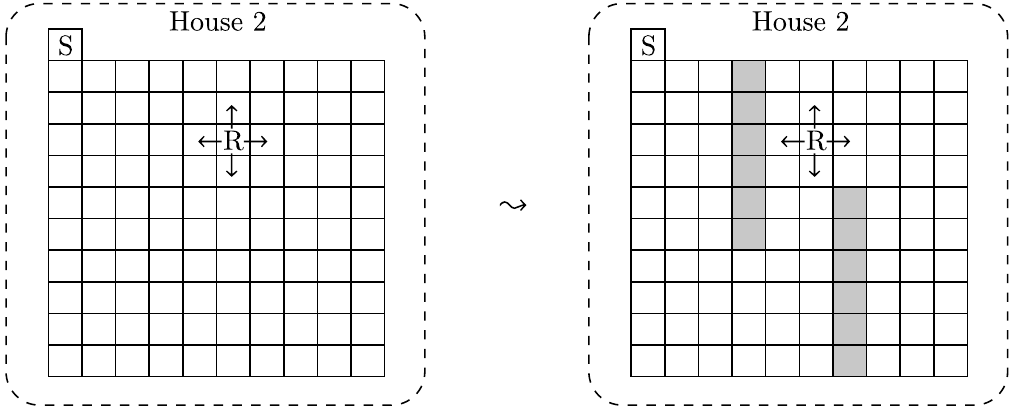}
  \caption{Change in House 2 in the HiMM given by Fig. \ref{fig:robot_example_detailed_overview}. Grey areas depict blocked locations.} 
  \label{fig:robot_example_change_overview}
\end{figure}

\subsection{Related work}
There is an extensive amount of work related to mitigating the complexity of decision-making in large-scale, high-dimensional and complex systems in a broader context, with notable examples including behavioral trees \cite{ogren2022behavior,pezzato2023active,biggar2021expressiveness,biggar2022modularity}, hierarchical reinforcement learning \cite{pateria2021hierarchical,lee2023hierarchical}, hierarchical temporal logic \cite{luo2024decomposition,wei2025hierarchical} and symbolic abstraction \cite{mallik2019compositional}. The general goal is to reduce the decision space to lower the computational complexity without losing significant performance.

The use of HFSMs is one solution to lower the computational complexity of large-scale systems and have mainly been used to model reactive systems that transition according to non-chosen external inputs (e.g., hand-watch moving the minute-watch one click after 60 seconds has passed), see e.g., \cite{harel1987statecharts,schillinger2016human,millington2018artificial}, with applications including model checking \cite{alur1998model,alur2000efficient,alur2005analysis,laster1998modular}, modelling non-human agents in games \cite{millington2018artificial}, and implementing high-level controllers in robotics \cite{schillinger2016human} together with related reactive control architectures such as behaviour trees \cite{ogren2022behavior,pezzato2023active} and similar architectures \cite{biggar2021expressiveness,biggar2022modularity}. HFSMs are also widely used tool in the Unified Modelling Language \cite{booch2005unified}. Here, we instead use HFSMs to model control systems, i.e., systems where an agent can choose control inputs to alter the system. In discrete-event systems \cite{cassandras2008introduction}, hierarchical systems have been used extensively with applications such as fault diagnosis \cite{de2020hierarchical}, safety checking \cite{perk2006hierarchical,ma2006nonblocking,wang2020real}, hierarchical coordination and consistency \cite{da2002assume,zhong1990consistency}, and designing hierarchical controllers \cite{hubbard2002dynamical,wong1996hierarchical}, just to mention a few. In particular, state tree structures have been used as an alternative HFSM formalisation to compute safe executions  \cite{ma2006nonblocking,wang2020real}. We consider a different formalism and focus instead on an optimal planning setup minimising~cost.


There is a great variety of optimal planning algorithms for graphs \cite{bast2016route} such as Dijkstra's algorithm \cite{Dijkstra1959,DijkstraFibonacci}. Several algorithms use a preprocessing step to obtain optimal plans faster at run-time \cite{bast2016route}. In particular, hierarchical approaches \cite{geisberger2012exact,dibbelt2016customizable,mohring2007partitioning,10.1145/1498698.1564502,delling2017customizable,storandt2022customizable,wan2025parallel} reminiscent our approach the most. However, these methods typically rely on heuristics to speed up the query time, which in the worst-case perform no better than Dijkstra's algorithm. Moreover, to use them, one needs to flatten the HFSM to an equivalent FSM, a process which hides the hierarchical structure making it less suitable for, for instance, changes in the HFSM. The flattening process could also be exponential in time (with respect to the depth of the HFSM) for HFSMs that compactly represent identical machines (as in this paper), and could thus be intractable. The work \cite{timo2014reachability} considers a variation of an HFSM formalism seeking a plan that minimises the length between two configurations. We consider the more general non-negative costs, where the minimial length objective can be recovered by setting all costs equal to one. 

The work \cite{biggar2021modular} considers how to decompose an FSM without costs into an equivalent HFSM, using analogies with the modular decomposition for graphs. No optimal planning is considered. In this work, we instead consider HFSM with costs and how to optimally plan in such~systems.

Finally, closest to this work is our recent paper \cite{stefansson2025faster}, which also exploits hierarchical structure for shortest-path computation. In contrast, \cite{stefansson2025faster} considers single-source shortest paths on weighted graphs and computes a hierarchy of nested graphs for preprocessing (the acyclic-connected tree), whereas this paper assumes a given hierarchy and address the more general all-pairs setting on HiMMs with modifications and identical machines.


\subsection{Outline}
The remaining of the paper is structured as follows. Section~\ref{problem_formulation} formalises HiMMs, their modifications, how identical machines can be compactly represented, and formulates the problem statement. Section~\ref{efficient_hier_planning} proposes the planning algorithm, Section~\ref{reconfig_hier_planning} extends the algorithm to be reconfigurable, and Section~\ref{duplicate_hier_planning} extends the algorithm to efficiently handle identical machines. Section~\ref{numerical_evaluations} presents numerical evaluations. Finally, Section~\ref{hier:conclusion} concludes the paper. \if\longversion0 An extended version of this paper can be found at \cite{stefansson2024modular} that contains all the proofs in Appendix.\else
All proofs are found in Appendix.
\fi

\subsection{Notation}
We let $f: A \rightharpoonup B$ denote a partial function between sets $A$ and $B$, where we use the short-hand notation $f(a) = \emptyset$ to mean that $f(a)$ is not defined at $a \in A$. For a labelled directed graph $G$, we use $u \xrightarrow{a} v$ to denote an arc from $u$ to $v$ labelled with $a$, and $(u \xrightarrow{a} v) \in G$ to mean that the arc $u \xrightarrow{a} v$ exists in $G$. Finally, $|A|$ denotes the cardinality of a set $A$.

\section{Problem Formulation}\label{problem_formulation}
In this section, we introduce HiMMs, present changes given as modifications, describe how identical MMs can be compactly represented, and formalise the problem statement.

\subsection{Hierarchical Mealy Machines} 
We formalise HiMMs, extending the setup in \cite{biggar2021modular} to machines with transition costs. To this end, we first define an~MM.

\begin{definition}[Mealy Machine]\label{def:mealy_machine}
An MM is a tuple $M = (Q,\Sigma,\delta,\gamma,s)$ where $Q$ is the finite set of states, $\Sigma$ is the finite set of inputs, $\delta: Q \times \Sigma \rightharpoonup Q$ is the transition function, $\gamma: Q \times \Sigma \rightharpoonup \mathbb{R}^+ = [0,\infty)$ is the cost function, and $s \in Q$ is the start state. We sometimes use subscript to stress that e.g., $\delta_M$ is the transition function of the MM $M$.
\end{definition}

The intuition is as follows. An MM $M$ starts in the start state $s \in Q$. Given current state $q \in Q$ and input $a \in \Sigma$ the machine transits to $\delta(q,a) \in Q$ and receives cost $\gamma(q,a) \in \mathbb{R}^+$ if $\delta(s,a) \neq \emptyset$ (i.e., the input $a$ is supported at $q$), otherwise $M$ stops. Repeating this procedure, we obtain a trajectory:

\begin{definition}
A sequence $(q_i,a_i)_{i=1}^N \in (Q \times \Sigma)^{N}$ is a trajectory of an MM $M = (Q,\Sigma,\delta,\gamma,s)$ if $q_{i+1} = \delta(q_i,a_i) \in Q$ for $i \in \{1,\dots,N-1\}$.
\end{definition}

In our case, we assume that an agent (e.g., a robot) can \emph{choose} the inputs. We therefore also define a plan as a sequence of inputs $(a_i)_{i=1}^N$ and the induced trajectory $(q_i,a_i)_{i=1}^N$ from state $q_1 \in Q$ to be the resulting trajectory by sequentially applying $(a_i)_{i=1}^N$ starting at $q_1$ (provided the trajectory exists, that is, the machine does not stop due to unsupported input). Next, we define an~HiMM, motivated first with a simple example followed by the formal definition.

\begin{figure}
	\centering
  \includegraphics[width=0.20\textwidth]{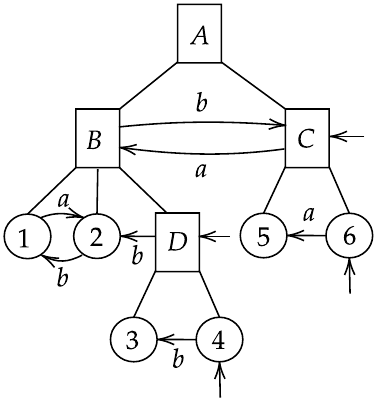} 
  \caption{Example of an HiMM depicted in a tree-form.} 
  \label{fig:Himm_transition_example_v2}
\end{figure}


\begin{example} \label{ex_intuition_before_HiMM_def}
    Consider MMs $A$, $B$, $C$ and $D$ in Fig. \ref{fig:Himm_transition_example_v2}, nested as indicated by lines in the tree. For example, state $4$ belongs to $D$, which in turn is nested as a state of $A$ (cf. Fig. \ref{fig:robot_example_detailed_overview}). Transitions in each MM are given by labeled arrows, while the start state is depicted with a small arrow. For brevity, we do not depict costs. The six leaves $1$ to $6$ are the states of the hierarchy. We illustrate how a transition is defined in this hierarchy by applying input $a$ at state $5$. Here, $C$ does not have an $a$-transition from $5$. To handle this, we instead move up in the hierarchy until we find an MM that supports $a$ and move accordingly, in this case, from $C$ to $B$. Then we go downwards in the hierarchy recursively following the start states, from $D$ and then $4$. Thus, applying $a$ at $5$ transits the nested machine to $4$. Moreover, the cost of this transition is where the $a$-transition occurs, i.e., from $C$ to $B$ in the MM $A$. The intuition of this design is to ensure \emph{modularity}. That is, from $A$'s perspective, what happens inside $C$ is invisible (cf. Fig. \ref{fig:robot_example_detailed_overview} with the robot moving inside a house). Instead, $A$ only sees that $C$ took an $a$-exit and charges the corresponding cost (cf. Fig. \ref{fig:robot_example_detailed_overview} with the robot moving from one house to the next). For this abstraction to be consistent with optimal planning, all $a$-exits out of $C$ must have the same cost at the $A$-level (otherwise $A$ would need to distinguish between trajectories inside $C$). This becomes the formal condition that makes each MM a \emph{module} in the hierarchy, proved in Section \ref{HiMM_are_modular}.
    
\end{example}

\begin{remark}[What the input set $\Sigma$ represents]
We stress that although all MMs in Example \ref{ex_intuition_before_HiMM_def} share the same input set $\Sigma = \{a,b\}$, the \emph{same input label may represent different operations in different MMs}, e.g., action $a$ could be a local action inside $C$ and a different action at the $A$-level. Thus, ascending the hierarchy occurs only when action $a$ is genuinely unsupported at the current state and we want it to ascend up in the hierarchy, not by default. This is true in general for HiMMs. For example, in the robot application in Section~\ref{illustrative_example}, the robot can only move to another house when it is at the entrance $S$; for all other states inside the house, the robot instead moves to a neighbouring location or remain idle if it cannot move, i.e., blocked by a wall (the idle action is the general way to treat actions that are not applicable at a state but also undesired to ascend up in the hierarchy). Thus, the interesting quantity is really the cardinality $|\Sigma|$ of $\Sigma$ since it encodes how many actions we can have in each MM.
\end{remark}

With the ideas from Example \ref{ex_intuition_before_HiMM_def}, we now formally define a HiMM using three objects: a start function (descent through start states), a hierarchical transition function $\psi$ (ascend-then-descend), and a hierarchical cost function $\chi$ (the cost charged in the process).
 
 
\begin{definition}[Hierarchical Mealy Machine]\label{HiMM_def}
An HiMM $Z = (X,T)$ consists of a set of MMs $X$ with common input set $\Sigma$, and a tree $T$ with nodes being MMs from $X$ (specifying how the machines in $X$ are nested). More precisely, a node $M \in X$ of $T$ has $|Q_M|$ outgoing arcs on the form $M \xrightarrow{q} N_q$ where either $N_q \in X$, specifying that state $q \in Q_M$ corresponds to the MM $N_q$ one layer down in the hierarchy, or $N_q = \emptyset$, specifying that $q$ has no further refinement and is simply just a state. We call $Q_Z := \cup_{M \in X} Q_M$ the augmented states of $Z$ and $S_Z := Q_Z \cap \{q: N_q = \emptyset \}$ the states of $Z$.\footnote{We implicitly assume that all state sets $Q_M$ are distinct.} We have corresponding notions for start state, transition function and cost function:
\begin{enumerate}[(i)]
\item Start state ${\mathrm{start}: X \rightarrow S_Z}$ is the function recursively defined as ($s_M$ is the start state of $M$)
\begin{equation*}
\mathrm{start}(M) =
\begin{cases}
\mathrm{start}(N), &  \textrm{$(M \xrightarrow{s_M} N) \in T$, $N \in X$} \\
s_M, & \textrm{otherwise.}
\end{cases}
\end{equation*}
\item Let $q \in Q_M$ with $M \in X$ and $v=\delta_{M}(q,a)$. The hierarchical transition function $\psi: Q_Z \times \Sigma \rightharpoonup S_Z$ is recursively defined as
\begin{equation*}
\psi(q,a) =
\begin{cases} 
\mathrm{start}(Y), & v \neq \emptyset, (M \xrightarrow{v} Y) \in T, Y \in X \\ 
v, & v \neq \emptyset, \mathrm{otherwise} \\
\psi(w,a), & v = \emptyset, (W \xrightarrow{w} M) \in T, W \in X \\ 
\emptyset, & v = \emptyset, \mathrm{otherwise,}
\end{cases}
\end{equation*}
and the hierarchical cost function $\chi: Q_Z \times \Sigma \rightharpoonup \mathbb{R}^+$ is recursively defined as
\begin{equation*}
\chi(q,a) =
\begin{cases} 
\gamma_{M}(q,a), \hspace{-2pt} & v \neq \emptyset \\
\chi(w,a), & v = \emptyset, (W \xrightarrow{w} M) \in T, W \in X \\ 
\emptyset, & v = \emptyset, \mathrm{otherwise.}
\end{cases}
\end{equation*} 
\end{enumerate}
\end{definition}
An HiMM $Z = (X,T)$ works similarly to an MM. The HiMM starts in $\mathrm{start}(M_0)$ where $M_0$ is the root MM of $T$. Given current state $q \in S_Z$ and input $a \in \Sigma$, we transit to $\psi(q,a) \in S_Z$ with cost $\chi(q,a) \in \mathbb{R}^+$ if $\psi(q,a) \neq \emptyset$, otherwise $Z$ stops. A trajectory, plan and induced trajectory are defined analogously to the MM case.


\subsubsection{HiMMs are Modular}\label{HiMM_are_modular}
We now justify the definition of an HiMM $Z=(X,T)$, by proving that $Z$ is modular in a mathematical sense, extending the modular theory for FSMs from \cite{biggar2021modular}. The idea is that a subset of states $S$ in an MM $M$ is a \emph{module} if it can be replaced with an abstract state without loosing any information of $M$. That is, if we contract $S$ to an point and later replace this point with $S$, we get back $M$ (cf. Fig \ref{fig:robot_example_detailed_overview} with $S$ being one of the houses).

To put this idea into a mathematical context, we need three operations: Contraction $M/S$ contracts a subset of states $S$ to a single state (also denoted $S$); Restriction $M[S]$ keeps only the states in $S$; Expansion $M \cdot_v S$ replaces state $v$ in $M$ with $S$ \if\longversion0{(see \cite{stefansson2024modular} for formal definitions).}\else{(see Appendix for formal definitions).}\fi With these operations, we can define a module:
\begin{definition}[MM Modules]\label{def:module_text_new}
A non-empty subset $S \subseteq Q_M$ of states of an MM $M$ is a \emph{module} if $M/S \cdot_S M[S] = M$.
\end{definition}
Similar to the FSM case \cite{biggar2021modular}, there is an equivalent  definition that is easier to verify, given by the following result:
\begin{proposition}\label{prop:module_text}
A non-empty subset of states $S \subseteq Q_M$ of an MM $M$ is a module if and only if it has at most one entrance\footnote{An entrance $e$ of $S$ is such that $e \in S$ and there is a $v \in Q_M \backslash S$ that can transit to $e$ via some input.} and for each $a \in \Sigma_M$, if $S$ has an $a$-exit\footnote{An $a$-exit to $S$ is a state $u \notin S$ such that some $v \in S$ transits to $u$ via some input.} then (i) that $a$-exit is unique (ii) every state in $S$ has an $a$-transition and (iii) all $a$-transitions out of $S$ have equal costs.
\end{proposition}
\begin{proof}[Sketch of proof]
Conditions (i)-(iii) are exactly what is needed for the operations contraction, restriction and expansion with respect to $S$ to be well-defined with equality $M/S \cdot_S M[S] = M$. Conversely, any violation of the conditions (i)-(iii) results in some of the operations to be ill-defined.\if\longversion0{ See \cite{stefansson2024modular} for a full proof.}\else{ See Appendix for a full proof.}\fi
\end{proof}
The theory of modules can easily be extended to HiMMs by saying that a subset of augmented states $S$ of an HiMM $Z$ is a \emph{module} if the states they contain is a module with respect to the equivalent flat MM $Z^F$, where $Z^F$ is formed by collapsing the hierarchy of $Z$ to one big MM with identical states and transitions.\footnote{Formally, $Z^F$ of an HiMM $Z$ is the MM given by $Z^F = (Q,\Sigma,\delta,\gamma,s) := (S_Z,\Sigma,\psi,\chi,\mathrm{start}(M_0))$ (with notation as in Definition~\ref{HiMM_def}).} With this definition, we can show that an HiMM is naturally modular:
\begin{proposition}\label{prop:himms_are_modular}
An HiMM $Z=(X,T)$ is modular in the sense that each MM $M \in X$ is a module of $Z$.
\end{proposition}
\begin{proof}[Sketch of proof]
We prove this using Proposition~\ref{prop:module_text}. For any MM $M$ in $Z$, the states nested within $M$ have a unique entrance ($\mathrm{start}(M)$) by construction. Moreover, conditions (i) to (iii) follow because the hierarchical transition function $\psi$ and cost function $\chi$ resolve unsupported inputs by ascending to the parent. Thus, all $a$-exits from $M$ are forced to share both their destination (whichever ancestor first supports $a$) and their cost (the cost charged at that ancestor).\if\longversion0{ See \cite{stefansson2024modular} for a full proof.}\else{ See Appendix for a full proof.}\fi
\end{proof}
Proposition \ref{prop:himms_are_modular} is the theoretical backbone of our optimal planning algorithm since it allows us to compute shortest paths in a modular fashion, by solving local shortest path problem in each MM in the HiMM. The remaining paper do no use these modularity results explicitly, but it is nonetheless the underlying reason why the theory works, and thus important to~highlight.\footnote{The analysis up to this point is not restricted to the codomain of the (hierarchical) cost function. In fact, the cost function in Definition \ref{def:mealy_machine} and the hierarchical cost function in Definition \ref{HiMM_def} could have any set $\Lambda$ as codomain. We made the restriction since the remaining paper needs $\Lambda$ to be $\mathbb{R}^+$.}


\subsubsection{Optimal Planning} In this paper, we consider optimal planning in an HiMM $Z$. To this end, the cost of a trajectory $z = (q_i,a_i)_{i=1}^N$ is the cumulative cost $C(z):=\sum_{i=1}^N \chi(q_i,a_i)$ if all $\chi(q_i,a_i) \neq \emptyset$ for all $i$ and $C(z) = \infty$ otherwise.\footnote{Note that $\chi(q_i,a_i) \neq \emptyset$ for all $i$ is equivalent to $\chi(q_N,a_N) \neq \emptyset$ since $(q_i,a_i)_{i=1}^N$ is a trajectory. We set $C(z) = \infty$ when this is not fulfilled since such a trajectory makes the HiMM stop (due to unsupported input), which is~undesirable.} In this setting, for a given initial state $s_{\mathrm{init}} \in S_Z$ and goal state $s_{\mathrm{goal}} \in S_Z$, we want to find a plan $u = (a_i)_{i=1}^N$ that takes us from $s_{\mathrm{init}}$ to $s_{\mathrm{goal}}$ with minimal cumulative cost. That is, if $U(s_{\mathrm{init}},s_{\mathrm{goal}})$ denotes the set of plans $u = (a_i)_{i=1}^N$ such that the induced trajectory $z = (q_i,a_i)_{i=1}^N$ with $q_1 = s_{\mathrm{init}}$ satisfies $\psi(q_N,a_N) = s_{\mathrm{goal}}$ (namely, we reach $s_{\mathrm{goal}}$), then we want to find a plan solving
\begin{equation}\label{planning_objective}
\min_{u \in U(s_{\mathrm{init}},s_{\mathrm{goal}})} C(z),
\end{equation}
where $z$ is the induced trajectory of $u$. We call a plan $u$ solving \eqref{planning_objective} for an optimal plan to the planning objective $(Z,s_{\mathrm{init}},s_{\mathrm{goal}})$ and the induced trajectory $z$ for an optimal trajectory. In Section \ref{efficient_hier_planning}, we present our planning algorithm that computes optimal plans efficiently.

\subsection{Modifications}\label{modifications}

\begin{figure}
	\centering
  \includegraphics[width=0.45\textwidth]{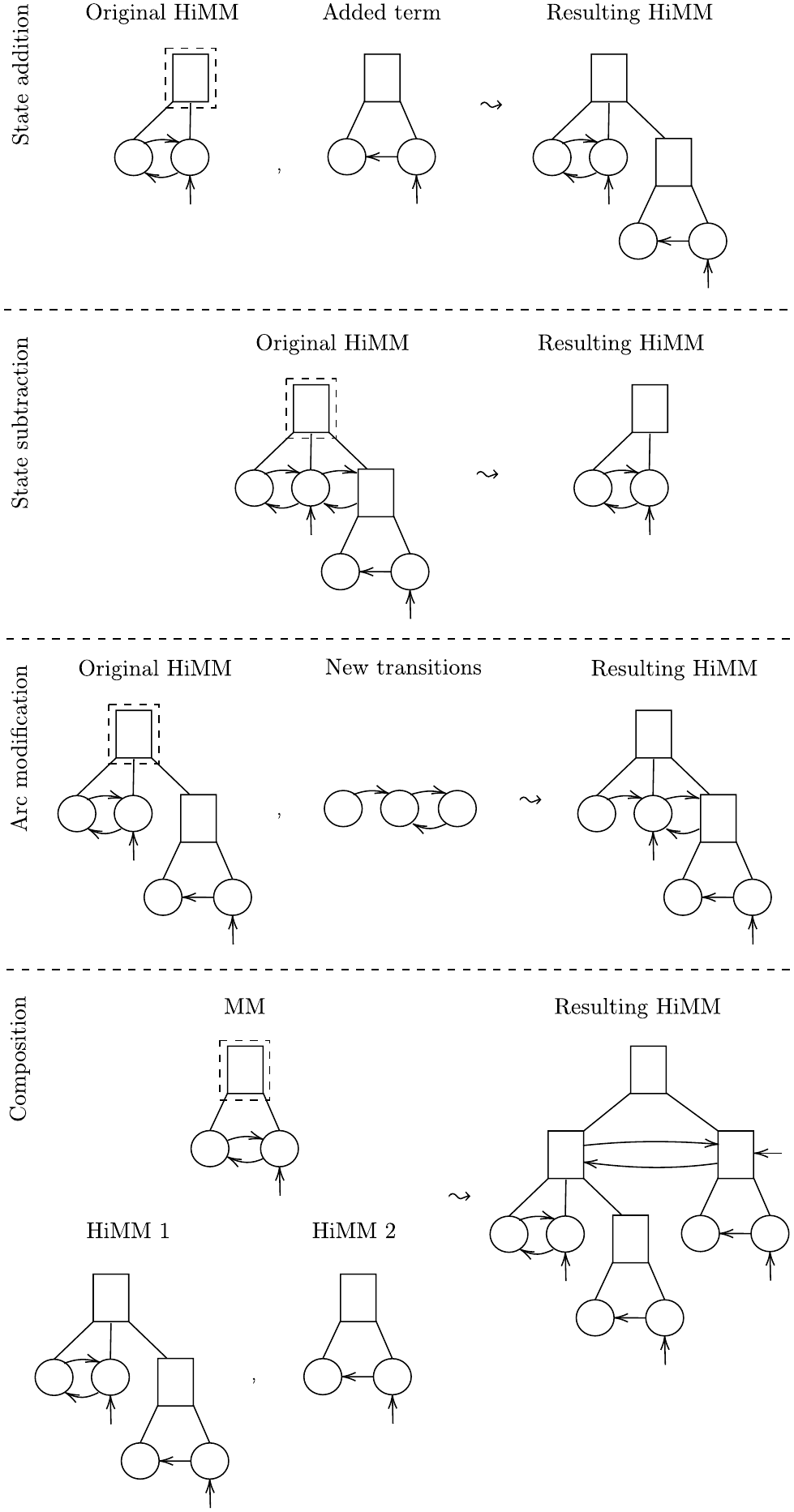}
  \caption{The four modifications of an HiMM depicted. The dashed rectangle specifies the MM $M$ subject to change.} 
  \label{fig:all_operations}
\end{figure}

In this section, we introduce changes of an HiMM $Z$, formalised by modifications. The four types of modifications we consider are state addition, state subtraction, arc modification and composition, depicted in Fig.~\ref{fig:all_operations}. Intuitively, state addition adds a state or an HiMM $Z_{\mathrm{add}}$ to an MM $M$ in a given HiMM $Z$, where the root of $Z_{\mathrm{add}}$ becomes $M$'s new state; state subtraction is the reverse operation removing a state from an MM; arc modification changes the transition function of an MM in the HiMM; and, composition attaches several HiMMs to a new MM as states. The formal definitions are given by Definitions \ref{def:state_addition} to \ref{def:composition}.

\begin{definition}[State Addition]\label{def:state_addition}
Let HiMM $Z = (X,T)$ be given and MM $M \in X$. We say that an HiMM $Z_{\mathrm{add}} = (X_{\mathrm{add}},T_{\mathrm{add}})$ is added to $M$ in $Z$ forming a new HiMM $Z'$ with tree $T'$ equal to $T$ except that an arc $M \xrightarrow{q} N$ has been added in $T'$ from $M$ to the root MM $N$ of $T_{\mathrm{add}}$ (connecting $Z_{\mathrm{add}}$ to $M$), where $q$ is a new arbitrary (distinct) state of $M$.\footnote{Hence, the new state set of $M$ is $Q_M \cup \{q \}$, where no transitions in $M$ are going from or to $q$, see Fig. \ref{fig:all_operations} for an example.} Adding a state to $M$ is identical except that the added arc is $M \xrightarrow{q} \emptyset$ instead.
\end{definition}

\begin{definition}[State Subtraction]\label{def:state_subtraction}
Let HiMM $Z = (X,T)$ be given and MM $M \in X$. We say that we subtract a state $q \in Q_M$ giving the new HiMM $Z'$ having tree $T'$ equal to $T$ except that the $q$-arc from $M$ is removed.\footnote{Hence, $q$ is no longer a state of $M$ and all the previous transitions to and from $q$ in $M$ are removed, see Fig. \ref{fig:all_operations} for an example. Furthermore, for simplicity, we prohibit the start state of $M$ to be removed. This can be achieved instead by applying an arc modification first, changing the start state, and then remove the previous start state using a state subtraction.}
\end{definition}

\begin{definition}[Arc Modification]\label{def:arc_modification}
Let HiMM $Z = (X,T)$ be given and MM $M \in X$. Let $\delta': Q_M \times \Sigma_M \rightharpoonup Q_M$, $\gamma': Q_M \times \Sigma_M \rightarrow \mathbb{R}^+$ and $s' \in Q_M$ be a new transition function, cost function and start state and form $M' = (Q_M, \Sigma_M, \delta',\gamma',s')$. Replace $M$ with $M'$ in $Z$ to form a new HiMM $Z'$. This modification is called an arc modification.
\end{definition}

\begin{definition}[Composition]\label{def:composition}
Let $M$ be an MM with states $Q = \{q_1,\dots,q_{|Q|}\}$ and $Z_{\mathrm{seq}} = \{ Z_1,\dots,Z_{n} \} $ be a set of HiMMs such that $n \leq |Q|$. The composition of $Z_{\mathrm{seq}}$ with respect to $M$ is the HiMM $Z'$ with tree $T'$ having root $M$ and arcs $M \xrightarrow{q_i} R_i$ for $i \leq n$ and $M \xrightarrow{q_i} \emptyset$ for $i>n$, where $R_i$ is the root MM of the tree of $Z_i$.\footnote{Intuitively, we connect all subtrees given by $Z_1,\dots,Z_{n}$ to the $n$ first states of $M$, while the remaining states of $M$ are just states without any refinements.}
\end{definition}
These four modifications form the atomic changes of an HiMM, where a sequence of modifications can change any HiMM $Z = (X,T)$ to any other HiMM $Z' = (X',T')$. This can be done for example by adding all subtrees of the root $M'$ of $T'$ to the root $M$ of $T$ (using state addition), use an arc modification to change $M$ so that it is identical to $M'$ on the added subtrees (including changing the start state), and then remove the old subtrees from $M$. Then $Z$ has been changed to $Z'$. However, the primary intention of modifications is not to completely change one HiMM into another, but instead to keep track and account for more minor changes, e.g., removing just some states. Such minor changes can then be integrated efficiently into the planning algorithm to avoid recomputing everything from scratch whenever an HiMM has changed, done in Section \ref{reconfig_hier_planning}. Finally, practical examples of modifications are given in Section \ref{numerical_eval_robot_warehouse} for the robot application from Section \ref{illustrative_example}.

\subsection{Identical Machines}\label{duplicates}

\begin{figure}
\centering
\includegraphics[width=0.25\linewidth]{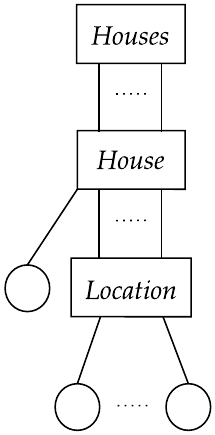}
\caption{The HiMM with identical machines for the robot application from Section \ref{illustrative_example}. The MM Houses consists of several identical houses, depicted by lines from the MM Houses to the MM House (where the number of lines specify the number of identical houses). The MM House consists in turn of one start state and several identical MMs given by Location, where Location is an MM with states having no further~refinement.}
\label{fig:duplicates_example}
\end{figure}

In this section, we extend the HiMM formalism to the case when an HiMM $Z=(X,T)$ can compactly represent identical subtrees in $T$, by grouping them together. To motivate this, consider again the robot application from Section \ref{illustrative_example}, and assume all houses are identical. In this case, the HiMM can be compactly represented as in Fig.~\ref{fig:duplicates_example}. Here, we have lines just as in Fig.~\ref{fig:Himm_transition_example_v2} to represent the hierarchical relationship between MMs and states (we have omitted transition arcs to keep the illustration clean). However, several lines \emph{can now go to the same MM}. For example, all lines from the MM Houses goes to the same MM House. This captures that each house is in fact identical sharing the same subtree with root MM House. Using such identifications, we can more compactly represent an HiMM and may yield (in the best case) a logarithmic reduction in the number of MMs one needs for saving the HiMM. For brevity, we call such a HiMM a HiMM with identical machines.

Mathematically, the key difference with an HiMM $Z=(X,T)$ with identical machines is that $T$ is no longer a tree but instead a directed acyclic graph (DAG) with a unique root (e.g., Fig. \ref{fig:duplicates_example}). The specification of the states are also subtly different. More precisely, when $T$ is a tree, a state $q \in S_Z$ is a leaf of $T$. This is sufficient since there is a unique path from the root of $T$ to the leaf. However, this representation of a state breaks down for the identical machines. More precisely, when $T$ is a DAG with a unique root, there might be several paths from the root to a leaf of $T$, which makes the leaf insufficient for specifying the state of the HiMM. For example, in the robot application, saying that the robot is at the entrance of a house does not specify which house. Instead, one needs to specify the whole path from the root of $T$ to the leaf. Such a path is a state for HiMM with identical machines. For example, state $s =(1,2,3)$ in the robot application specifies that the robot is in House 1, Location 2, and state 3 (at Location~2). We provide a formal~definition:


\begin{definition}\label{HiMM_duplicates_def}[HiMM with identical machines]
An HiMM with identical machines is a tuple $Z = (X,T)$ just as in Definition \ref{HiMM_def} except that multiple arcs in $T$ can now go to the same node, i.e., $T$ is not a tree anymore but instead a DAG with a unique root. Furthermore, a state of $Z$ is a directed path of arcs $s = (s_1,\dots,s_d)$ from the root of $T$ to a leaf $s_d$, and an augmented state is any subpath $(s_1,\dots,s_i)$ of a state $s = (s_1,\dots,s_d)$ ($i \leq d$). Finally, the start function, hierarchical transition and cost function are defined as in Definition \ref{HiMM_def} with straightforward modifications to account for the modified state form, \if\longversion0 see \cite{stefansson2024modular} for derivations and~details.\else
see Appendix for derivations and~details.
\fi
\end{definition}


\subsubsection{Modifications}
Modifications can be readily extended to a HiMM with identical machines. More precisely, let an HiMM $Z=(X,T)$ with identical machines be given. For state addition, the definition is identical to Definition \ref{def:state_addition}. except that one can also add an additional arc $M \xrightarrow{q} N$ from an MM $M \in X$ to an \emph{already existing child} MM $N$ of $M$ ($q$ is still an arbitrary but distinct state), intuitively corresponding to the copies of the subtree with root $N$. For state subtraction, arc modification and composition, the definitions are identical to the ones in Section \ref{modifications}.

\subsection{Problem Statement}
In this paper, we seek a planning algorithm for HiMMs that:
\begin{enumerate}[(i)]
\item Computes optimal plans between any two states in the HiMM with low time complexity.
\item Handles any of the four types of modifications with low time complexity (assuming the algorithm knows what has been changed).
\item Exploit identical machines to lower the time complexity (assuming the algorithm knows which MMs are~identical).
\end{enumerate}
We address (i), (ii) and (iii) in Sections \ref{efficient_hier_planning}, \ref{reconfig_hier_planning} and \ref{duplicate_hier_planning} respectively, and demonstrate the resulting planning algorithm, given by Algorithm \ref{alg:hierarchical_planning_new}, with numerical evaluations in Section~\ref{numerical_evaluations}.

\begin{algorithm}[t]
\caption{Hierarchical Planning}\label{alg:hierarchical_planning_new}
\begin{algorithmic}[1]
\Require HiMM$ \; Z$, \textcolor{BrickRed}{modifications $m_{seq}$} and states $s_{\mathrm{init}}, s_{\mathrm{goal}}$. 
\Ensure Optimal plan $u$ to $(Z,s_{\mathrm{init}}, s_{\mathrm{goal}})$.
\State \textbf{Optimal Exit Computer:}
\For{\textcolor{BrickRed}{$m$ in $m_{\mathrm{seq}}$}}
\State \textcolor{BrickRed}{$\mathrm{Modify}(Z,m)$} \Comment{\textcolor{BrickRed}{This line is executed in System}} \label{alg:hierarchical_planning_new:line3}
\State \textcolor{BrickRed}{$\mathrm{Mark}(Z,m)$} \label{alg:hierarchical_planning_new:line4}
\EndFor
\State $(c_a^M,z_a^M)_{a \in \Sigma}^{M \in X} \gets \mathrm{Compute\_optimal\_exits}(Z)$ \label{alg:hierarchical_planning_new:line6}
\State \textbf{Online Planner:}
\State $\bar{Z} \gets \mathrm{Reduce}(Z,s_{\mathrm{init}},s_{\mathrm{goal}},(c_a^M,z_a^M)_{a \in \Sigma}^{M \in X})$ \label{alg:hierarchical_planning_new:line8}
\State $z \gets \mathrm{Solve}(\bar{Z},s_{\mathrm{init}},s_{\mathrm{goal}},(c_a^M,z_a^M)_{a \in \Sigma}^{M \in X})$ \label{alg:hierarchical_planning_new:line9}
\State $u \gets \mathrm{Expand}(z,(z_a^M)_{a \in \Sigma}^{M \in X}, Z)$ \label{alg:hierarchical_planning_new:line10}
\end{algorithmic}
\end{algorithm}

\section{Efficient Hierarchical Planning}\label{efficient_hier_planning}
In this section, we present a planning algorithm computing optimal plans for a given HiMM $Z =(X,T)$,\footnote{Here, we assume an HiMM as in Definition \ref{HiMM_def}. The case for an HiMM with identical machines as in Definition \ref{HiMM_duplicates_def} is considered in Section \ref{duplicate_hier_planning}.} summarised by Algorithm \ref{alg:hierarchical_planning_new} (ignore pseudo-code in red for now). The algorithm consists of two steps. In the first step, an Optimal Exit Computer computes optimal exit costs for each MM $M \in X$ (line \ref{alg:hierarchical_planning_new:line6}). This is a preprocessing step done only once for a fixed HiMM $Z$. In the second step, an Optimal Planner computes an optimal plan $u$ for any initial state $s_{\mathrm{init}}$ and goal state $s_{\mathrm{goal}}$ (line \ref{alg:hierarchical_planning_new:line8}--\ref{alg:hierarchical_planning_new:line10}). This is done by first reducing the HiMM $Z$ to an equivalent reduced HiMM  $\bar{Z}$, removing subtrees not containing $s_{\mathrm{init}}$ or $s_{\mathrm{goal}}$ (line \ref{alg:hierarchical_planning_new:line8}), then computing an optimal trajectory for this equivalent reduced HiMM $\bar{Z}$ (line \ref{alg:hierarchical_planning_new:line9}), and finally expanding this optimal trajectory to an optimal plan for the original HiMM $Z$ (line \ref{alg:hierarchical_planning_new:line10}). We next provide details of the Optimal Exit Computer and the Optimal Planner, starting with the former.

\begin{figure}[t]
\centering
\includegraphics[width=0.9\linewidth]{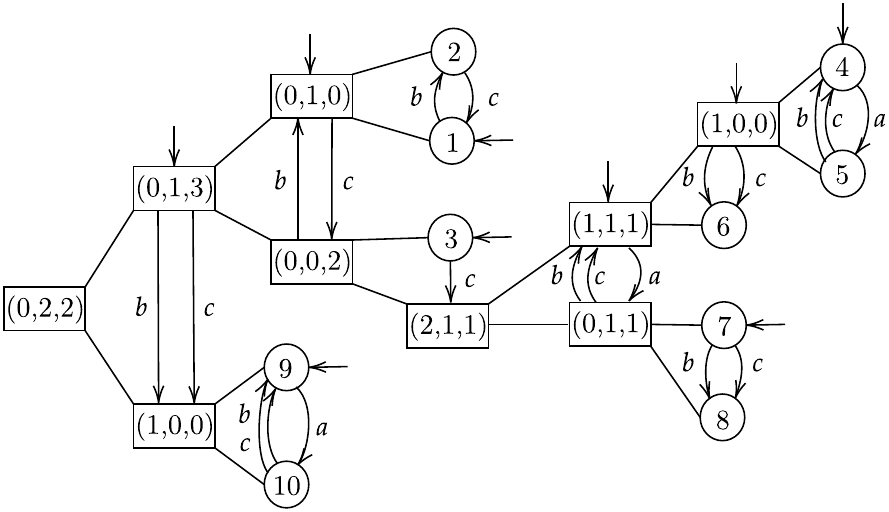}
\caption{HiMM $Z=(X,T)$ with optimal exit costs given as tuples $(c^M_a, c^M_b, c^M_c)$ over each MM $M$.}
\label{fig:hfsm_example_tree_cut}
\end{figure}

\subsection{Optimal Exit Computer}\label{optimal_exit_computer}
For a given HiMM $Z=(X,T)$, the Optimal Exit Computer computes optimal exit costs for each MM $M \in X$, given as a tuple $(c^M_a)_{a \in \Sigma}$, where $c^M_a$ is the minimal cost for exiting the subtree having $M$ as root with input $a$.


\begin{definition}[Optimal Exit Costs]\label{def:optimal_exit_costs}
Let HiMM $Z=(X,T)$ and MM $M \in X$ be given and $S_M \subseteq S_Z$ be the subset of states nested within $M$.\footnote{The set of states contained in the subtree of $Z$ having $M$ as root.} A trajectory $z =(q_i,a_i)_{i=1}^N$ such that $q_1 = \mathrm{start}(M)$, $q_i \in S_M$ for all $i$, $\psi(q_N,a_N) \notin S_M$ and $a_N = a$ is called an $a$-exit trajectory, with $a$-exit cost $\sum_{i=1}^{N-1} \chi(q_i,a_i)$.\footnote{We exclude the final cost $\chi(q_N,a_N)$ since this cost is happening outside the subtree with root $M$.} An optimal $a$-exit trajectory is one with minimal $a$-exit cost, where this minimal cost is called the optimal exit cost $c^M_a$. 
\end{definition}

The Optimal Exit Computer computes $(c^M_a)_{a \in \Sigma}$ for each MM $M \in X$ recursively over the tree $T$ by first obtaining optimal exit costs for the MMs furthest out in $T$ and then  working itself upwards in the hierarchy. Here, for each MM $M \in X$, the tuple $(c^M_a)_{a \in \Sigma}$ is computed using Dijkstra's algorithm on an augmented MM $\hat{M}$ to account for the MMs nested within $M$.  More precisely, for a state $q \in Q_M$ in $M$, define $c_a^{q} = c_a^{N_q}$ if $q$ corresponds to an MM $N_q$ (further down in the hierarchy) and $c_a^{q} = 0$ otherwise.\footnote{This can be defined since the optimal exit costs $(c^{N_q}_a)_{a \in \Sigma}$ for each $N_q$ (further down in the hierarchy) have already been computed due to the construction of the recursion.} Intuitively, $c_a^{q}$ is the optimal cost for starting in the subtree corresponding to $q$ and exiting it with $a$. The augmented MM $\hat{M}$ is then
\begin{equation*}
\hat{M} := (Q_M \cup \{E_a\}_{a \in \Sigma},\Sigma_M, \hat{\delta},\hat{\gamma},s_M).
\end{equation*}
Here, the transition function $\hat{\delta}(q,a)$ equals $\delta_M(q,a)$ whenever $\delta_M(q,a)$ is defined (transition in $M$ exists), and $\hat{\delta}(q,a) = E_a$ otherwise ($\hat{\delta}(E_a,\cdot)$ is immaterial). In this way, $\hat{\delta}$ book-keeps when we exit $M$ with input $a$ given by a transition to $E_a$.
Moreover, the cost function $\hat{\gamma}(q,a)$ equals $\gamma_M(q,a) + c_a^{q}$ if $\delta_M(q,a)$ is defined (transition in $M$ exists) and otherwise $0 + c_a^{q}$ (we go to $E_a$ with zero cost on the $M$-level since we exit $M$). Here, we add $c_a^{q}$ to augment the cost of exiting the subtree corresponding to $q$ with the transition cost of applying $a$ at $q$ in $M$ (again, $\hat{\gamma}(E_a,\cdot)$ is immaterial). With this, we can search in $\hat{M}$ from $s_M$ using Dijkstra's algorithm until we have reached each $E_a$, where $c_a^M$ then equals the obtained cost to $E_a$. We also get corresponding trajectories $z_a^M$ to each $E_a$ in $M$ for free. We summarise the Optimal Exit Computer in Algorithm \ref{alg:compute_optimal_exits} with correctness given by Proposition \ref{th:compute_optimal_exits_correct} and time complexity given by Proposition \ref{th:compute_optimal_exits_time}. An example is provided by Fig. \ref{fig:hfsm_example_tree_cut}.

\begin{algorithm}[t]
\caption{Compute\_optimal\_exits}\label{alg:compute_optimal_exits}
\begin{algorithmic}[1]
\Require HiMM $Z = (X,T)$
\Ensure Computed $(c_a^M,z_a^M)_{a \in \Sigma}$ for each MM $M$ of $Z$
\State $\mathrm{Optimal\_exit}(M_0)$ \Comment{Run from root MM $M_0$ of $T$}
\State $\mathrm{Optimal\_exit}(M)$: \Comment{Recursive help function} 
\For {each state $q$ in $Q_M$ \label{alg:compute_optimal_exits:line3}}
\If{$q \in S_Z$}
\State $(c_a^q)_{a \in \Sigma} \gets 0_{|\Sigma|}$
\Else
\State Let $N_q$ be the MM corresponding to $q$.
\State $(c_a^{q},z_a^{q})_{a \in \Sigma} \gets \mathrm{Optimal\_exit}(N_q)$ \label{alg:compute_optimal_exits:line8}
\EndIf
\EndFor
\State Construct $\hat{M}$ \label{alg:compute_optimal_exits:line11}
\State $(c_a^M,z_a^M)_{a \in \Sigma}$ $\gets$ Dijkstra($s_M,\{E_a\}_{a \in \Sigma},\hat{M}$) \label{alg:compute_optimal_exits:line12}
\State return $(c_a^M,z_a^M)_{a \in \Sigma}$
\end{algorithmic}
\end{algorithm}

\begin{proposition}\label{th:compute_optimal_exits_correct}
$c_a^M \in [0,\infty]$ computed by Algorithm \ref{alg:compute_optimal_exits} equals the optimal exit cost $c_a^M$ in Definition~\ref{def:optimal_exit_costs}.
\end{proposition}

\begin{proposition}\label{th:compute_optimal_exits_time}
The time complexity of Algorithm \ref{alg:compute_optimal_exits} is $O(|X| \cdot [b_s |\Sigma|+(b_s+|\Sigma|) \log(b_s+|\Sigma|)])$, where $b_s$ is the maximum number of states in an MM of $Z$, and $|X|$ is the number of MMs in $Z$.
\end{proposition}



\subsection{Optimal Planner}\label{optimal_planner}

We continue with the Optimal Planner. For a given HiMM $Z =(X,T)$, the Optimal Planner computes  an optimal plan between any initial state $s_{\mathrm{init}}$ and goal state $s_{\mathrm{goal}}$ by truncating away subtrees of the tree $T$ that does not contain $s_{\mathrm{init}}$ or $s_{\mathrm{goal}}$ as nested states, using the optimal exit costs $(c^M_a)_{a \in \Sigma}$. Intuitively, this is possible since each subtree $H$ not containing $s_{\mathrm{init}}$ and $s_{\mathrm{goal}}$ must be exited if entered (by optimality), so we can replace $H$ with just a state (representing $H$) having transition costs augmented with the optimal exit costs $(c^M_a)_{a \in \Sigma}$. This results in an equivalent reduced HiMM $\bar{Z}$ on which we rapidly compute an optimal trajectory ${z}$ from $s_{\mathrm{init}}$ to $s_{\mathrm{goal}}$. Finally, we expand ${z}$ to obtain an optimal plan $u$ for the original HiMM $Z$. The whole procedure is illustrated by Fig. \ref{fig:online_planner_overview_cut}. We next provide details of the steps Reduce, Solve and Expand, finishing with the time complexity of the Optimal~Planner.


\begin{figure*}[t!]
\centering
\includegraphics[width=1\linewidth]{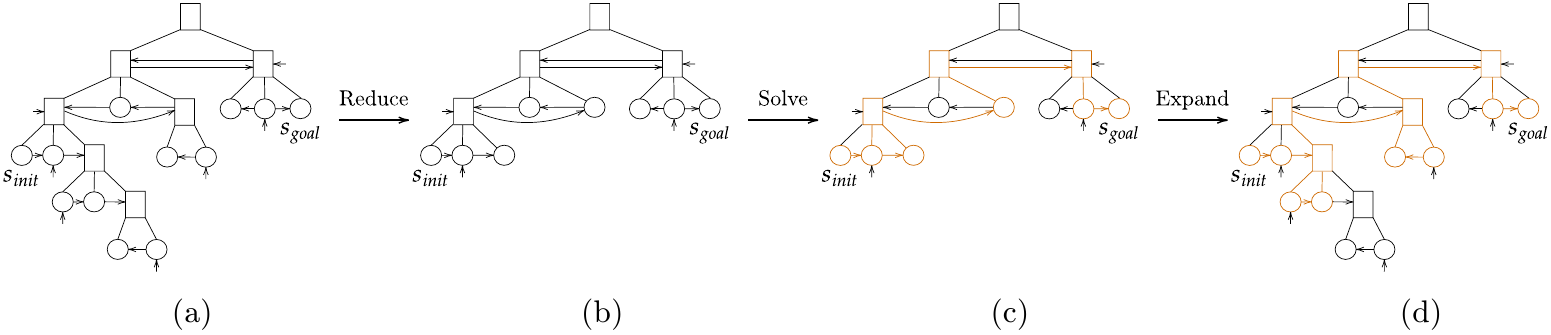}
\caption{Overview of the procedure of the Optimal Planner. (a) The original HiMM $Z$. (b) The reduced HiMM $\bar{Z}$ computed by truncating subtrees of $Z$ not containing $s_{\mathrm{init}}$ or $s_{\mathrm{goal}}$. (c) Computed optimal trajectory $z$ (marked orange) to the reduced HiMM $\bar{Z}$. (d) Optimal plan $u$ (with its optimal trajectory marked orange) to the original HiMM $Z$ computed by expanding $z$.}
\label{fig:online_planner_overview_cut}
\end{figure*}

\subsubsection{Reduce}\label{reduce_step}
We reduce the HiMM $Z =(X,T)$ for given $s_{\mathrm{init}}$ and $s_{\mathrm{goal}}$ as follows. Let $U_1$ be the MM having $s_{\mathrm{init}}$ as state, $U_2$ having $U_1$ as state and so on until $U_n$ equals the root MM of $Z$. This results in a sequence $U_1,\dots,U_n$ of nested MMs from $s_{\mathrm{init}}$ up to the root of $Z$. Similarly, we can define $D_1,\dots,D_m$ as the nested MMs from $s_{\mathrm{goal}}$ to the root. Consider now the subgraph of $T$ consisting of nodes $\M := \{ U_1,\dots,U_n \} \cup \{ D_1, \dots, D_m\}$ and their outgoing arcs representing their states. Note that this subgraph is exactly the MMs we are left with by removing all the subtrees of $T$ that do not contain $s_{\mathrm{init}}$ and $s_{\mathrm{goal}}$ as nested states. Thus, to form $\bar{Z}$, we only need to modify the MMs in $\M$ to account for the subtrees we have removed. We describe this procedure for the MM $U_i$ in $\M$, constructing a modified MM $\bar{U}_i$. More precisely, if $U_i = (Q,\Sigma,\delta,\gamma, s)$, then $\bar{U}_i = (Q,\Sigma, \delta,\bar{\gamma}, s)$ is identical to $U_i$ except with a modified cost function $\bar{\gamma}$ given by
\begin{equation}\label{eq:bar_gamma}
\bar{\gamma}(q,a) =
\begin{cases}
\gamma(q,a), & \parbox[t]{.2\textwidth}{$\delta(q,a) \neq \emptyset, \\ ({U}_i \xrightarrow{q} M) \in {T}, M \in \M$} \\
c_a^{q}+\gamma(q,a), & \delta(q,a) \neq \emptyset, \mathrm{otherwise} \\
0, & \parbox[t]{.2\textwidth}{$\delta(q,a) = \emptyset, \\ ({U}_i \xrightarrow{q} M) \in {T}, M \in \M$} \\
c_a^{q}, & \delta(q,a) = \emptyset, \mathrm{otherwise.}
\end{cases}
\end{equation}
Here, condition $({U}_i \xrightarrow{q} M) \in {T}, M \in \M$ checks if $q$ corresponds to one of the MMs of $\M$. If true, then the MM corresponding to $q$ has not been removed, and we therefore do not need to account for any truncation. Consequently, no exit cost is added in the first and the third condition. Otherwise, we need to account for the fact that taking input $a$ at $q$ actually corresponds to exiting a whole subtree in $Z$ with exit cost $c^q_a$. Thus, in the second and forth condition, we add $c^q_a$. Furthermore, if we move inside $U_i$, then we add $\gamma(q,a)$ to account for the transition (as in the first and second condition), otherwise we add nothing. In particular, the zero in the third condition mainly serve as a middle man for getting the correct cumulated cost in $\bar{Z}$. The procedure for modifying $D_i$ to $\bar{D}_i$ is~analogous.

With the procedure above, we get a modified subgraph $\bar{\M} := \{ \bar{U}_1,\dots,\bar{U}_n \} \cup \{ \bar{D}_1, \dots, \bar{D}_m\}$ identical to $\M$ except that $U_i$ ($D_i$) has been replaced by $\bar{U}_i$ ($\bar{D}_i$). This is our reduced HiMM $\bar{Z}$:

\begin{definition}\label{reduced_himm_def}
Let HiMM $Z$ be given and $s_{\mathrm{init}}$, $s_{\mathrm{goal}}$ be states of $Z$. The reduced HiMM with respect to $s_{\mathrm{init}}$, $s_{\mathrm{goal}}$ is the tuple $\bar{Z} = (\bar{X},\bar{T})$, where $\bar{T} = \bar{\M}$ and $\bar{X}$ consists of the MMs of $\bar{\M}$. The hierarchical transition function $\bar{\psi}$ for $\bar{Z}$ is constructed as in Definition \ref{HiMM_def}, while the hierarchical cost function $\bar{\chi}$ is modified. More precisely, given $q \in Q_M$ with $M \in \bar{\M}$ and $a \in \Sigma$,
\begin{equation*}
\bar{\chi}(q,a) = 
\begin{cases}
\gamma_M(q,a) & \parbox[t]{.2\textwidth}{$\delta_M(q,a) \neq \emptyset$} \\
\gamma_M(q,a)+\bar{\chi}(w,a) & \parbox[t]{.3\textwidth}{$\delta_M(q,a) = \emptyset, \\ (W \xrightarrow{w} M) \in \bar{T}, W \in \bar{X}$} \\
\emptyset & \parbox[t]{.1\textwidth}{otherwise.}
\end{cases}
\end{equation*}
\end{definition}

\begin{remark}\label{remark_reduce}
\if\longversion0 See \cite{stefansson2024modular} for an algorithm constructing the reduced HiMM~$\bar{Z}$.\else
See Appendix for an algorithm constructing the reduced HiMM~$\bar{Z}$.
\fi

\end{remark}

A trajectory and a plan are defined analogously for $\bar{Z}$ as for $Z$. The cumulative cost $\bar{C}(z)$ for a trajectory $z = (q_i,a_i)_{i=1}^N$ is also defined analogously ($\bar{C}(z) = \sum_{i=1}^N \bar{\chi}(q_i,a_i)$ if $\bar{\psi}(q_i,a_i) \neq \emptyset$ for all $i$ and $\bar{C}(z) = \infty$ otherwise), so we can consider optimal plans and optimal trajectories to the planning objective $(\bar{Z},s_{\mathrm{init}},s_{\mathrm{goal}})$. Of course, finding an optimal trajectory to $(\bar{Z},s_{\mathrm{init}},s_{\mathrm{goal}})$ only make sense if it can relate it to an optimal trajectory in $({Z},s_{\mathrm{init}},s_{\mathrm{goal}})$. In fact, the two objectives are equivalent:

\begin{theorem}[Planning equivalence]\label{th:planning_equivalence}
Consider HiMM $Z$ and reduced HiMM $\bar{Z}$ for initial state $s_{\mathrm{init}}$ and goal state $s_{\mathrm{goal}}$. We have:
\begin{enumerate}[(i)]
\item Let $z$ be an optimal trajectory to $({Z},s_{\mathrm{init}},s_{\mathrm{goal}})$. Then, the reduced trajectory $z_R$ is an optimal trajectory to $(\bar{Z},s_{\mathrm{init}},s_{\mathrm{goal}})$.
\item Let $z$ be an optimal trajectory to $(\bar{Z},s_{\mathrm{init}},s_{\mathrm{goal}})$. Then, the optimal expansion $z_E$ is an optimal trajectory to $({Z},s_{\mathrm{init}},s_{\mathrm{goal}})$.
\end{enumerate}
\end{theorem}
Theorem \ref{th:planning_equivalence} suggests that we can search for an optimal trajectory in the possibly much smaller reduced HiMM $\bar{Z}$ instead of $Z$, and then just expand it to get an optimal trajectory for $Z$. Here, the reduced trajectory $z_R$ of a trajectory $z$ is intuitively equal to $z$ except that we have removed all transitions happening inside the removed subtrees, while the optimal expansion can be seen as the reverse operator: for a trajectory $z$ of $\bar{Z}$, each station-input pair $(q,a)$ of $z$ is replaced by an optimal $a$-exit trajectory of the subtree that $q$ corresponds to in $Z$, with result $z_E$. See Fig. \ref{fig:reduce_expand_example} for an illustration. \if\longversion0 We refer to \cite{stefansson2024modular} for formal definitions of a reduced trajectory and an optimal expansion.\else
We refer to the Appendix for formal definitions of a reduced trajectory and an optimal expansion.
\fi

\subsubsection{Solve}\label{solve}
With the insight of Theorem \ref{th:planning_equivalence}, we now solve $(\bar{Z},s_{\mathrm{init}},s_{\mathrm{goal}})$, summarised by Algorithm \ref{alg:solve}. For a more streamlined notation, we also include $s_{\mathrm{init}}$ ($s_{\mathrm{goal}}$) in the sequence of $\bar{U}_i$:s ($\bar{D}_i$:s) by setting $\bar{U}_0 = s_{\mathrm{init}}$ ($\bar{D}_0 = s_{\mathrm{goal}}$).  Algorithm \ref{alg:solve} is divided into two parts. In the first part, we find an optimal trajectory from $s_{\mathrm{init}}$ to $\bar{D}_\beta$. 

Here, $\bar{D}_\beta$ is the last $\bar{D}_i$ before the sequences $(\bar{D}_0, \bar{D}_1,\dots, \bar{D}_m)$ and $(\bar{U}_0, \bar{U}_1,\dots, \bar{U}_n)$ become identical (e.g., $\bar{D}_\beta = \bar{C}$ in Fig. \ref{fig:Himm_transition_example_v2} for $s_\mathrm{init} = 3$ and $s_\mathrm{goal} = 6$).
In the second part, we find an optimal trajectory from $\bar{D}_\beta$ to $s_{\mathrm{goal}}$. Finally, we just patch them together to get an optimal trajectory for $(\bar{Z},s_{\mathrm{init}},s_{\mathrm{goal}})$. This decomposition is possible due to the following elementary~observation:
\begin{lemma}\label{lemma:pass_through}
Any optimal trajectory $z = (q_i,a_i)_{i=1}^N$ to $(\bar{Z},s_\mathrm{init},s_\mathrm{goal})$ has a state-input pair $(q_k,a_k)$ such that $\bar{\psi}(q_k,a_k) = \mathrm{start}(\bar{D}_\beta)$ (informally, $z$ must pass through~$\bar{D}_\beta$).\footnote{Here, we use the natural convention that $\mathrm{start}(q) = q$ if $q$ is a state of $\bar{Z}$ (e.g., $\bar{D}_\beta = s_{\mathrm{goal}}$).}
\end{lemma}

We now provide the details of the first step. To solve $(\bar{Z},s_{\mathrm{init}},\bar{D}_\beta)$, we  construct a graph $G$ with nodes being a subset of the states of the $\bar{U}_i$:s. These nodes are
\begin{itemize}
\item From $\bar{U}_1$: $s_\mathrm{init}$ and $s_{\bar{U}_1}$;
\item From $\bar{U}_i$ with $i>1$: $s_{\bar{U}_i}$ and all states in $\bar{U}_i$ we get to by exiting $\bar{U}_{i-1}$;
\item A node representing $\bar{D}_\beta$.
\end{itemize}
A key feature of $G$ is that when we enter or exit an MM $\bar{U}_i$, we always end up in a node of $G$. Therefore, we can search locally in each $\bar{U}_i$ using Dijkstra's algorithm to obtain optimal trajectories and costs from the nodes of $G$ in $\bar{U}_i$ to the nodes we get to by exiting $\bar{U}_i$ or entering $\bar{U}_{i-1}$ (loosely speaking, the adjacent nodes to the nodes in $\bar{U}_i$).\footnote{In the MM $\bar{U}_i$ that has $\bar{D}_\beta$ as state, then we also search to $\bar{D}_\beta$. Moreover, we do not search from $\bar{D}_\beta$ since we only want to get to $\bar{D}_\beta$.}  These optimal trajectories and costs form the labelled arcs of $G$. Then, we use a bidirectional Dijkstra to find a shortest-path in $G$ from $s_\mathrm{init}$ to $\bar{D}_\beta$.\footnote{We could have used Dijkstra's algorithm here again, but since we only have one destination node ($\bar{D}_\beta$), bidirectional Dijkstra is generally faster \cite{bast2016route}.} By concatenating the optimal trajectories (on the arcs) on this shortest-path, we obtain an optimal trajectory $z_1$ from $s_\mathrm{init}$ to $\bar{D}_\beta$ and thus also to $\mathrm{start}(\bar{D}_\beta)$. The algorithm is summarised by Algorithm \ref{alg:solve}, Part 1. \if\longversion0 We refer to \cite{stefansson2024modular} for additional~details.\else
We refer to the Appendix for additional~details.
\fi

\begin{figure}[t]
\centering
\includegraphics[width=0.5\linewidth]{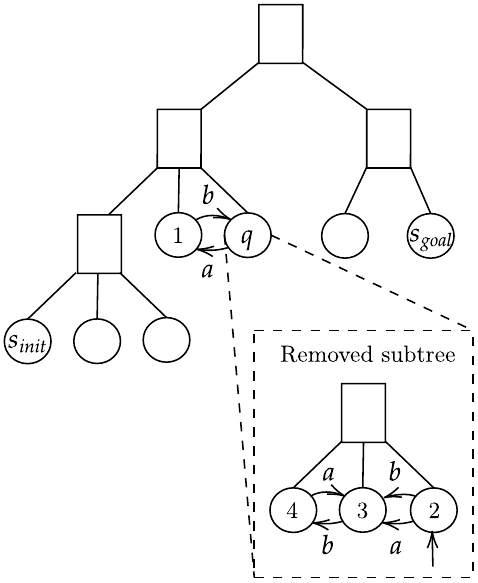}
\caption{Illustration of a reduced trajectory and optimal expansion. Here, $q$ has replaced the depicted subtree in $\bar{Z}$. For the trajectory $z = ( (1,b), (2,b), (3,a) )$ of $Z$, the reduced trajectory is therefore $z_R = ( (1,b),(q,a) )$. Conversely, the optimal expansion of $z = ( (1,b),(q,a) )$ is $z_E = ( (1,b), (2,b), (3,a) )$. For simplicity, we plot just the relevant transitions and assume unit costs.}
\label{fig:reduce_expand_example}
\end{figure}





\begin{algorithm}[t]
\caption{Solve}\label{alg:solve} 
\begin{algorithmic}[1]
\Require $(\bar{Z},s_\mathrm{init},s_\mathrm{goal})$ and $(c_a^M,z_a^M)_{a \in \Sigma, M \in X}$.
\Ensure An optimal trajectory $z$ to $(\bar{Z},s_\mathrm{init},s_\mathrm{goal})$.
\State \textbf{Part 1: Solve $(\bar{Z},s_\mathrm{init},\mathrm{start}(\bar{D}_\beta))$}
\State Set $G$ to an empty graph \Comment{To be constructed}
\State Add nodes to $G$ 
\For {$i =1,\dots,n$} \Comment{Consider $\bar{U}_{i}$}
\State Get states to search from and too, $I$ and $D$, in $\bar{U}_i$.
\For {$s \in I$}
\State $(c_d, z_d)_{d \in D} \gets \mathrm{Dijkstra}(s, D, \bar{U}_{i})$
\State Add arcs corresponding to $(c_d, z_d)_{d \in D}$ in $G$ 
\EndFor
\EndFor
\State $z_1 \gets \mathrm{Bidirectional\_Dijkstra}(s_\mathrm{init}, \bar{D}_\beta, G)$
\State \textbf{Part 2: Solve $(\bar{Z},\mathrm{start}(\bar{D}_\beta),s_\mathrm{goal})$}
\State Let $z_2$ be an empty trajectory \Comment{To be constructed}
\State Let $\bar{D}_{i}$ be the MM with state $\mathrm{start}(\bar{D}_\beta)$
\State $c \gets \mathrm{start}(\bar{D}_\beta)$ 
\State $g \gets$ state in $\bar{D}_{i}$ corresponding to $\bar{D}_{i-1}$ 
\While {$c \neq s_\mathrm{goal}$}
\State $z \gets \mathrm{Bidirectional\_Dijkstra}(c,g,\bar{D}_{i})$ 
\State $z_2 \gets z_2 z$ \Comment{Concatenate $z_2$ with $z$}
\State $c \gets \mathrm{start}(\bar{D}_{i-1})$
\If {$c \neq s_\mathrm{goal}$}
\State Let $\bar{D}_{i}$ be the MM with state $c$
\State $g \gets$ corresponding state in $M$ to $\bar{D}_{i-1}$ 
\EndIf
\EndWhile \Comment{End of Part 2}
\State return $z \gets z_1 z_2$ \Comment{Concatenate $z_1$ with $z_2$}
\end{algorithmic}
\end{algorithm}


We now continue with the second part solving $(\bar{Z},\mathrm{start}{(\bar{D}_\beta)},s_{\mathrm{goal}})$. In this case, we need to go down in the hierarchy of $\bar{D}_i$:s. Here, we note that if we are in the MM $\bar{D}_i$, then we can only go down in the hierarchy by reaching its state $\bar{D}_{i-1}$. Therefore, we only need to search for an optimal trajectory in $\bar{D}_i$ from our current state to $\bar{D}_{i-1}$, and then repeat this procedure (starting at $\mathrm{start}{(\bar{D}_\beta)}$) until we are down in the MM $\bar{D}_1$, where we only need to search to $s_{\mathrm{goal}}$. Concatenating these optimal trajectories, we get an optimal trajectory $z_2$ to $(\bar{Z},\mathrm{start}{(\bar{D}_\beta)},s_{\mathrm{goal}})$. The procedure is summarised by Algorithm \ref{alg:solve}, Part 2.

Concatenating the optimal trajectory $z_1$ from the first part with the optimal trajectory $z_2$ from the second part, gives an optimal trajectory $z$ to $(\bar{Z},s_{\mathrm{init}},s_{\mathrm{goal}})$, done in the final line of Algorithm \ref{alg:solve}. The correctness of Algorithm \ref{alg:solve} is given by the following result:

\begin{proposition}\label{proposition:solve_correct}
Algorithm \ref{alg:solve} computes an optimal trajectory $z$ to $(\bar{Z},s_{\mathrm{init}},s_{\mathrm{goal}})$.
\end{proposition}

\subsubsection{Expand}
In the final step, we expand the optimal trajectory $z$ to $(\bar{Z},s_{\mathrm{init}},s_{\mathrm{goal}})$ from the solve step to obtain an optimal plan $u$ to $({Z},s_{\mathrm{init}},s_{\mathrm{goal}})$. Expanding to an optimal plan $u$ is identical to the case where one expands to an optimal trajectory except that one saves just the inputs. We can either expand $z$ sequentially, obtaining the next optimal input from $u$, or obtain the full $u$ at once. The former is more beneficial for real-time executions having a lower time complexity, as seen next.

\subsubsection{Time Complexity}
We conclude with the time complexity of the Optimal Planner:
\begin{proposition}\label{efficient_time_complexity}
The three steps of the Optimal Planner have time complexities:
\begin{itemize}
\item $\mathrm{Reduce}$ has time complexity $\mathrm{t}_\mathrm{Reduce} = O(\depth(Z) \cdot b_s)$.
\item $\mathrm{Solve}$ has time complexity
\begin{align*}
\mathrm{t}_\mathrm{Solve} = \nonumber \\
O \Big ( \depth(Z) |\Sigma| \cdot \Big [ b_s |\Sigma| + (b_s+|\Sigma| ) \log \big ( (b_s+|\Sigma| )\big ) \Big]  \Big ) + \nonumber \\
O\Big (|\Sigma|^2 \depth(Z)+ |\Sigma| \depth(Z) \log \big ( |\Sigma| \depth(Z) \big ) \Big ) + \nonumber \\
O\Big (\depth(Z) \cdot \big ( b_s |\Sigma|+b_s\log(b_s) \big ) \Big ).
\end{align*}
\item $\mathrm{Expand}$ has time complexity $\mathrm{t}_\mathrm{expand} = O(\mathrm{depth}(Z))$ for obtaining the next input $u$ and $\mathrm{t}_\mathrm{expand} = O(\mathrm{depth}(Z) |u|)$ for obtaining the full optimal plan $u$ at once, where $|u|$ is the length of $u$.
\end{itemize}
Therefore, the time complexity of the Optimal Planner is
\begin{equation*}
\mathrm{t}_\mathrm{Reduce}+\mathrm{t}_\mathrm{Solve}+\mathrm{t}_\mathrm{expand}.
\end{equation*}
In particular, with bounds on the state set size and input size of each MM in $Z$, the Optimal planner obtains the next optimal input in time $O(\mathrm{depth}(Z) \log(\mathrm{depth}(Z)))$.
\end{proposition}

Note that the time complexity $O(\mathrm{depth}(Z) \log(\mathrm{depth}(Z)))$ of the Optimal planner for bounded MMs  is near-linear in the depth of the HiMM $Z$ (as oppose to Dijkstra's algorithm which in the worst case can be more than exponential in $\depth(Z)$). We therefore consider the Optimal Planner to be~efficient.

\section{Reconfigurable Hierarchical Planning}\label{reconfig_hier_planning}
In this section, we extend the planning algorithm in Section \ref{efficient_hier_planning} to be reconfigurable. More precisely, the planning algorithm in Section \ref{efficient_hier_planning} is based on a fixed HiMM $Z=(X,T)$. Therefore, the more computationally costly Optimal Exit Computer needs to be executed only once (in time $O(|X|)$). The computed optimal exit costs that can then be repeatedly used by the faster Optimal Planner to rapidly compute optimal plans for different initial and goal states (in time $O(\depth(Z) \cdot \log (\depth(Z)))$). However, if $Z$ \emph{changes}, then the optimal exit costs might not be valid anymore and need to be recomputed. Naively recomputing all optimal exits costs every time a change has occurred would require a significant amount of computations slowing down the planning algorithm, undermining the benefit for computing the optimal exit costs in the first place. Fortunately, there is better way. More precisely, by keeping track of the changes that have occurred, we also keep track of the optimal exit costs that might not be valid anymore, marking every MM $M$ whose optimal exit costs needs to be recomputed. These MMs are typically only a fraction of all the MMs, therefore, updating the Optimal Exit Computer only for these MMs is typically much faster.

In detail, the changes are formalised by the modifications introduced in Section \ref{modifications}. These modifications change the given HiMM $Z=(X,T)$, and the planning algorithm gets informed of these modifications schematically illustrated by Fig. \ref{fig:algorithm_overview}, with resulting algorithm summarised by Algorithm \ref{alg:hierarchical_planning_new}, where parts handling modifications are coloured red. More precisely, a change is given as a sequence of modifications $m_{\mathrm{seq}}$, where each modification $m$ in $m_{\mathrm{seq}}$ modifies $Z$ (line \ref{alg:hierarchical_planning_new:line3}), and the planning algorithm correspondingly marks each MM that needs to recompute its optimal exit costs due to $m$ (line \ref{alg:hierarchical_planning_new:line4}). Here, the marked MMs, due to $m$, are the MM affected by $m$ plus all its ancestor MMs upstream in the hierarchy (e.g., if $D$ in Fig. \ref{fig:Himm_transition_example_v2} is affected by an arc modification, then $D$, $B$ and $A$ are all marked), \if\longversion0 see \cite{stefansson2024modular} for details. \else
see Appendix for details.
\fi
Then, the optimal exit costs gets updated (line \ref{alg:hierarchical_planning_new:line6}), while the Optimal Planner (line \ref{alg:hierarchical_planning_new:line8}-\ref{alg:hierarchical_planning_new:line10}) is as before.


\subsubsection{Computing Optimal Exit Costs}
The optimal exit costs are computed using Algorithm \ref{alg:update_offline_step}. The difference with Algorithm \ref{alg:compute_optimal_exits} are the added lines \ref{alg:update_offline_step:line3} and \ref{alg:update_offline_step:line16}, where line \ref{alg:update_offline_step:line3} checks if one needs to compute the optimal exit costs $(c_a^M)_{a \in \Sigma}$ for a given MM $M$ (i.e., either $(c_a^M)_{a \in \Sigma}$ needs to be updated due to a modification, or $(c_a^M)_{a \in \Sigma}$ has never been computed), while line \ref{alg:update_offline_step:line16} unmarks $M$ since $(c_a^M)_{a \in \Sigma}$ are now~{up-to-date}.

\begin{proposition}\label{prop:optimal_exit_costs_correct_reconfig}
Algorithm \ref{alg:update_offline_step} correctly updates the optimal exit costs given a sequence of modifications $m_{\mathrm{seq}}$ (with $m_{\mathrm{seq}} = \emptyset$ if no change has occurred).
\end{proposition}

\subsubsection{Time complexity}
The benefit with the marking procedure is the savings in computations when $Z$ changes. More precisely, given a modification $m$, the Optimal Exit Computer needs only $O(\depth(Z))$ time to update the optimal exit costs:


\begin{proposition}[Reconfigurability]\label{prop:reconfigurability}
Let HiMM $Z = (X,T)$ be given and consider a single modification $m_{\mathrm{seq}} = m$. Assume $Z$ and all added subtrees (i.e., $Z_{\mathrm{add}}$ if $m$ is a state addition or $Z_{\mathrm{seq}} = \{Z_1,\dots,Z_n\}$ if $m$ is composition) have up-to-date optimal exit costs before the modification $m$. Then the time complexity of the Optimal Exit Computer is $O( \mathrm{depth}(Z) \cdot [b_s |\Sigma|+(b_s+|\Sigma|) \log(b_s+|\Sigma|)] )$
if $m$ is a state addition, state subtraction or arc modification, and $O( b_s |\Sigma|+(b_s+|\Sigma|) \log(b_s+|\Sigma|) )$ if $m$ is a composition. In particular, with bounded state set size and input set size for each MM in $Z$, we get time complexity $O( \mathrm{depth}(Z))$. 
\end{proposition}

The planning algorithm is in this sense reconfigurable.

\begin{figure}
	\centering
  \includegraphics[width=0.49\textwidth]{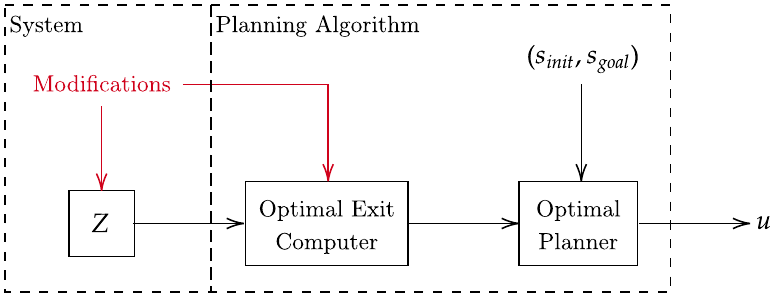}
  \caption{The planning algorithm and its interaction with the changing system $Z$.} 
  \label{fig:algorithm_overview}
\end{figure}

\begin{algorithm}[t]
\caption{Compute\_optimal\_exits}\label{alg:update_offline_step}
\begin{algorithmic}[1]
\Require HiMM $Z = (X,T)$ with markings.
\Ensure Computed $(c_a^M,z_a^M)_{a \in \Sigma}$ for each MM $M$ of $Z$
\State $\mathrm{Optimal\_exit}(M_0)$ \Comment{Run from root MM $M_0$ in $T$}
\State $\mathrm{Optimal\_exit}(M)$: \Comment{Recursive help function} 
\If {\textcolor{BrickRed}{$M$ is not marked} \label{alg:update_offline_step:line3} }
\State return $(c_a^M,z_a^M)_{a \in \Sigma}$ \Comment{Since already up-to-date}
\Else
\For {each state $q$ in $Q(M)$}
\If{$q \in S_Z$}
\State $(c_a^q)_{a \in \Sigma} \gets 0_{|\Sigma|}$
\Else
\State Let $N_q$ be the MM corresponding to $q$.
\State $(c_a^{q},z_a^{q})_{a \in \Sigma} \gets \mathrm{Optimal\_exit}(N_q)$
\EndIf
\EndFor
\State Construct $\hat{M}$
\State $(c_a^M,z_a^M)_{a \in \Sigma}$ $\gets$ Dijkstra($s(M),\{E_a\}_{a \in \Sigma},\hat{M}$)
\State \textcolor{BrickRed}{Unmark $M$} \label{alg:update_offline_step:line16}
\State return $(c_a^M,z_a^M)_{a \in \Sigma}$
\EndIf
\end{algorithmic}
\end{algorithm}

\section{Hierarchical Planning with \\ Identical Machines}\label{duplicate_hier_planning}
Finally, we extend the planning algorithm to HiMM with identical machines to further lower the time complexity. The overall structure is as in Algorithm \ref{alg:hierarchical_planning_new}. The difference lies instead in the details of the called functions. We go through these details for the Optimal Exit Computer and the Optimal Planner separately, assuming a given HiMM $Z = (X,T)$ with identical machines, as in Section \ref{duplicates}.
\subsection{Optimal Exit Computer}
Computing the optimal exit costs for a HiMM with identical machines is done using Algorithm \ref{alg:update_offline_step} just as in the reconfigurable case. Note that $T$ is not a tree anymore but a DAG with a single root. However, Algorithm \ref{alg:update_offline_step} works nonetheless. In particular, the marking procedure now serves as a neat way of only computing the optimal exit costs once for multiple identical MMs (e.g., computing the optimal exit costs of the MM House in Fig. \ref{fig:duplicates_example} once, even though this MM is used for several of the houses). Furthermore, modifications are handled in the same way as before. Therefore, the Optimal Exit Computer is identical to the one in Algorithm \ref{alg:hierarchical_planning_new}. 

\subsection{Optimal Planner}\label{optimal_planner_duplicate_case}
The Optimal Planner is similar to the case without identical machines given in Section \ref{optimal_planner}. A minor technical difference is in the reduce step (line \ref{alg:hierarchical_planning_new:line8} in Algorithm \ref{alg:hierarchical_planning_new}) where the paths of $s_{\mathrm{init}}$ and $s_{\mathrm{goal}}$ (see how states are specified for a HiMM with identical machines in Section \ref{duplicates}) can be used directly to obtain $\{ U_1,\dots,U_n \}$ and $\{ D_1, \dots, D_m\}$, illustrated by Fig. \ref{fig:reduce_duplicate_case}. This difference is needed since an MM $M \in X$ may now have several parent MMs. The reduced HiMM $\bar{Z}$ is then constructed analogously to before as in Definition \ref{reduced_himm_def}. In particular, $\bar{Z}$ has a tree structure (as opposed to $Z$). We can therefore proceed just as in Section \ref{optimal_planner}, computing an optimal trajectory $z$ to $(\bar{Z},s_{\mathrm{init}},s_{\mathrm{goal}})$ and then expand it to an optimal plan $u$ to $({Z},s_{\mathrm{init}},s_{\mathrm{goal}})$. The expand step is possible since it is agnostic to identical machines in the HiMM. \if\longversion0 We refer to \cite{stefansson2024modular} for details.\else
We refer to the Appendix for details.
\fi


\begin{figure}[t]
\centering
\includegraphics[width=\linewidth]{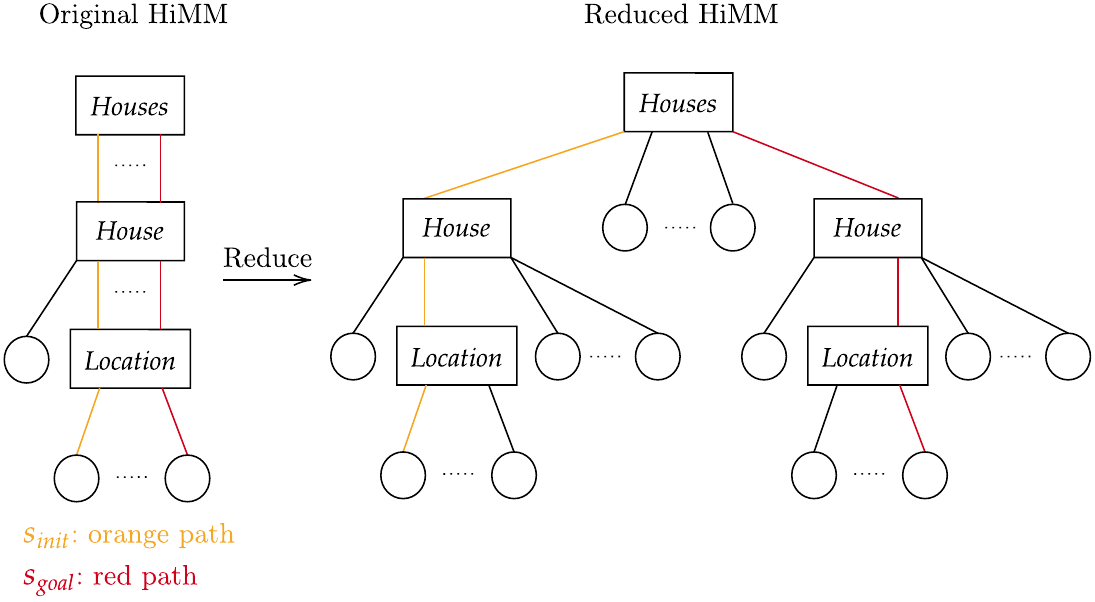}
\caption{Illustration how the reduced HiMM is obtained, for the robot application from Section \ref{illustrative_example}, from the original HiMM $Z$ with identical machines. The yellow (red) path in $Z$ denotes $s_{\mathrm{init}}$ ($s_{\mathrm{goal}}$), which are used as paths directly for the reduced HiMM. The other paths are replaced with states just as before in Section~\ref{reduce_step}.}
\label{fig:reduce_duplicate_case}
\end{figure}

\subsection{Time Complexity}
The time complexity of the Optimal Exit Computer for an HiMM with identical machines has the same expression as before given by Proposition \ref{th:compute_optimal_exits_time}, but we can still get speed-ups. The reason is because the number of MMs $|X|$ can now be significantly smaller due to the more compact representation a HiMM with identical machines can have. In fact, $|X|$ can be reduced logarithmically (compared to the case when one would not use the compact representation given by HiMM with identical machines). As an example, note that the robot application given in Fig. \ref{fig:duplicates_example} needs only 3 MMs using a HiMM with identical machines, while  treating all MMs as different needs $1+10+10 \cdot 100 =1011$ MMs (see Section \ref{numerical_evaluations} for details). This difference in $|X|$ can lead to a huge difference in computing time for the Optimal Exit Computer since it scales linearly with the number of MMs $|X|$ in $Z$. The Optimal Planner is not affected by the more compact representation identical machines may yield, therefore having identical time complexity.


\section{Numerical Evaluations}\label{numerical_evaluations}
In this section, we provide two numerical studies validating the performance of our planning algorithm (Algorithm \ref{alg:hierarchical_planning_new}). More precisely, in the first study (Large-scale System), we validate how the planning algorithm scales with system size. In the second study (Robot Warehouse Application), we showcase our planning algorithm on the robot application given in Section \ref{illustrative_example} and validate the reconfigurability of the algorithm. In both studies, we also consider how the examples can be modelled more compactly and solved faster using HiMMs with identical machines. We compare our algorithm with Dijkstra's algorithm \cite{Dijkstra1959,DijkstraFibonacci}, Bidirectional Dijkstra \cite{bast2016route} and Contraction~Hierarchies \cite{geisberger2012exact}. \if\longversion0 See \cite{stefansson2024modular} for a brief discussion of Contraction~Hierarchies, how it differ from our algorithm, and implementation details.\else See Appendix \ref{implementation_details} for a brief discussion of Contraction~Hierarchies, how it differ from our algorithm, and implementation details.
\fi
 The simulations are run on a MacBook Pro with 2.9 GHz i7 processor and 16 GB~RAM.


\subsection{Study 1: Large-scale System}
In the first study, the system is created by recursively expanding an MM $M$ with itself creating an arbitrarily large HiMM $Z$. More precisely, the MM $M$ we consider is given by Fig. \ref{fig:recursive_example} (left). In a recursion, we replace state 0 and 2 in $M$ by $M$ itself, resulting in an HiMM $Z$ given by Fig. \ref{fig:recursive_example} (right). Repeating this procedure, replacing state 0 and 2 of each leaf MM in the HiMM with $M$, we construct an HiMM $Z$ of arbitrary depth (and size). For planning, we consider unit costs for all transitions with initial state $s_\mathrm{init}$ being the left-most state (e.g., state 2 in Fig. \ref{fig:recursive_example}) and goal state $s_\mathrm{goal}$ being the right-most state (e.g., state 7 in Fig. \ref{fig:recursive_example}).

\begin{figure}[t]
\centering
\includegraphics[width=0.7\linewidth]{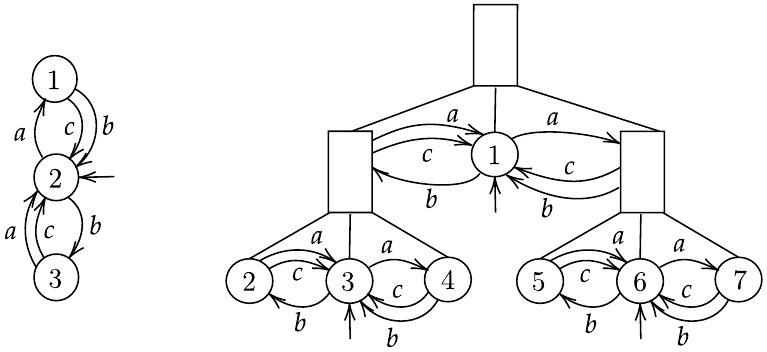}
\caption{The MM $M$ (left) and resulting HiMM $Z$ (right) from Study 1 using one recursion.}
\label{fig:recursive_example}
\end{figure}  

\subsubsection{Result}
The execution time is shown in Fig. \ref{fig:large_recursive_system_journal}, where we vary the depth of the HiMM $Z$ by varying the number of recursions. Here, we see that for small depths $\leq 8$, both Dijkstra and Bidirectional Dijkstra perform better than our Optimal Planner. However, our planning algorithm scales significantly better as the system size increases. This is reflected for systems with depth $>8$ where Optimal Planner outperforms Dijkstra, and at depths $>13$ where Optimal Planner outperforms Bidirectional Dijkstra. In particular, at depth 20 with a system of about 2 million states, Optimal Planner finds an optimal plan in 0.0024 s compared to 0.029 s for Bidirectional Dijkstra and 12 s for~Dijkstra. The reason for the speed-up in the Optimal Planner is the preprocessing step of the Optimal Exit Computer. This step takes at most 42 s (for the HiMM with depth 20), that is, around the same order of magnitude as one Dijkstra, but needs to be done only once and can then be repeatedly used for several queries, a huge advantage during repeated planning in the same HiMM. 


We continue by comparing our planning algorithm with Contraction Hierarchies. We see that Contraction Hierarchies is generally faster than the Optimal Planner, but seems to have a worse scaling factor as system size grows. In particular, Contraction Hierarchies is just 2 times faster for the largest system with depth 20.\footnote{It would be interesting to increase the depth further to see if this scaling pattern persists. However, manipulating such large systems, such as flattening them to make them suitable for Dijkstra, Bidirectional Dijkstra and Contraction Hierarchies, makes computations cumbersome for much larger depths. However, if using an HiMM with identical machines, our planning algorithm can handle systems of much larger depths, considered in the next~paragraph.} This speed-up comes also at the cost of a more expensive preprocessing step. In particular, for the HiMM $Z$ with depth 20, the preprocessing step of Contraction Hierarchies is 20 times slower than the Optimal Exit Computer (1200 s compared to 42 s). 
Thus, our method obtains a good balance between preprocessing and query time, where it is faster than Bidirectional Dijkstra during query (at the expense of a slower preprocessing step), and slower than Contraction Hierarchies (but with a faster preprocessing step). 

Finally, we show how we can rapidly speed-up the computation time of the Optimal Exit Computer by compactly representing the large-scale system $Z$ as an HiMM with identical machines. More precisely, note that all MMs at a given depth in $Z$ are actually roots of identical subtrees (of nested machines). Therefore, we can compactly represent $Z$ using just $\depth(Z)$ number of MMs (one for each depth), i.e. the number of machines scales linearly with depth (as opposed to exponentially when not grouping identical machines together). Consequentially, the Optimal Exit Computer using this compact representation preprocesses $Z$ much faster, as seen in Fig. \ref{fig:large_recursive_system_journal} (Optimal Exit Computer: Identical Machines). In particular, for $Z$ with depth 20, the Optimal Exit Computer is done in just 0.00083 s, compared to 42 s not exploiting that machines are identical, or 1200 s for the preprocessing step of Contraction Hierarchies. In fact, the compact representation enables computation of optimal plans for much larger depths. As an example, for the HiMM $Z$ with depth 500 and around $10^{150}$ states, the Optimal Exit Computer using the compact representation takes just 0.019 s while the Online Planner finds an optimal plan in 0.088 s. For all the other methods, this problem is intractable.

\begin{figure}[t]
\centering
\includegraphics[width=0.9\linewidth]{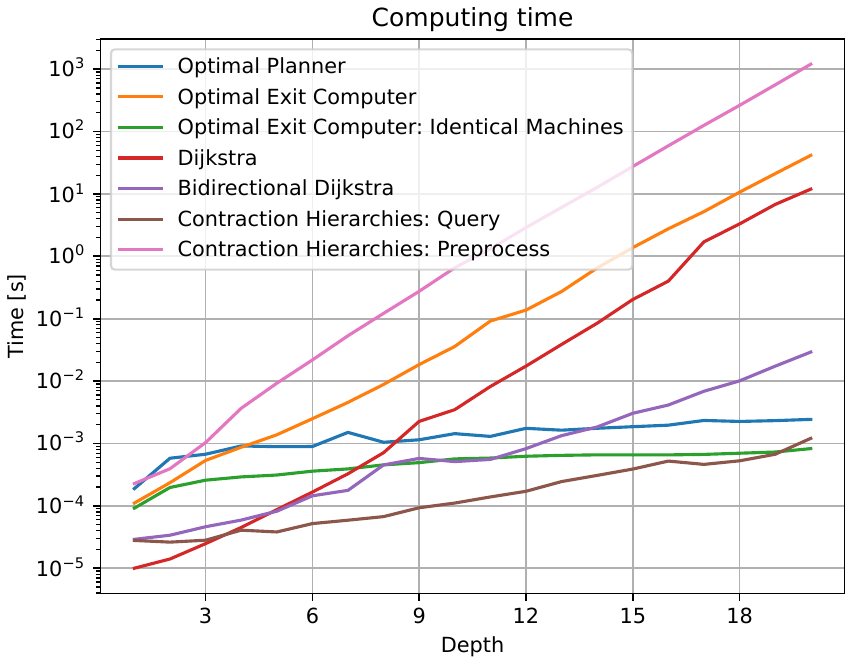}
\caption{Computing time for varying depth in Study 1.}
\label{fig:large_recursive_system_journal}
\end{figure}

\subsection{Study 2: Robot Warehouse Application}\label{numerical_eval_robot_warehouse}
The second study considers the robot application in Section \ref{illustrative_example}, here formalised. The HiMM $Z$ is constructed using three~MMs.

The first MM $M_1$ has 10 states corresponding to the houses, ordered in a line with the robot can move to neighbouring houses, at a cost of 100, as in Fig. \ref{fig:robot_example_detailed_overview}. The start state is $s_{M_1} = 1$.

The second MM $M_2$ models the robot moving in a house. The robot can move, in a $10 \times 10$ grid, to nearby locations at a cost of 1. The robot can also go out to an entrance state $S$ as in Fig. \ref{fig:robot_example_detailed_overview}, and exit $M_2$ by going left or right (at $S$), moving to the corresponding nearby house. The start state is~$S$.

The third MM $M_3$ specifies what the robot can do at a location. The robot starts in an idle state $S$, and can either (at $S$): exit $M_3$ and transition to the corresponding nearby location in $M_2$; or, interact with a lab desk. In the latter, the robot starts steering an arm over a $3 \times 3$ test tube rack (similar to the grid dynamics in $M_2$), with arm initially at (1,1). The robot can also scan the tube the arm is currently over (the tube is then remembered and no other tube can be scanned). The robot can quit steering the arm at (1,1), going back to $S$, see Fig. \ref{fig:robot_example_detailed_overview}. All costs are 0.5, except scanning which costs 4. Number of states is $1+9\times 10 = 91$ states with $S$ as start~state.\footnote{The number of states is derived as follows. For each grid position (i.e., 9 positions), there are 10 different tubes that can have been scanned: any of the 9 tubes or none. This contributes to $9\times 10$ states. With $S$, we thus have $1+9\times 10$ (what has been scanned is forgotten at $S$).}


The HiMM $Z$ is constructed by nesting $M_1$, $M_2$, $M_3$ into a hierarchy as in Fig. \ref{fig:robot_example_detailed_overview}, with approximately $91 000$ states.



\subsubsection{Cases}
We consider three different cases to show how our algorithm handles changes. Case 1 is $Z$ just described. In Case 2, we change $Z$ by adding another identical house, House 11, with House 10 as neighbour, formally done by applying the modifications state addition (adding House 11) followed by an arc modification (connected House 11 with House 10).  In Case 3, we again start from $Z$ in Case 1, but instead block som locations in House 2 (as depicted in Fig. \ref{fig:robot_example_change_overview}), by applying state subtractions. In all cases, the robot starts with the arm over tube (3,3) in the bottom-right corner of House 1, with goal to scan tube (3,3) in the bottom-right corner of House 10, 11 and 2, in Case 1, 2 and 3 respectively.

\subsubsection{Result}
The computation times are given in Table \ref{table_1}. Here, the Optimal Planner finds an optimal plan in 18–26 ms, around 2–31 times faster than Dijkstra's algorithm and Bidirectional Dijkstra. On the other hand, Contraction Hierarchies is approximately 14–35 times faster than the Optimal Planner. However, this improvement in speed comes at the expense of a much more expensive preprocessing step for Contraction Hierarchies. Indeed, the Optimal Exit Computer (row 2) is around $10^3$ to $10^6$ times faster than the preprocessing step of Contraction Hierarchies. Here, we also note the improvement in preprocessing time for the Optimal Exit Computer due to reconfigurability: 0.087 s and 0.0010 s (row 2) for just updating the affected optimal exit costs compared to 0.96 s and 0.93 s (row 3) for re-computing them all again (Case 1 has the same times since all exit costs needs to be updated due to initialisation). In fact, the update step (in Case 2 and 3) is so quick that the total computing time for Optimal Planner plus Optimal Exit Computer is significantly faster than the computing time of Dijkstra's algorithm and Bidirectional~Dijkstra.

Finally, we consider the HiMM $Z$ in Case 1, but now modelled as an HiMM with identical machines. More precisely, all houses are identical and all locations are identical. Thus, we can compactly represent $Z$ with just three MMs, as schematically depicted in Fig.~\ref{fig:duplicates_example}. The resulting Optimal Exit Computer using this compact representation computes the optimal exit costs in 0.0030 s, compared to 0.80 s for the Optimal Exit Computer treating all machines as different, a significant speed up.

\begin{remark} 
Finally, we stress that an HiMM $Z$ with identical machines do not need to consist of just identical machines (as in the example above), but could just as well have distinct parts too. In this sense, an HiMM $Z$ with identical machines (Definition \ref{HiMM_duplicates_def}) is a generalisation of Definition \ref{HiMM_def}. For example, in Case 2 for the Robot Warehouse Application, House 2 becomes distinct (due to the state subtractions) and is thus automatically removed from the group of identical houses and becomes a separate (marked) subtree from $M_3$ that needs to recompute its optimal exit costs. The examples having HiMMs with identical machines were merely picked to showcase the full potential of the theory. 
\end{remark}





\begin{table}[]
\centering
\caption{Computing times for the cases in Study 2.}
\label{table_1}
\begin{tabular}{llll}
\hline
  &  Case 1 & Case 2  & Case 3  \\ \hline
Optimal Planner           & 0.018 s & 0.026 s & 0.026 s \\
Optimal Exit Computer: Update           & 0.80 s & 0.087 s & 0.0010 s \\
Optimal Exit Computer: All                & 0.80 s & 0.96 s & 0.93 s \\
Dijkstra's Algorithm                & 0.52 s & 0.67 s & 0.057 s \\
Bidirectional Dijkstra                & 0.56 s & 0.69 s & 0.057 s \\
CH – compute optimal plan               & 0.0013 s & 0.00097 s & 0.00075 s \\
CH – preprocessing               & 820 s  & 900 s & 870 s \\
\hline  
\end{tabular}
\end{table}

\section{Conclusion}\label{hier:conclusion}

\subsection{Summary}
In this paper, we have considered optimal planning in modular large-scale systems formalised as HiMMs. 

First, we introduced the formalism of an HiMM, given as a nested hierarchy of MMs, and extended this formalism to handle changes in the HiMM, given as modifications, as well as compactly representing identical machines. We also argued why this formalism is modular in a mathematical sense.

Second, we proposed a planning algorithm computing optimal plans between any two states for a given HiMM $Z$, consisting of two steps. The first step, the Optimal Exit Computer, preprocess the HiMM by computing optimal exit costs for each MM in $Z$. This step needs to be done only once for a given HiMM $Z$ with time complexity $O(|X|)$, where $|X|$ is the number of machines in $Z$. The second step, the Optimal Planner, can then repeatedly use the optimal exit costs to rapidly compute optimal plans between any states, by first removing irrelevant subtrees of $Z$, then searching on the reduced HiMM $\bar{Z}$, and finally expanding the found trajectory of $\bar{Z}$ to an optimal plan to $Z$. The resulting time complexity of the Optimal Planner is $O(\mathrm{depth}(Z) \log ( \mathrm{depth}(Z)))$. The algorithm also handles modifications efficiently in the sense that the Optimal Exit Computer only needs to recompute the optimal exit costs for parts of the HiMM affected by the modifications, with time complexity $O ( \mathrm{depth}(Z)))$. The Optimal Exit Computer also efficiently exploits identical machines, only needing to compute the optimal exit costs for one of the identical machines, which in the best case reduces the time complexity to $O ( \mathrm{depth}(Z)))$. These time complexity expressions are simplified for brevity, see the paper for the full expressions.

Finally, we validated our planning algorithm on two systems, comparing it with Dijkstra's algorithm, Bidirectional Dijkstra and Contraction Hierarchies. We noted that we significantly outperform Dijkstra's algorithm and Bidirectional Dijkstra on large-scale systems. Moreover, Contraction Hierarchies generally founds optimal plans faster than our planning algorithm (2-35 times faster), but at the expense of a slower preprocessing step in general (in the worst case several magnitudes slower for large scale systems). Thus, our planning algorithm obtains a good balance between preprocessing time and query time. Moreover, our Optimal Exit Computer handles modifications and exploits identical machines several magnitudes faster than Contraction Hierarchies (which instead recomputes the preprocessing step from scratch), and enables computation of optimal plans for enormous systems with identical machines (e.g., having $10^{150}$ states if one would not group together identical machines), intractable for the other~algorithms.

\subsection{Future Work}
There are a variety of interesting future work. First, every MM $M$ in an HiMM $Z$ has a unique start state, e.g., the entrance $S$ of House 2 in Fig. \ref{fig:robot_example_detailed_overview}. What if there are instead several possible start states? In such situations, the current formalism would need to aggregate $M$ into a larger MM to keep a unique start state. An interesting future research direction is instead to a modify the formalism to allow several start states, we aim to obtain a more decomposed hierarchy. 


Second, the HiMM $Z$ is given to us, with motivation that this originates from the design process, see e.g., \cite{harel1987statecharts}. However, it may also be the case that we encounter a flat MM without a hierarchy and desire to decompose this MM into an equivalent HiMM (to get all the benefits of having it in HiMM form such as efficient handling of changes and grouping together identical machines). Designing such a decomposition algorithm is an open research question. The decomposition algoritm found in \cite{biggar2021modular} for FSMs (without costs) and the recent acyclic-connected tree decomposition for weighted graphs \cite{stefansson2025faster} could serve as a start.

Third, there are several additional shortest path algorithms with a preprocessing step that are worth comparing with in the future such as more advanced and dynamic variants of contraction hierarchies \cite{dibbelt2016customizable,wan2025parallel}, customizable route planning \cite{delling2017customizable} and customizable hub labeling \cite{storandt2022customizable}. Another research question is how these algorithms could efficiently incorporate modifications and identical subsystems in an HiMM to speed up the preprocessing step, similar to our~algorithm. 
Furthermore, it would also be interesting to see how our planning algorithm perform on more randomly generated HiMMs (and related hierarchical structures) compared to mentioned shortest path~algorithms. 

Finally, extending the formalism to a broader class of systems such as stochastic and continuous systems, as well as modelling systems as HiMMs to efficiently solve more complex planning problems (e.g., using temporal logic \cite{luo2024decomposition,wei2025hierarchical}), are two interesting future research directions. It would also be interesting to see how this formalism can be integrated into more practical hierarchical machine frameworks such as the Unified Modelling Language \cite{booch2005unified}.

\section*{References}

\bibliographystyle{plain}
\bibliography{Ref3,Ref3_hfsm_v2,after_reviews}

\begin{IEEEbiography}[{\includegraphics[width=1in,height=1.25in,clip,keepaspectratio]{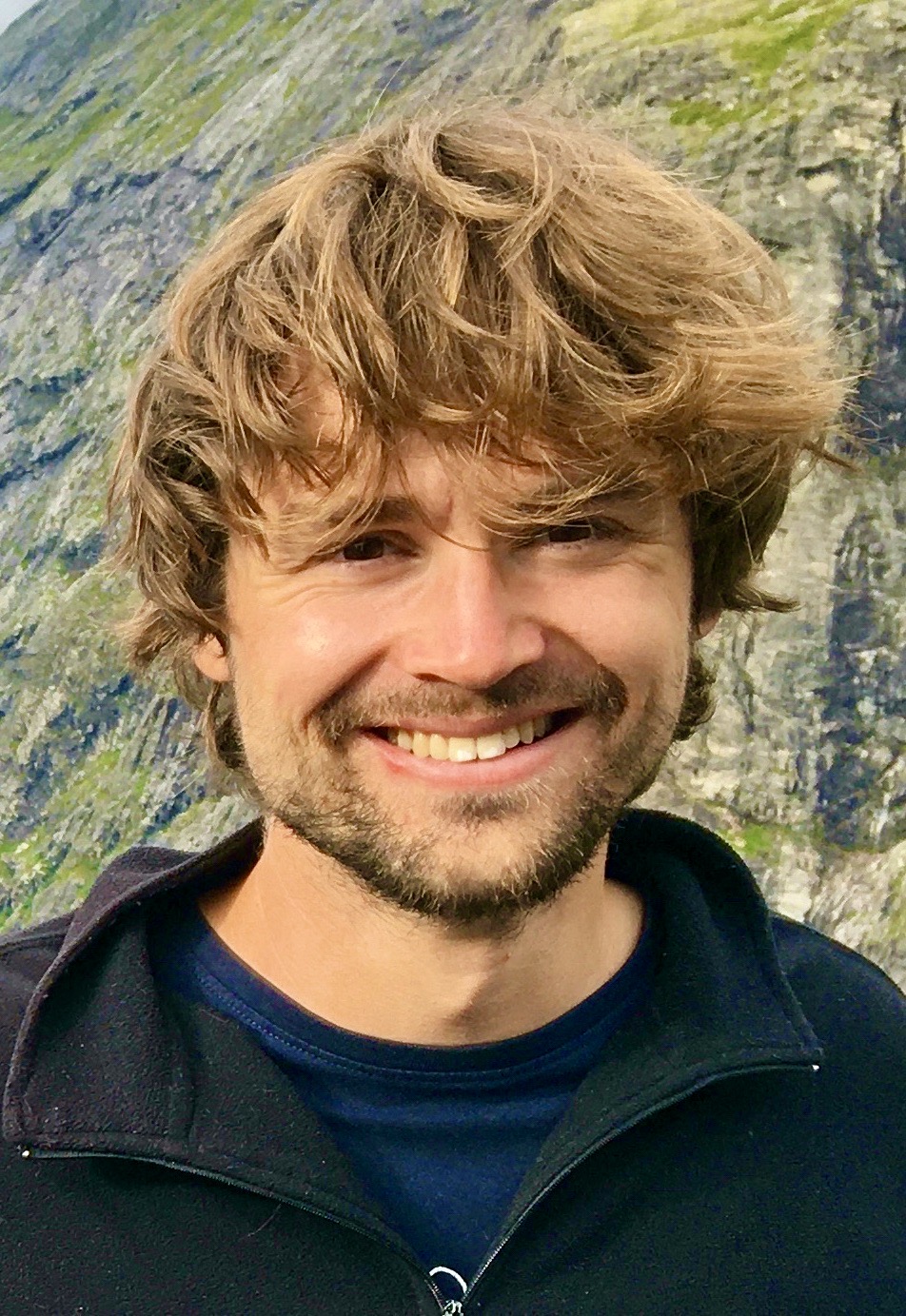}}]{Elis Stefansson} received a B.Sc. degree in Engineering Physics from KTH Royal Institute of Technology, Sweden in 2015 and a M.Sc. degree in Mathematics jointly from KTH Royal Institute of Technology and Stockholm University, Sweden in 2018. He has been a research visitor at Caltech and UC Berkeley and is currently pursuing a Ph.D. degree in Electrical Engineering and Computer Science at KTH Royal Institute of Technology. His research interests are in the intersection of control and complexity, as well as human-robot interactions, receiving an Outstanding Paper Prize at CDC 2021 and being a Best Student Paper Finalist at ECC 2019. 
\end{IEEEbiography}

\begin{IEEEbiography}[{\includegraphics[width=1in,height=1.25in,clip,keepaspectratio]{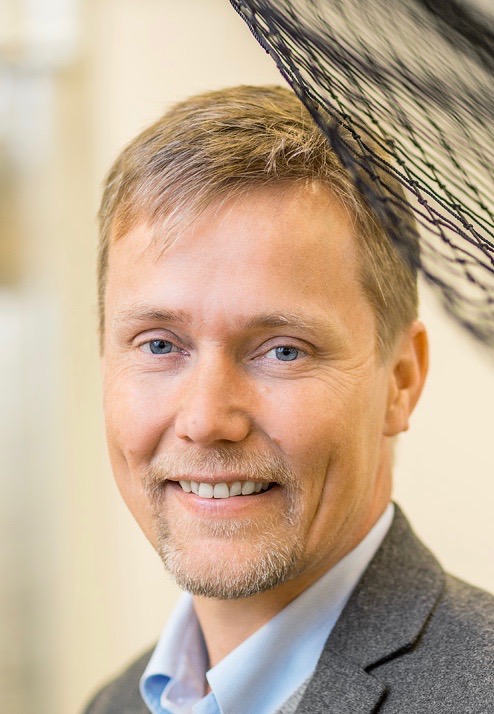}}]{Karl H. Johansson} is Swedish Research Council Distinguished Professor in Electrical Engineering and Computer Science at KTH Royal Institute of Technology in Sweden and Founding Director of Digital Futures. He earned his MSc degree in Electrical Engineering and PhD in Automatic Control from Lund University. He has held visiting positions at UC Berkeley, Caltech, NTU and other prestigious institutions. His research interests focus on networked control systems and cyber-physical systems with applications in transportation, energy, and automation networks. For his scientific contributions, he has received numerous best paper awards and various distinctions from IEEE, IFAC, and other organizations. He has been awarded Distinguished Professor by the Swedish Research Council, Wallenberg Scholar by the Knut and Alice Wallenberg Foundation, Future Research Leader by the Swedish Foundation for Strategic Research. He has also received the triennial IFAC Young Author Prize and IEEE CSS Distinguished Lecturer. He is the recipient of the 2024 IEEE CSS Hendrik W. Bode Lecture Prize. His extensive service to the academic community includes being President of the European Control Association, IEEE CSS Vice President Diversity, Outreach \& Development, and Member of IEEE CSS Board of Governors and IFAC Council. He has served on the editorial boards of Automatica, IEEE TAC, IEEE TCNS and many other journals. He has also been a member of the Swedish Scientific Council for Natural Sciences and Engineering Sciences. He is Fellow of both the IEEE and the Royal Swedish Academy of Engineering Sciences.
\end{IEEEbiography}

\if\longversion1

\section{Appendix}

\subsection{HiMMs are Modular}\label{appendix_himms_are_modular}
In this section, we provide the details to the modular analysis of HiMMs described in Section \ref{HiMM_are_modular}. For ease of reading, we start by restating the definition of a module of an MM (also given in Section \ref{HiMM_are_modular}), naturally extending the definition of the FSM case found in \cite{biggar2021modular}. To follow \cite{biggar2021modular} more closely, we here first define an MM module using entrances an exits (Definition \ref{def:module}), and then show that this definition is equivalent to the definition using the three operations contraction, restriction and expansion (Proposition \ref{prop:modular_equivalence}), that is, reversed presentation order compared to Section \ref{HiMM_are_modular}. We stress that this reversed order is immaterial since both definitions are equivalent. Finally, we prove Proposition \ref{prop:himms_are_modular} showing that HiMMs are modular.

To this end, note first that there is a natural correspondence between an MM $M$ and a labelled directed graph with arcs on the form $u \xrightarrow{a,o} v$, where $u$ and $v$ are states of $M$ and the arc $u \xrightarrow{a,o} v$ means that there is a $a$-transition from $u$ to $v$ with cost $o$. We also use the notation $u \xrightarrow{a} v$ to mean that there is a $a$-transition to $u$ and $v$ with some cost. 

For simplicity, following \cite{biggar2021modular}, we assume throughout this section that every MM $M$ is reachable in the sense that every state in $M$ can be reached from the start state of $M$. This assumption is quite mild since the potential non-reachable part of an MM can be removed to make it reachable.

\begin{definition}[Entrances and exits \cite{biggar2021modular}]
Consider a non-empty subset $S \subseteq Q_M$. For any $a \in \Sigma$, an $a$-entrance to $S$ is a state $v \in S$ such that $u \xrightarrow{a} v$ for some $u \notin S$, and an $a$-exit to $S$ is a state $u \notin S$ such that $v \xrightarrow{a} u$ for some $v \in S$. An entrance (exit) is any $a$-entrance ($a$-exit). 
\end{definition}
\begin{definition}[MM Modules]\label{def:module}
A non-empty subset $S \subseteq Q_M$ of states of an MM $M$ is a module if it has at most one entrance and for each $a \in \Sigma_M$, if $M$ has an $a$-exit then (i) that $a$-exit is unique (ii) every state in $S$ has an $a$-arc and (iii) all $a$-arcs out of $S$ have equal costs.
\end{definition}
\begin{remark}
The definition is identical to that of an FSM module \cite{biggar2021modular} except that we have imposed the additional requirement that all $a$-arcs out of $S$ should have equal costs.
\end{remark}
To motivate Definition~\ref{def:module}, we introduce three operations, contraction, restriction and expansion, and prove that Definition~\ref{def:module} is equivalent to an independence property capturing the intuitive notion of a module, formally defined through these operations (Proposition \ref{prop:modular_equivalence}). The analysis closely follows the FSM case found in \cite{biggar2021modular} with minor modifications to adapt to the MM case.

\begin{definition}[Contraction]
Let $M$ be an MM and $S \subseteq Q_M =:Q$. The contraction of $S$ in $M$, denoted $M/S$ is defined as
\begin{equation*}
M/S = ( Q \backslash S \cup \{ S \},  \Sigma, \delta', \gamma', s') 
\end{equation*}
where
\begin{equation*}
\delta'(q,a) =
\begin{cases}
\delta_{M}(q,a), & \textrm{if $q \in Q \backslash S$ and $\delta_{M}(q,a) \in Q \backslash S$}  \\
S, & \textrm{if $q \in Q \backslash S$ and $\delta_{M}(q,a) \in S$} \\
\delta_{M}(u,a), & \parbox[t]{.3\textwidth}{\textrm{if $q = S$ and \\ $\exists u$ s.t. $\emptyset \neq \delta_{M}(u,a) \in Q \backslash S$}} \\
\emptyset, & \textrm{otherwise},
\end{cases}
\end{equation*}
\begin{equation*}
\gamma'(q,a) =
\begin{cases}
\gamma_{M}(q,a), & \textrm{if $q \in Q \backslash S$}  \\
\gamma_{M}(u,a), & \parbox[t]{.3\textwidth}{\textrm{if $q = S$ and \\ $\exists u$ s.t. $\emptyset \neq \delta_{M}(u,a) \in Q \backslash S$}}
\end{cases}
\end{equation*}
and
\begin{equation*}
s' =
\begin{cases}
s_M, & \textrm{if $s_M \notin S$} \\
S, & \textrm{otherwise}.
\end{cases}
\end{equation*}
\end{definition}
Intuitively, a contraction of a set of states $S$ in an MM $M$ replaces all the states of $S$ with just one state (also called $S$) and draw arcs to and from this state $S$ if there exists corresponding arcs to and from the subset of states $S$. Just as in the FSM case \cite{biggar2021modular}, this is not always well-defined since $\delta'$ and $\gamma'$ may not be functions (not even partial functions) if $S$ has multiple $a$-arcs out of $S$. However, if every $a$-exit is unique, then $\delta'$ is well-defined. Furthermore, if all $a$-arcs out of $S$ has the same cost, then $\gamma'$ is well-defined. This justifies conditions (ii) and (iii) in Definition \ref{def:module}.

\begin{definition}[Restriction]
Let $M$ be an MM and $S \subseteq Q_M =: Q$. Considering $M$ as a directed graph, the restriction $M[S]$ is given by the subgraph of $M$ with nodes $S$. The start state of $M[S]$ is $s_M$ if $s_M \in S$, otherwise it is any entrance of $S$.
\end{definition}
In words, $M[S]$ is created by considering only the states $S$ cutting away all arcs to and from the remaining MM. Again, as in the FSM case \cite{biggar2021modular}, $M[S]$ is not always well-defined since there might be multiple entrances. However, with the restriction of at most one entrance, $M[S]$ is well-defined, motivating condition (i) in Definition \ref{def:module}.

\begin{definition}[Expansion]
Let $G$ and $H$ be MMs and let $v$ be a state of $G$. Then, the expansion of $H$ at $v$, written $G \cdot_v H$ is the MM
\begin{equation*}
\left ( Q_G \backslash \{v \} \cup Q_H, \Sigma_G \cup \Sigma_H, \delta', \gamma', s' \right ),
\end{equation*}
where
\begin{align*}
\delta'(q,a) =
\begin{cases}
\delta_G(q,a), & \textrm{if $q \in Q_G \backslash \{v \}$ and $\delta_G(q,a) \neq v$} \\
s_H, & \textrm{if $q \in Q_G \backslash \{v \}$ and $\delta_G(q,a) = v$}  \\
\delta_H(q,a), & \textrm{if $q \in Q_H$ and $\delta_H(q,a) \neq \emptyset$} \\
\delta_G(v,a), & \textrm{if $q \in Q_H$ and $\delta_H(q,a) = \emptyset$,}
\end{cases}
\end{align*}

\begin{align*}
\gamma'(q,a) =
\begin{cases}
\gamma_G(q,a), & \textrm{if $q \in Q_G \backslash \{v \}$} \\
\gamma_H(q,a), & \textrm{if $q \in Q_H$ and $\delta_H(q,a) \neq \emptyset$} \\
\gamma_G(v,a), & \textrm{if $q \in Q_H$ and $\delta_H(q,a) = \emptyset$}
\end{cases}
\end{align*}
and
\begin{equation*}
s' =
\begin{cases}
s_H, & v = s_G \\
s_G, & v \neq s_G.
\end{cases}
\end{equation*}
\end{definition}
The expansion operator is intuitively the opposite of the contraction operator. That is, we take a state $v$ in an MM $G$ and replace that state with an MM $H$. Any previous arc to $v$ is now going to the start state $s_H$ of $H$ instead, and any input $a$ in $H$ that is unsupported (from some state $q \in Q_H$) now transitions according to the $a$-transition from $v$ (intuitively speaking, is taken care of by $G$ higher up in the hierarchy). The costs are constructed similarly.

With these operations, we are in a position to prove the equivalent independent property characterisation of an MM~module.

\begin{proposition}\label{prop:modular_equivalence}
Let $M$ be an MM. A subset $S \subseteq Q_M$ is a module of $M$ if and only if $M/S \cdot_S M[S] = M$.
\end{proposition}

\begin{proof}
The proof is almost identical to the FSM case found in \cite{biggar2021modular}, with straightforward modifications to account for costs. We provide a proof where we occasionally refer to \cite{biggar2021modular} when the derivations are totally analogous to the FSM~case. 

Assume first that $S$ is a module. We want to prove that $M/S \cdot_S M[S] = M$. Note that $M/S \cdot_S M[S]$ and $M$ have the same set of states. Therefore, what remains to show is that $u \xrightarrow{a,o} v$ is an arc in $M/S \cdot_S M[S]$ if and only if $u \xrightarrow{a,o} v$ is an arc in $M$. We consider four cases: (i) $u,v \notin S$ (ii) $u \notin S$, $v \in S$ (iii) $u \in S$, $v \notin S$ and (iv) $u,v \in S$. We only provide the details to (iii) to show how one can adapt the proof from \cite{biggar2021modular} with minor modifications to account for costs. The other cases are totally analogous to the FSM proof and we refer to \cite{biggar2021modular} for details. 

To show (iii) assume $u \in S$, $v \notin S$. Consider an arc $u \xrightarrow{a,o} v$ in $M$. Then $u$ has no $a$-arc in $M[S]$ as $v \in S$. Since $S \xrightarrow{a,o} v$ is an arc in $Z/S$, we have that all states of $Z[M]$ in $M/S \cdot_S M[S]$ without $a$-arcs now have $a$-arcs to $v$ with cost $o$ (due to the definition of an expansion), so $u \xrightarrow{a,o} v$ is in $M/S \cdot_S M[S]$. Conversely, consider an arc $u \xrightarrow{a,o} v$ in $M/S \cdot_S M[S]$. Then $S \xrightarrow{a,o} v$ is an arc in $M/S$ and $u$ cannot have an arc in $M[S]$. Moreover, $S \xrightarrow{a,o} v$ implies that there exists $k \in S$ such that $k \xrightarrow{a,o} v$ is in $M$. Hence $S$ has an $a$-exit and by definition of a module, every state in $S$ either has an $a$-arc inside $S$ or an $a$-arc that goes to $v$ with cost $o$. Since $u$ does not have an $a$-arc in $M[S]$, we conclude that $u \xrightarrow{a,o} v$ is in $M$.

By the four cases, we conclude that $u \xrightarrow{a,o} v$ is an arc in $M/S \cdot_S M[S]$ if and only if $u \xrightarrow{a,o} v$ is an arc in $M$, and hence, $M/S \cdot_S M[S] = M$. This completes the first part of the proof.

For the second part of the proof, assume instead that $S$ is not a module. We want to show that $M/S \cdot_S M[S] \neq M$. There are four cases how $S$ can violate being a module: either (i) $S$ has two distinct entrances, or (ii) $S$ has two distinct $a$-exits, or (iii) $S$ has an $a$-exit but a state $u$ in $S$ has no $a$-arc, or (iv) there is a unique $a$-exit $v$ but two states $u_1,u_2 \in S$ such that  $u_1 \xrightarrow{a,o_1} v$ and $u_2 \xrightarrow{a,o_1} v$ with different costs $o_1 \neq o_2$. Cases (i)-(iii) are treated totally analogously to the FSM case in \cite{biggar2021modular}. We therefore only consider case (iv). In case (iv), note simply that the cost function $\gamma'$ of the expansion $M/S \cdot_S M[S] \neq M$ is not even well-defined if (iv) is true and therefore $M/S \cdot_S M[S] \neq M$. We conclude that if $S$ is not a module then $M/S \cdot_S M[S] \neq M$.

The first part and the second part of the proof gives the desired equivalence. The proof is therefore complete.

\end{proof}

Proposition \ref{prop:modular_equivalence} intuitively formalises a module as an independent nested submachine of the considered machine and can be seen as an alternative definition of a module. Furthermore, this alternative definition can also be used to define modules for FSM and graphs (with appropriate contraction, restriction and expansion), see \cite{biggar2021modular} for details, thereby forming a more universal characterisation of what a module is. However, from a practical standpoint, Definition \ref{def:module} is more suitable.

Finally, we define modules of an HiMM $Z$ to be modules of the equivalent flat HiMM $Z^F$.\footnote{The flat HiMM $Z^F$ of an HiMM $Z$ is the MM given by $Z^F = (Q,\Sigma,\delta,\gamma,s) := (S_Z,\Sigma,\psi,\chi,\mathrm{start}(M_0))$ (with notation as in Definition \ref{HiMM_def}). Intuitively, $Z^F$ is the MM we get by flattening the whole hierarchy of machines into one equivalent MM.} Note that $Z^F$ is by definition an MM.

\begin{definition}[HiMM module]\label{HiMM_module}
Let $Z=(X,T)$ be an HiMM. Consider a subset $Q \subseteq Q_Z$ of augmented states and let $S \subseteq S_Z$ be all states of $Z$ nested within $Q$. Then $Q$ is a module of $Z$ if $S$ is a module of the flatted HiMM $Z^F$. Furthermore, we say that an MM $M \in X$ is a module of $Z$ if $Q_M$ is a module of $Z$.
\end{definition}

\begin{remark}
In \cite{biggar2021modular}, a module of an HFSM $Z$ without costs is defined as a subset of HFSM states $S$ such that there exists a module $Q \subseteq Q_M$ to some FSM $M$ (of $Z$) whose nested states in $Z$ equals $S$. Therefore, in this definition, all FSMs are trivially modules since the whole state set of an FSM is always a module. Definition \ref{HiMM_module} is slightly different since it says that a subset of HiMM states is a module if it is a module in the corresponding flat HiMM $Z^F$. That is, the HiMM states are not restricted to be the nested states of just one MM. We adopt this later definition to get a stronger (non-trivial) result that each MM of an HiMM is modular.
\end{remark}

We are now in a position to prove Proposition \ref{prop:himms_are_modular}.

\begin{proof}[Proof of Proposition \ref{prop:himms_are_modular}]
Let $Z=(X,T)$ be an HiMM and $M \in X$. Let $S$ be the HiMM states recursively nested within $Q_M$. We want to prove that $S$ is a module of $Z^F$. We first note that $S$ can have at most one entrance given by $\mathrm{start}(M)$. Assume $v$ is an $a$-exit of $S$. Let $u$ be any state $u \in S$ such that $u \xrightarrow{a} v$. By definition, $\psi(u,a) = v$. Moreover, considering the chain of nested MMs
\begin{equation*}
\emptyset \xleftarrow{u} K_{1} \xleftarrow{k_{1}} \dots \xleftarrow{k_{n-1}} K_n \xleftarrow{k_{n}} M.
\end{equation*} 
We have that $\delta_{K_1}(u,a) = \delta_{K_i}(k_i,a) = \delta_{M}(k_n,a) = \emptyset$ since we exit $M$. Thus, $\psi(u,a) = \psi(m,a)$ where $m$ is such that $M \xleftarrow{m} W$ for some $W \in X$ (note that such $m$ and $W$ must exist since otherwise $v = \psi(u,a) = \emptyset$). We conclude that $v = \psi(m,a)$ so the $a$-exit is unique. Similarly, $\chi(u,a) = \chi(m,a)$ so all $a$-arcs out of $S$ (to $v = \psi(m,a)$) have equal costs. Finally, every state of $S$ has an $a$-arc by construction. To see it, consider an arbitrary $u \in S$ and its corresponding chain of nested MMs as above. If $\delta_{K_1}(u,a) = \delta_{K_i}(k_i,a) = \delta_{M}(k_n,a) = \emptyset$, then again $\psi(u,a) = \psi(m,a)$ (where $m$ is identical to before). Otherwise, there exists some transition further down in the chain, so we transition to a state in $S$. In both cases, $u$ has an $a$-arc. This completes the proof.
\end{proof}

\subsection{HiMM with Identical Machines}
In this section, we complement Definition \ref{HiMM_duplicates_def} in Section \ref{duplicates} with a more complete definition given by Definition \ref{HiMM_duplicates_def_detailed}, filling out the details of the start function, hierarchical transition function and cost~function.

\begin{definition}\label{HiMM_duplicates_def_detailed}
An HiMM with identical machines is a tuple $Z = (X,T)$ just as in Definition \ref{HiMM_def} except that multiple arcs in $T$ can now go to the same node, i.e., $T$ is not a tree anymore but instead a DAG with a unique root. Furthermore, a state of $Z$ is a directed path of arcs $s = (s_1,\dots,s_d)$ from the root of $T$ to a leaf $s_d$, and an augmented state is any subpath $(s_1,\dots,s_i)$ of a state $s = (s_1,\dots,s_d)$ ($i \leq d$). Moreover, the start function, hierarchical transition and cost function are defined as in Definition \ref{HiMM_def} with straightforward modifications to account for the modified state form:
\begin{enumerate}[(i)]
\item Start function: Given an auxiliary state $(s_1,\dots,s_k)$ of $Z$ and an MM $M$ such that $s_k$ corresponds to $M$, we define the start function by
\begin{align*}
\mathrm{start}((s_1,\dots,s_k),M) = \\
\begin{cases}
\mathrm{start}((s_1,\dots,s_k,s_M),N), &  \textrm{$(M \xrightarrow{s_M} N) \in T$} \\
(s_1,\dots,s_k,s_M), & \textrm{otherwise.}
\end{cases}
\end{align*}
\item Hierarchical transition function: Let $(s_1,\dots,s_k)$ be an auxiliary state of $Z$ with $s_k \in Q_M$, and denote $v =\delta_{M}(s_k,a)$. The hierarchical transition function $\psi: Q_Z \times \Sigma \rightharpoonup S_Z$ is given by 
\begin{align*}
\psi((s_1,\dots,s_k),a) = \hspace{2.1cm} \\
\begin{cases} 
\mathrm{start}((s_1,\dots,s_{k-1},v),Y), & \parbox[t]{.3\textwidth}{$v \neq \emptyset$, \\ $(M \xrightarrow{v} Y) \in T, Y \in X$} \\ 
(s_1,\dots,s_{k-1},v), & v \neq \emptyset, \mathrm{otherwise} \\
\psi((s_1,\dots,s_{k-1}),a), & v = \emptyset, k \geq 2 \\ 
\emptyset, & v = \emptyset, \mathrm{otherwise.}
\end{cases}
\end{align*}
\item Hierarchical cost function: Let $(s_1,\dots,s_k)$ be an auxiliary state of $Z$ with $s_k \in Q_M$, and denote $v =\delta_{M}(s_k,a)$. The hierarchical cost function $\chi: Q_Z \times \Sigma \rightharpoonup \mathbb{R}^+$ is given~by
\begin{align*}
\chi((s_1,\dots,s_k),a) = \\
\begin{cases} 
\gamma_{M}(s_k,a), & v \neq \emptyset \\
\chi((s_1,\dots,s_{k-1}),a), & v = \emptyset, k \geq 2 \\ 
\emptyset, & v = \emptyset, \mathrm{otherwise.}
\end{cases}
\end{align*}
\end{enumerate}
\end{definition}

\subsection{Optimal Exit Computer: Additional Details}
In this section, we provide the proofs of each statement in Section \ref{optimal_exit_computer}.

\begin{proof}[Proof of Proposition \ref{th:compute_optimal_exits_correct}]

We prove the proposition by induction over the tree $T$. Towards this, consider any MM $M$ of $Z$ and assume each $q \in Q_M$ has correct optimal exit costs $c_a^q$. That is, $c_a^q = c_a^{N_q}$ if $q$ corresponds to an MM $N_q$ and $c_a^{q} = 0$ otherwise. Note that this assumption is trivially true for all leaf MMs of $T$ (the base case of the induction). Let $\tilde{c}_a^M$ be the optimal cost to $E_a$ in $\hat{M}$ computed by Algorithm \ref{alg:compute_optimal_exits}. We want to show that $\tilde{c}_a^M = c_a^M$ (the induction step), since then by induction over $T$, Algorithm \ref{alg:compute_optimal_exits} computes the correct optimal exit costs for all MM. It thus remains to prove that $\tilde{c}_a^M = c_a^M$. We divide the proof into two cases.

In the first case, $c_a^M = \infty$. In this case, no $a$-exit trajectory exists. Assume by contradiction that $\tilde{c}_a^M < \infty$ and let $z_a^M = (q_i,a_i)_{i=1}^N$ be the corresponding trajectory in $\hat{M}$. Note that each $c_{a_i}^{q_i}$ must be finite, hence, there exists an optimal $a_i$-exit trajectory $w_i$ of the subtree $q_i$ corresponds to. Then, $w = (w_i)_{i=1}^N$ is an $a$-exit trajectory of $M$, a contradiction since $c_a^M = \infty$. We conclude that $\tilde{c}_a^M = \infty$.

In the second case, $c_a^M < \infty$. Let $w$ be an optimal $a$-exit trajectory of $M$. Note that $w$ can be partitioned into $w = (w_i)_{i=1}^N$ such that each $w_i$ is an $a_i$-exit trajectory of some $q_i \in Q_M$. By optimality, the exit cost of $w$ is $\sum_{i=1}^{N-1} [c_{a_i}^{q_i}+\gamma_M(q_i,a_i)]+c_{a_N}^{q_N}$ (since each $w_i$ must be an optimal $a_i$-exit trajectory). Thus, $c_a^M = \sum_{i=1}^{N-1} [c_{a_i}^{q_i}+\gamma_M(q_i,a_i)]+c_{a_N}^{q_N}$. Consider now trajectory $z = (q_i,a_i)_{i=1}^N$ in $\hat{M}$. By construction, $z$ has cumulative cost $\sum_{i=1}^{N-1} [c_{a_i}^{q_i}+\gamma_M(q_i,a_i)]+c_{a_N}^{q_N} = c_a^M$ and hence $\tilde{c}_a^M \leq c_a^M$. To prove equality, pick any trajectory $z = (q_i,a_i)_{i=1}^N$ from $s_M$ to $E_a$ in $\hat{M}$ with cumulative cost $\tilde{c}$. Note that $\tilde{c} = \sum_{i=1}^{N-1} [c_{a_i}^{q_i}+\gamma_M(q_i,a_i)]+c_{a_N}^{q_N}$. This cumulative cost coincides with the $a$-exit cost of $w = (w_i)_{i=1}^N$ where $w_i$ is an $a_i$-exit trajectory of $q_i \in Q_M$. Since $w$ is an $a$-exit trajectory of $M$, we have that $c_a^M \leq \tilde{c}$. Since $z$ was arbitrary, $c_a^M \leq \tilde{c}_a^M$, and hence, $\tilde{c}_a^M = c_a^M$. This completes the proof.

\end{proof}

\begin{proof}[Proof of Proposition \ref{th:compute_optimal_exits_time}]

We construct $\hat{M}$ from $M$ by adding the states $\{E_a\}_{a \in \Sigma}$ in time $O(|\Sigma |)$, constructing $\hat{\delta}$ by going through all the values of the function in time $O(b_s | \Sigma |)$, and constructing  $\hat{\gamma}$ analogously. The total time complexity for line \ref{alg:compute_optimal_exits:line11} in Algorithm \ref{alg:compute_optimal_exits} is therefore $O(b_s | \Sigma |)$. Furthermore, note that the maximum number of states of $\hat{M}$ is $b_s+|\Sigma|$ and the maximum number of transition arcs of $\hat{M}$ is $b_s |\Sigma|$ (considering $\hat{M}$ as a graph). Searching with Dijkstra's algorithm \cite{DijkstraFibonacci}, line \ref{alg:compute_optimal_exits:line12} in Algorithm \ref{alg:compute_optimal_exits}, has therefore complexity $O(E+V \log V) = O(b_s |\Sigma|+(b_s+|\Sigma|) \log(b_s+|\Sigma|))$ (where $E$ is the number of arcs and $V$ is the number of nodes in the graph used in Dijkstra's algorithm). Finally, line \ref{alg:compute_optimal_exits:line3}--\ref{alg:compute_optimal_exits:line8} takes time $O(b_s |\Sigma|)$ excluding the time it takes to compute $\mathrm{Optimal\_exit}(M_q)$ since that time is already accounted for (when considering the MM $M_q$). We conclude that the total time spent on one MM $M$ is $O(b_s |\Sigma|+(b_s+|\Sigma|) \log(b_s+|\Sigma|))$. This is done for all MMs in $Z$, hence, the time complexity for Algorithm~\ref{alg:compute_optimal_exits} is
\begin{align}\label{eq:time_complexity_1}
O(|X| \cdot [b_s |\Sigma|+(b_s+|\Sigma|) \log(b_s+|\Sigma|)]).
\end{align}
This completes the proof. 

\end{proof}

\subsection{Optimal Planner: Additional Details}
In this section, we provide additional details concerning Section \ref{optimal_planner}, given in chronological order with respect to Section \ref{optimal_planner}, including proofs to all statements.

\subsubsection{Reduce: Algorithm}
We start by providing the algorithm for computing the reduced HiMM $\bar{Z}$, mentioned in Remark \ref{remark_reduce}. This algorithm is given by Algorithm \ref{alg:reduce}, where the subroutine $\mathrm{Compute\_paths}$, given by Algorithm \ref{alg:compute_alpha_and_beta}, computes the sequences $U_1,\dots,U_n$ and $D_1,\dots,D_m$ (called $U_{path}$ and $D_{path}$ respectively) and $\alpha$ and $\beta$. Here, $\alpha$ is analogous to $\beta$ but for $\bar{U}_i$, that is, $\bar{U}_\alpha$ is the last $\bar{U}_i$ before the $\bar{U}_i$ and $\bar{D}_i$ sequences becomes identical. Algorithm \ref{alg:reduce} and \ref{alg:compute_alpha_and_beta} are important when proving the time complexity of the Online~Planner.

\begin{algorithm}[t]
\caption{Reduce}\label{alg:reduce}
\begin{algorithmic}[1]
\Require $(Z,s_{\mathrm{init}},s_{\mathrm{goal}})$ and $(c_a^M,z_a^M)_{a \in \Sigma, M \in X})$
\Ensure $(\bar{Z},U_{path},D_{path},\alpha,\beta)$
\State $(U_{path}, D_{path}, \alpha, \beta) \gets$ Compute\_paths($Z,s_\mathrm{init},s_\mathrm{goal}$) \label{alg:reduce:line1}
\State Let $\bar{X}$ be an empty set \Comment{To be constructed} \label{alg:reduce:line2}
\State Let $\bar{T}$ be an empty tree \Comment{To be constructed} \label{alg:reduce:line3}
\State Construct $\bar{U}_1$ from $U_1$ and add to $\bar{X}$ \label{alg:reduce:line4}
\State Add node $\bar{U}_1$ to $\bar{T}$ with arcs out as $U_1$ in $T$ but no destination nodes \label{alg:reduce:line5}
\State Save $q$ where $(U_2 \xrightarrow{q} U_1) \in T$\label{alg:reduce:line6}
\For {$i = 2,\dots,\mathrm{length}(U_{path})-1$ \label{alg:reduce:line7}}
\State Construct $\bar{U}_i$ from $U_i$ and add to $\bar{X}$ \label{alg:reduce:line8}
\State Add node $\bar{U}_i$ to $\bar{T}$ with arcs out as $U_i$ in $T$ but no destination nodes \label{alg:reduce:line9}
\State Add destination node $\bar{U}_{i-1}$ to $\bar{U}_{i} \xrightarrow{q} \emptyset$ so that $(\bar{U}_{i} \xrightarrow{q}\bar{U}_{i-1}) \in \bar{T}$ \label{alg:reduce:line10}
\If{$i < \mathrm{length}(U_{path})-1$\label{alg:reduce:line11}}
\State Save $q$ where $U_{i+1} \xrightarrow{q} U_i \in T$ \label{alg:reduce:line12}
\EndIf
\EndFor
\If {$\beta = 0$\label{alg:reduce:line15}} \Comment{Then nothing more to do}
\State {return} $(\bar{Z} = (\bar{X},\bar{T}), U_{path})$ \label{alg:reduce:line16}
\EndIf 
\State Construct $\bar{D}_1$ from $D_1$ and add to $\bar{X}$ \label{alg:reduce:line18}
\State Add node $\bar{D}_1$ to $\bar{T}$ with arcs out as $D_1$ in $T$ but no destination nodes
\State Save $q$ where $(D_2 \xrightarrow{q} D_1) \in T$
\For {$i = 2,\dots,\beta$}
\State Construct $\bar{D}_i$ from $D_i$ and add to $\bar{X}$
\State Add node $\bar{D}_i$ to $\bar{T}$ with arcs out as $D_i$ in $T$ but no destination nodes
\State Add destination node $\bar{D}_{i-1}$ to $\bar{D}_{i} \xrightarrow{q} \emptyset$ so that $(\bar{D}_{i} \xrightarrow{q}\bar{D}_{i-1}) \in \bar{T}$
\State Save $q$ where $(D_{i+1} \xrightarrow{q} D_i) \in T$
\EndFor
\State Add destination node $\bar{D}_{\beta}$ to $\bar{U}_{\alpha+1} \xrightarrow{q} \emptyset$ so that $(\bar{U}_{\alpha+1} \xrightarrow{q} \bar{D}_{\beta}) \in \bar{T}$ \label{alg:reduce:line27}
\State {return} $(\bar{Z} = (\bar{X},\bar{T}),U_{path})$
\end{algorithmic}
\end{algorithm}

\begin{algorithm}
\caption{Compute\_paths}\label{alg:compute_alpha_and_beta}
\begin{algorithmic}[1]
\Require $(Z,s_\mathrm{init},s_\mathrm{goal})$
\Ensure $(U_\mathrm{path}, D_\mathrm{path}, \beta)$
\State $U_\mathrm{path} \gets []$ \Comment{Compute $U_\mathrm{path}$}
\State $U_0 \gets s_\mathrm{init}$
\State $U_\mathrm{current} \gets U_0$ 
\State Append $U_0$ to $U_{path}$
\While {$U_\mathrm{current}.\mathrm{parent} \neq \emptyset$}
\State Append $U_\mathrm{current}.\mathrm{parent}$ to $U_\mathrm{path}$
\State $U_\mathrm{current} \gets U_\mathrm{current}.\mathrm{parent}$
\EndWhile
\State $D_\mathrm{path} \gets []$ \Comment{Compute $D_\mathrm{path}$}
\State $D_0 \gets s_\mathrm{goal}$
\State $D_\mathrm{current} \gets D_0$
\State Append $D_0$ to $D_\mathrm{path}$
\While {$D_\mathrm{current}.\mathrm{parent} \neq \emptyset$}
\State Append $D_\mathrm{current}.\mathrm{parent}$ to $D_\mathrm{path}$
\State $D_\mathrm{current} \gets D_\mathrm{current}.\mathrm{parent}$
\EndWhile
\State $(i,j) \gets (\mathrm{length}(U_\mathrm{path}),\mathrm{length}(D_\mathrm{path}))$ \Comment{Compute $\alpha$ and $\beta$}
\While {$U_\mathrm{path}[i] = D_\mathrm{path}[j]$} 
\State $(i,j) \gets (i-1,j-1)$
\EndWhile
\State $(\alpha, \beta) \gets (i-1, j-1)$
\State {return} $(U_\mathrm{path}, D_\mathrm{path}, \alpha, \beta)$
\end{algorithmic}
\end{algorithm}

\subsubsection{Reduced Trajectory and Optimal Expansion}

We continue with the complete definitions of a reduced trajectory $z_R$ and an optimal expansion $z_E$ outlined in Section \ref{reduce_step}. To this end, consider an HiMM $Z$ and its reduced HiMM $\bar{Z}$ with respect to some initial state $s_{\mathrm{init}}$ and goal state $s_{\mathrm{goal}}$. The reduced trajectory $z_R$ of a trajectory $z$ of $Z$ is then given by Algorithm \ref{alg:reduction}. It is easy to check that $z_R$ is a trajectory of~$\bar{Z}$ (except one special case when $z_R = \emptyset$):
\begin{proposition}\label{prop:reduced_is_traj}
The reduced trajectory $z_R$ given by Algorithm \ref{alg:reduction} is a trajectory of $\bar{Z}$, except in one special case. The special case is when $z$ is always contained in one subtree $H$ that gets replaced by a state $p$ in $\bar{Z}$ and $z$ never exits it. In this case, $z_R = \emptyset$.
\end{proposition}
\begin{proof}


Let $z = (q_i,a_i)_{i=1}^N$ be a given trajectory of $Z$. Note first that $z$ can be partitioned into $z = z_1 \dots z_L$ where each subtrajectory $z_i$ is either a state-input pair $(p_i,b_i)$ of $\bar{Z}$, or $z_i$ is contained in a subtree $H_i$ in $Z$ that gets replaced by a state $p_i$ in $\bar{Z}$. We may assume that $p_i \neq p_{i+1}$ for all $i$ (otherwise, just concatenate $z_i$ and $z_{i+1}$). We claim that 

\begin{equation}\label{eq:z_R_equation}
z_R = 
\begin{cases}
(p_i,b_i)_{i=1}^L & \textrm{if $z_L$ exits $p_L$}\\
(p_i,b_i)_{i=1}^{L-1} & \textrm{if $z_L$ does not exit $p_L$ and $L>1$} \\
\emptyset & \textrm{otherwise,}
\end{cases}
\end{equation}
where $b_i$ is the last input of $z_i$.\footnote{Here, we say that $z_L$ exits $p_L$ if either $z_L = (p_L,b_L)$ or $z_L$ exits the subtree $H_i$ that gets reduced to $p_L$ in $\bar{Z}$.} To see the claim, consider $z_1$. If $z_1 = (p_1,b_1)$ is a state-input pair of $\bar{Z}$, then by Algorithm \ref{alg:reduction} (line \ref{alg:reduction:line4}--\ref{alg:reduction:line6}), the first state-input pair of $z_R$ is $(p_1,b_1)$. Otherwise, $z_1$ is contained in a subtree $H_1$. If $z_1$ never exits $H_1$ (hence, it never exits $p_1$), then Algorithm \ref{alg:reduction} returns $z_R = []$ (see line \ref{alg:reduction:line7}--\ref{alg:reduction:line13}) which corresponds to $z_R = \emptyset$. Moreover, there can not be any further subtrajectories $z_i$ since $z_1$ would then have to exit $H_1$, hence, $L=1$. This special case corresponds to the third row of Equation \eqref{eq:z_R_equation}, and is exactly the special case given by Proposition \ref{prop:reduced_is_traj}. If this is not the case, then $z_1$ exits $H_1$, hence, by Algorithm \ref{alg:reduction} (line \ref{alg:reduction:line7}--\ref{alg:reduction:line13}), $z_1$ will be replaced by $(p_1,b_1)$. Since Algorithm \ref{alg:reduction} iterates sequentially, we get that $z_R = (p_i,b_i)_{i=1}^L$ if $z_L$ exits $p_L$, or $z_R = (p_i,b_i)_{i=1}^{L-1}$ if $z_L$ does not exit $p_L$ (where the last $(p_L,b_L)$ is not included in the latter case due to a similar reasoning as in the special~case).  

It remains to show that $z_R$ is a trajectory of $\bar{Z}$, i.e., that $\bar{\psi}(p_i,b_i) = p_{i+1}$ for $i < L$. To see it, consider the last state-input pair $(e_i,b_i)$ of $z_i$ and the first state $s_{i+1}$ of $z_{i+1}$. Since $z$ is a trajectory, we have that $\psi(e_i,b_i) = s_{i+1}$. Moreover, $\psi(e_i,b_i) = \psi(p_i,b_i)$ since $z_i$ exits $p_i$, so $\psi(p_i,b_i) = s_{i+1}$. Note now that the recursive evaluation of $\bar{\psi}(p_i,b_i)$ is identical to $\psi(p_i,b_i)$ except that $\bar{\psi}(p_i,b_i)$ terminates earlier at the $\bar{Z}$-state containing $\psi(p_i,b_i)$, in this case $p_{i+1}$ since $\psi(p_i,b_i) = s_{i+1}$. Hence, $\bar{\psi}(p_i,b_i) = p_{i+1}$ This completes the~proof.


\end{proof}

\begin{algorithm}[t]
\caption{Reduced\_trajectory}\label{alg:reduction}
\begin{algorithmic}[1]
\Require HiMM$ \; Z$, reduced HiMM $\bar{Z}$ with respect to $s_{\mathrm{init}}$ and $s_{\mathrm{goal}}$, and trajectory $z = (q_i,a_i)_{i=1}^N$ of $Z$.
\Ensure A trajectory $z_R$ of $\bar{Z}$, or $z_R = []$ (corresponding to $z_R = \emptyset$) 
\State $z_R \gets []$
\State $i \gets 1$
\While{$i \leq N$} \Comment{Iterate over state-action pairs of $z$}
\If{$q_i$ is a state of $\bar{Z}$ \label{alg:reduction:line4}}
\State Append $(q_i,a_i)$ to $z_R$
\State $i \gets i+1$ \label{alg:reduction:line6}
\Else \label{alg:reduction:line7}
\State $p \gets$ state in $\bar{Z}$ replacing subtree $H$ containing $q_i$ in $Z$.
\While{$i \leq N$ and $\psi(q_{i},a_i)$ is a state of $H$}
\State $i \gets i+1$
\EndWhile
\If{$\psi(q_{i-1},a_{i-1})$ is not a state of $H$}
\State Append $(p,a_{i-1})$ to $z_R$ \Comment{Exit $H$ means exit $p$ in $\bar{Z}$} \label{alg:reduction:line13}
\EndIf
\EndIf
\EndWhile
\State return $z_R$
\end{algorithmic}
\end{algorithm}

We now continue with the optimal expansion. The optimal expansion $z_E$ of a trajectory $z$ of $\bar{Z}$ is given by Algorithm \ref{alg:optimal_expansion}. Similarly, $z_E$ is a trajectory of $Z$:
\begin{proposition}\label{prop:optimal_expansion_is_trajectory}
The optimal expansion $z_E$ given by Algorithm \ref{alg:optimal_expansion} is a trajectory of $Z$. Moreover, an optimal expansion of a single augmented state-input pair $(q,a)$ is an optimal $a$-exit trajectory of the MM $N_q$ corresponding to $q$.
\end{proposition}

\begin{algorithm}[t]
\caption{Optimal\_expansion}\label{alg:optimal_expansion}
\begin{algorithmic}[1]
\Require HiMM$ \; Z =(X,T)$, $\{z^{M}_a\}^{M \in X}_{a \in \Sigma}$ from Optimal Exit Computer, and trajectory $z = (q_i,a_i)_{i=1}^N$ of $\bar{Z}$.
\Ensure A trajectory $z_E$ of $Z$
\For{$(q_i,a_i)$ in $z$}
\State $z_i \gets \mathrm{Traj\_expand}(q_i,a_i)$
\EndFor
\State return $z_E \gets (z_1,\dots,z_N)$
\State 
\State $\mathrm{Traj\_expand}(q,a)$: \Comment{Expand recursively} 
\If {$q \in S_Z$}
\State return $(q,a)$. \label{alg:optimal_expansion:line8}
\Else
\For {$(q_i,a_i)$ in $z^{N_q}_a$} \Comment{$z^{N_q}_a$ from Optimal Exit Computer}
\State $z_i  \gets  \mathrm{Traj\_expand}(q_i,a_i)$
\EndFor
\State return $(z_1, \dots, z_m)$ \Comment{$m$ is the length of $z^{N_q}_a$}
\EndIf 
\end{algorithmic}
\end{algorithm}

We need the second fact when proving the planning equivalence.

\begin{proof}
To have a more compact terminology in the proof, we say that a trajectory $z$ is an optimal $a$-exit trajectory of $q \in Q_M$ (in some MM $M$ in $Z$) if either $z =(q,a)$ if $q$ is a state of $Z$, or $z$ is an optimal $a$-exit trajectory of $N_q$, where $N_q$ is the MM corresponding to $q$.

We start by the proving the second statement. Let an augmented state-input pair $(q,a)$ be given.
We show that $(q,a)$ is an optimal $a$-exit trajectory of $q$, using induction. In the base case, $q$ is a state of $Z$. In this case, $(q,a)$ is trivially an optimal $a$-exit trajectory of $q$. In the induction step, assume $q$ is not a state. Then, by Algorithm \ref{alg:optimal_expansion}, we have
\begin{equation*}
(q,a)_E = ( (q_i,a_i)_E )_{i=1}^m
\end{equation*}
where each $q_i$ is a state of the MM $N_q$ corresponding to $q$. By induction, each $(q_i,a_i)_E$ is an optimal $a_i$-exit trajectory of $q_i$. Also, the trajectory $(q_i,a_i)_{i=1}^m$ is an optimal trajectory to $E_a$ in $\hat{N}_q$. With these two facts, it is easy to see that $(q,a)_E$ is an $a$-exit trajectory of $q$: $(q_i,a_i)_E$ first exits $q_i$, then goes to $\delta_{N_q}(q_i,a_i) = q_{i+1}$ and finally follows the start states from $q_{i+1}$ arriving at the first state of $(q_{i+1},a_{i+1})_E$; and, the last $(q_i,a_i)_E$ exits $N_q$ since $(q_i,a_i)$ goes to $E_a$ in $\hat{N}_q$. It is also easy to deduce that $(q,a)_E$ is an optimal $a$-exit trajectory. To see it, assume that $z$ is an $a$-exit trajectory with strictly less exit cost $\tilde{c}$ than the exit cost $c$ of $(q_i,a_i)_E$. Partition $z$ into $z = z_1 \dots z_k$, where each $z_i$ is an $b_i$-exit trajectory of some state $s_i$ of $N_q$. We may assume that each $z_i$ is an optimal $b_i$-exit trajectory of $s_i$ (otherwise, change to an optimal one with even less resulting exit cost for $z$). Consider the trajectory $w = (s_i,b_i)_{i=1}^k$ in $\hat{M}$. By construction, $w$ reaches $E_a$ in $\hat{N}_q$ with the same cumulated cost $\tilde{c}$ as the exit cost of $z$. Furthermore, $(q_i,a_i)_{i=1}^m$ reaches $E_a$ with same cumulated cost $c$ as the exit cost of $(q,a)_E$. Thus, $(q_i,a_i)_{i=1}^m$ is not an optimal trajectory to $E_a$ in $\hat{N}_q$, which is a contradiction. Therefore, $(q,a)_E$ is an optimal $a$-exit trajectory.

We now prove the first statement. Let $z = (q_i,a_i)_{i=1}^L$ be a trajectory of $\bar{Z}$. We want to show that $z_E = ((q_i,a_i)_E)_{i=1}^L$ is a trajectory of $Z$. We already know that each $(q_i,a_i)_E$ is a trajectory of $Z$ so we only need to show that $\psi(e_i,b_i) = s_{i+1}$ for $i < L$, where $(e_i,b_i)$ is the last state-input pair of $(q_i,a_i)_E$ and $s_{i+1}$ is the first state of $(q_{i+1},a_{i+1})_E$. To show it, we know that $(q_i,a_i)_E$ is an $a_i$-exit trajectory of $q_i$, so $b_i = a_i$ and $\psi(e_i,b_i) = \psi(q_i,a_i)$. Moreover, since $\bar{\psi}(q_i,a_i) = q_{i+1}$, we have that $\psi(q_i,a_i) = \mathrm{start}(q_{i+1})$ where $\mathrm{start}(q_{i+1}) = q_{i+1}$ if $q_{i+1}$ is a state of $Z$, otherwise $\mathrm{start}(q_{i+1}) = \mathrm{start}(N_{q_{i+1}})$ where $N_{q_{i+1}}$ is the MM corresponding to $q_{i+1}$. In both cases $\psi(q_i,a_i) = \mathrm{start}(q_{i+1}) = s_{i+1}$ since $(q_{i+1},a_{i+1})_E$ is an $a_{i+1}$-exit trajectory of $q_{i+1}$. We conclude that $\psi(e_i,b_i) = s_{i+1}$ for $i < L$. This proves the first statement and thus completes the~proof.
\end{proof}
\subsubsection{Planning Equivalence} In this section, we prove the planning equivalence given by Theorem~\ref{th:planning_equivalence}. For a more compact terminology, we say that a trajectory $z = (q_i,a_i)_{i=1}^N$ is a \emph{feasible trajectory} to $({Z},s_{\mathrm{init}},s_{\mathrm{goal}})$ if it is a trajectory of $Z$ such that $q_1 = s_{\mathrm{init}}$ and $\psi(q_N,a_N) = s_{\mathrm{goal}}$. A feasible trajectory to $(\bar{Z},s_{\mathrm{init}},s_{\mathrm{goal}})$ is analogously defined.

We break up the proof of Theorem~\ref{th:planning_equivalence} by first proving two lemmas:

\begin{lemma}\label{lemma:planning_equivalence_1}
Let $z$ be an optimal trajectory to $({Z},s_{\mathrm{init}},s_{\mathrm{goal}})$. Then $z_R = (p_i,b_i)_{i=1}^L$ is a feasible trajectory to $(\bar{Z},s_{\mathrm{init}},s_{\mathrm{goal}})$ such that $\bar{C}(z_R) = C(z)$. 
\end{lemma}

\begin{proof}
From the proof of Proposition \ref{prop:reduced_is_traj}, we get that $z$ can be partitioned as $z = z_1 \dots z_L$ where each subtrajectory $z_i$ is either a state-input pair $(p_i,b_i)$ of $\bar{Z}$ or is contained in a subtree $H_i$ in $Z$ that gets replaced by a state-input pair $p_i$ in $\bar{Z}$. By Equation \eqref{eq:z_R_equation}, $z_R = (p_i,b_i)_{i=1}^L$. To see this, note that the last state-action pair $(e_L,b_L)$ of $z_L$ is such that $\psi(e_L,b_L) = s_{\mathrm{goal}}$. Hence, $z_L$ exits $p_L$, so $z_R$ is given by the first row of Equation \eqref{eq:z_R_equation}. Moreover, $p_1 = s_{\mathrm{init}}$ since the first state of $z$ is $s_{\mathrm{init}}$ which is also a state of $\bar{Z}$. 

We now show that $\bar{\psi}(p_L,b_L) = s_{\mathrm{goal}}$ similarly to the last part of the proof of Proposition \ref{prop:reduced_is_traj}. Let $(e_L,b_L)$ be the last state-action pair of $z_L$. Since $z$ reaches $s_{\mathrm{goal}}$, we have $\psi(e_L,b_L) = s_{\mathrm{goal}}$. Moreover, $\psi(e_L,b_L) = \psi(p_L,b_L)$ since $z_L$ exits $p_L$, so $\psi(p_L,b_L) = s_{\mathrm{goal}}$. Note that the recursive evaluation of $\bar{\psi}(p_L,b_L)$ is identical to $\psi(p_L,b_L)$ except that $\bar{\psi}(p_L,b_L)$ terminates at the $\bar{Z}$-state containing $\psi(p_L,b_L)$, in this case, $s_{\mathrm{goal}}$. Hence, $\bar{\psi}(p_L,b_L) = s_{\mathrm{goal}}$. 

It remains to show that $\bar{C}(z_R) = C(z)$. By optimality, each $z_i$ must exit $p_i$ with input $b_i$ optimally, that is $C(z_i) = c^{p_i}_{b_i}+\chi(p_i,b_i)$. By construction, this cost coincides with $\bar{\chi}(p_i,b_i)$. Hence,
\begin{equation*} 
\bar{C}(z_R) = \sum_{i=1}^L \bar{\chi}(p_i,b_i) = \sum_{i=1}^L C(z_i) = C(z).
\end{equation*}
The proof is complete.
\end{proof}


\begin{lemma}\label{lemma:planning_equivalence_2}
Let $z$ be a feasible trajectory of $(\bar{Z},s_{\mathrm{init}},s_{\mathrm{goal}})$. Then $z_E = (q_i,a_i)_{i=1}^N$ is a feasible trajectory of $({Z},s_{\mathrm{init}},s_{\mathrm{goal}})$ such that $C(z_E) = \bar{C}(z)$.
\end{lemma}

\begin{proof}
Let $z = (p_i,b_i)_{i=1}^L$ be given. Since $p_1 = s_{\mathrm{init}}$, we have by Algorithm \ref{alg:optimal_expansion} that $q_1 = s_{\mathrm{init}}$. Moreover $\bar{\psi}(p_L,b_L) = s_{\mathrm{goal}}$ and since the expansion $(p_L,b_L)_E$ of the last state-input pair $(p_L,b_L)$ of $z$ is by Proposition \ref{prop:optimal_expansion_is_trajectory} an optimal $b_L$-exit trajectory (of the MM $N_{p_L}$ corresponding to $p_L$), we have that $\psi(q_N,a_N) = s_{\mathrm{goal}}$. Finally, by construction, $\bar{\chi}(p_i,b_i) = c^{p_i}_{b_i}+\chi(p_i,b_i)$ so
\begin{equation*}
\bar{C}(z) = \sum_{i=1}^L \bar{\chi}(p_i,b_i) = \sum_{i=1}^L c^{p_i}_{b_i}+\chi(p_i,b_i),
\end{equation*}
and since $(p_i,b_i)_E$ is an optimal $b_i$-exit trajectory (of the MM $N_{p_i}$ corresponding to $p_i$) we have that its cumulated cost is $c^{p_i}_{b_i}+\chi(p_i,b_i)$, hence $C(z_E) = C( \, ( \, (p_i,b_i)_E \, )_{i=1}^L\,) = \sum_{i=1}^L c^{p_i}_{b_i}+\chi(p_i,b_i) = \bar{C}(z)$.
\end{proof}

\begin{proof}[Proof of Theorem~\ref{th:planning_equivalence}]
We now provide the proof of Theorem \ref{th:planning_equivalence}.

\textbf{(i).} Let $z$ be an optimal trajectory to $({Z},s_{\mathrm{init}},s_{\mathrm{goal}})$. Consider the reduced trajectory $z_R$. By Lemma \ref{lemma:planning_equivalence_1}, $z_R$ is a feasible trajectory to $(\bar{Z},s_{\mathrm{init}},s_{\mathrm{goal}})$ with $\bar{C}(z_R) = C(z)$. Assume by contradiction that there exists a feasible trajectory $w$ of $(\bar{Z},s_{\mathrm{init}},s_{\mathrm{goal}})$ such that $\bar{C}(w) < \bar{C}(z_R)$. By Lemma \ref{lemma:planning_equivalence_2}, we have that $w_E$ is a feasible trajectory to $({Z},s_{\mathrm{init}},s_{\mathrm{goal}})$ with cumulative cost
\begin{equation*}
{C}(w_E) = \bar{C}(w) < \bar{C}(z_R) = C(z),
\end{equation*}
which is a contradiction since $z$ is an optimal trajectory to $({Z},s_{\mathrm{init}},s_{\mathrm{goal}})$. This proves (i).

\textbf{(ii).} Let $z$ be an optimal trajectory to $(\bar{Z},s_{\mathrm{init}},s_{\mathrm{goal}})$. By Lemma \ref{lemma:planning_equivalence_2}, $z_E$ is a feasible trajectory to $({Z},s_{\mathrm{init}},s_{\mathrm{goal}})$ such that $C(z_E) = \bar{C}(z)$. Assume by contradiction that there exists a feasible trajectory $w$ such that $C(w) < C(z_E)$. We may assume $w$ is an optimal trajectory to $({Z},s_{\mathrm{init}},s_{\mathrm{goal}})$.\footnote{To see it, note that if $w$ is not an optimal trajectory, then there exits a feasible trajectory $w'$ such that $C(w') < C(w)$. Set $w'$ as our new $w$ and repeated this process until there is no $w'$ such that $C(w') < C(w)$. Note that this procedure must indeed eventually terminate since there is a finite number of states and transitions. When it terminates, we have a feasible trajectory $w$ such that no feasible trajectory $w'$ exists with $C(w') < C(w)$, i.e., $w$ is an optimal trajectory.} By Lemma \ref{lemma:planning_equivalence_1}, $w_R$ is a feasible trajectory to $(\bar{Z},s_{\mathrm{init}},s_{\mathrm{goal}})$ such that $\bar{C}(w_R) = {C}(w)$. Then
\begin{equation*}
\bar{C}(w_R) = \bar{C}(w) < C(z_E) = \bar{C}(z),
\end{equation*}
which is a contradiction since $z$ is an optimal trajectory to $(\bar{Z},s_{\mathrm{init}},s_{\mathrm{goal}})$. This proves (ii) and thus completes the proof.
\end{proof}

\subsubsection{Proof of Lemma \ref{lemma:pass_through}}
We now provide the proof of Lemma~\ref{lemma:pass_through} in Section~\ref{optimal_planner}:

\begin{proof}[Proof of Lemma \ref{lemma:pass_through}]
Let $z = (q_i,a_i)_{i=1}^N$ be an optimal trajectory to $(\bar{Z},s_\mathrm{init},s_\mathrm{goal})$. Note that $s_\mathrm{goal}$ is one of the nested states within the subtree with root $\bar{D}_\beta$. Thus, the only way to reach $s_\mathrm{goal}$ from a state outside this subtree is via a transition going through $\bar{D}_\beta$ ending up at $\mathrm{start}(\bar{D}_\beta)$. Since $s_\mathrm{init}$ is by construction outside this subtree, one transition of $z$ must go through $\bar{D}_\beta$ and thus end up at $\mathrm{start}(\bar{D}_\beta)$ (note that this transition might be the last transition). That is, some state-input pair $(q_k,a_k)$ of $z$ is such that $\bar{\psi}(q_k,a_k) = \mathrm{start}(\bar{D}_\beta)$, completing the proof.
\end{proof}

\subsubsection{Solve: Algorithm and Correctness}\label{solve_algorithm_and_correctness}

We continue by proving Proposition \ref{proposition:solve_correct}, assuring the correctness of Algorithm \ref{alg:solve}. To this end, we present a more formal and detailed version of Algorithm \ref{alg:solve}, given by Algorithm \ref{alg:solve_detailed_version}, calling subroutines given by Algorithm \ref{alg:add_nodes} to \ref{alg:add_arcs}. Before proving Proposition \ref{proposition:solve_correct}, it is beneficial to give an overview of this detailed version, complementing the overview given in Section \ref{solve}. More precisely, Part 1 of Algorithm \ref{alg:solve_detailed_version} first constructs the nodes of the graph $G$ by calling $\mathrm{Add\_nodes}$ (line \ref{alg:solve_detailed_version:line2}). Then, it precomputes transitions saving which node one goes to in $G$ if one enters or exits an $\bar{U}_i$, done by $\mathrm{Compute\_transitions}$ and saved as a look-up table $\tilde{\psi}$ (line \ref{alg:solve_detailed_version:line3}).\footnote{This step is mainly done to keep the time complexity low since precomputing these transitions has a lower time complexity than computing them when needed.} After this, the algorithm search locally in each $\bar{U}_i$ given by the for-loop (lines \ref{alg:solve_detailed_version:line4}--\ref{alg:solve_detailed_version:line8}). In particular, the algorithm uses the subroutine $\mathrm{Get\_search\_states}$ (line \ref{alg:solve_detailed_version:line5}) to find states in $\bar{U}_i$ to search from, given by $I$ in Algorithm \ref{alg:get_search_states} (lines \ref{alg:get_search_states:line1}--\ref{alg:get_search_states:line10}), and states to search to, given by $D$ in Algorithm \ref{alg:get_search_states} (lines \ref{alg:get_search_states:line14}--\ref{alg:get_search_states:line22}). Here, $\mathrm{exit}_a$ in $D$ is an added state the machine $\bar{U}_i$ transits to whenever an input $a$ is unsupported, just to book-keep $a$-exits. Note that $I$ is indeed the set of all states one can start from in $\bar{U}_i$ (i.e., $s_\mathrm{init}$ or $s_{\bar{U}_i}$ or any state in $\bar{U}_i$ we get to by exiting $\bar{U}_{i-1}$) while $D$ captures all ways one can escape $\bar{U}_i$ (i.e., entering $\bar{U}_{i-1}$ or $\bar{D}_\beta$, or exiting $\bar{U}_i$). After the subroutine $\mathrm{Get\_search\_states}$, the algorithm search in $\bar{U}_i$ from one $s \in I$ to each $d \in D$ using Dijkstra's algorithm, with output being the cost $c_d$ to go to $d$ and the corresponding trajectory $z_d$ (line \ref{alg:solve_detailed_version:line7}). An arc from $s$ to the corresponding state of $d$ is then added to $G$ (line \ref{alg:solve_detailed_version:line8}), where the corresponding state of $d$ is found using the look-up table $\tilde{\psi}$. For example, if $d = \bar{U}_{i-1}$, then the corresponding state is $\mathrm{start}(\bar{U}_{i-1})$, so the arc will go from $s$ to $\mathrm{start}(\bar{U}_{i-1})$ in $G$. Furthermore, the arc is labelled with the trajectory $z_d$ and its corresponding cost in $\bar{Z}$.\footnote{This cost is just $c_d$ if $d \in \{ \bar{U}_{i-1}, \bar{D}_{\beta} \}$. If $d = \mathrm{exit}_a$, then we need to also add the transition cost for the transition happening in the machine $\bar{U}_k$ that $z_d$ exits to, given by $\gamma_{\bar{U}_{k}}(\bar{U}_{k-1},a)$, see line \ref{alg:add_arcs:line7} in Algorithm \ref{alg:add_arcs}. Finally, we only add an arc to $G$ if $z_d$ is feasible (i.e., $\bar{Z}$ does not stop), checked by the if-statements in Algorithm \ref{alg:add_arcs} (lines \ref{alg:add_arcs:line1}, \ref{alg:add_arcs:line5} and \ref{alg:add_arcs:line10}).} This procedure is done for all $s \in I$, given by the for-loop in lines \ref{alg:solve_detailed_version:line6}--\ref{alg:solve_detailed_version:line8}. After each $\bar{U}_i$ has been searched in locally (lines \ref{alg:solve_detailed_version:line4}--\ref{alg:solve_detailed_version:line8}), $G$ is constructed and a bidirectional Dijkstra is carried out to find a trajectory from $s_\mathrm{init}$ to $\bar{D}_\beta$, which by construction is a trajectory from $s_\mathrm{init}$ to $\mathrm{start}(\bar{D}_\beta)$. Part 2 of Algorithm \ref{alg:solve_detailed_version} is identical to Algorithm \ref{alg:solve}, already explained in detail in Section~\ref{solve}.

With this overview, we are in a position to prove Proposition~\ref{proposition:solve_correct}.

\begin{algorithm}[t]
\caption{Solve: Detailed version}\label{alg:solve_detailed_version} 
\begin{algorithmic}[1]
\Require $(\bar{Z},s_\mathrm{init},s_\mathrm{goal})$, $U_{path}$, $D_{path}$, $\beta$ and $(c_a^M,z_a^M)_{a \in \Sigma, M \in X}$.
\Ensure An optimal trajectory $z$ to $(\bar{Z},s_\mathrm{init},s_\mathrm{goal})$.
\State \textbf{Part 1: Solve $(\bar{Z},s_\mathrm{init},\mathrm{start}(\bar{D}_\beta))$}
\State $G \gets \mathrm{Add\_nodes}(\bar{Z},s_\mathrm{init},s_\mathrm{goal})$ \Comment{Add nodes to $G$} \label{alg:solve_detailed_version:line2}
\State $\tilde{\psi} \gets \mathrm{Compute\_transitions}(\bar{Z}, U_{path})$ \Comment{Precompute transitions in $G$} \label{alg:solve_detailed_version:line3}
\For {$i =1,\dots,n$} \Comment{Consider $\bar{U}_{i}$ \label{alg:solve_detailed_version:line4}}
\State $(I,D) \gets \mathrm{Get\_search\_states}(\bar{Z},s_\mathrm{init},i)$ \Comment{States to search from and to given by $I$ and $D$ respectively} \label{alg:solve_detailed_version:line5}
\For {$s \in I$ \label{alg:solve_detailed_version:line6}}
\State $(c_d, z_d)_{d \in D} \gets \mathrm{Dijkstra}(s, D, \bar{U}_{i})$ \label{alg:solve_detailed_version:line7}
\State $G \gets \mathrm{Add\_arcs}(G,s,\bar{Z},i, (c_d, z_d)_{d \in D},\tilde{\psi})$ \Comment{Add arcs to $G$} \label{alg:solve_detailed_version:line8}
\EndFor
\EndFor
\State $z_1 \gets \mathrm{Bidirectional\_Dijkstra}(s_\mathrm{init}, \bar{D}_\beta, G)$\label{alg:solve_detailed_version:line11}
\State \textbf{Part 2: Solve $(\bar{Z},\mathrm{start}(\bar{D}_\beta),s_\mathrm{goal})$}
\State Let $z_2$ be an empty trajectory \Comment{To be constructed}\label{alg:solve_detailed_version:line13}
\State Let $\bar{D}_{i}$ be the MM with state $\mathrm{start}(\bar{D}_\beta)$ \label{alg:solve_detailed_version:line14}
\State $c \gets \mathrm{start}(\bar{D}_\beta)$\label{alg:solve_detailed_version:line15} 
\State $g \gets$ state in $\bar{D}_{i}$ corresponding to $\bar{D}_{i-1}$ \label{alg:solve_detailed_version:line16} 
\While {$c \neq s_\mathrm{goal}$\label{alg:solve_detailed_version:line17}}
\State $z \gets \mathrm{Bidirectional\_Dijkstra}(c,g,\bar{D}_{i})$ \label{alg:solve_detailed_version:line18} 
\State $z_2 \gets z_2 z$ \Comment{Concatenate $z_2$ with $z$}\label{alg:solve_detailed_version:line19}
\State $c \gets \mathrm{start}(\bar{D}_{i-1})$  \label{alg:solve_detailed_version:line20}
\If {$c \neq s_\mathrm{goal}$\label{alg:solve_detailed_version:line21}}
\State Let $\bar{D}_{i}$ be the MM with state $c$
\State $g \gets$ corresponding state in $M$ to $\bar{D}_{i-1}$  \label{alg:solve_detailed_version:line23} 
\EndIf
\EndWhile \Comment{End of Part 2}
\State return $z \gets z_1 z_2$ \Comment{Concatenate $z_1$ with $z_2$}  \label{alg:solve_detailed_version:line26}
\end{algorithmic}
\end{algorithm}

\begin{algorithm}[h]
\caption{Add\_nodes}\label{alg:add_nodes}
\begin{algorithmic}[1]
\Require $\bar{Z},U_{path},D_{path},\beta$ and empty graph $G$. 
\Ensure Graph $G$ with added nodes.
\State Add nodes $s_\mathrm{init}$, $s_{\bar{U}_1}$ and $\bar{D}_\beta$ to $G$. 
\For {$i =2,\dots,n$\label{alg:add_nodes:line2}}
\If {$s_{\bar{U}_i}\neq \bar{U}_{i-1}$}
\State Add node $s_{\bar{U}_i}$ to $G$
\EndIf
\For {$a \in \Sigma$\label{alg:add_nodes:line6}} 
\If {$\delta_{\bar{U}_i}(\bar{U}_{i-1},a) \neq \emptyset$} 
\State Add node $\delta_{\bar{U}_i}(\bar{U}_{i-1},a)$ to $G$
\EndIf
\EndFor
\EndFor
\State return $G$
\end{algorithmic}
\end{algorithm}

\begin{algorithm}[h] 
\caption{Compute\_transitions}\label{alg:compute_exit_transitions}
\begin{algorithmic}[1]
\Require $\bar{Z}$ and $U_{path}$ 
\Ensure Function $\tilde{\psi}$ prescribing: where one goes if exiting $\bar{U}_i$ with $a \in \Sigma$ (treating $\bar{B}$ as a state), saved as $\tilde{\psi}(\bar{U}_{i},a)$; or where one goes if entering $\bar{U}_i$ (treating $\bar{D}_\beta$ as a state), saved as $\tilde{\psi}(\bar{U}_{i})$.
\For {$a \in \Sigma$} \Comment{Compute exit transitions}
\State $\mathrm{exit}_a \gets \emptyset$
\EndFor
\For {$i = n-1,\dots,1$\label{alg:compute_exit_transitions:line4}}
\For {$a \in \Sigma$\label{alg:compute_exit_transitions:line5}}
\If {$\delta_{\bar{U}_{i+1}}(\bar{U}_{i},a) \neq \emptyset$}
\State $\mathrm{exit}_a \gets \delta_{\bar{U}_{i+1}}(\bar{U}_{i},a)$
\EndIf
\State $\tilde{\psi}(\bar{U}_{i},a) \gets \mathrm{exit}_a$
\EndFor
\EndFor
\State $\mathrm{go\_here} \gets s_\mathrm{init}$ \Comment{Compute enter transitions}
\State $\tilde{\psi}(\bar{U}_0) \gets s_\mathrm{init}$
\For {$i = 1,\dots,n-1$}
\If {$s_{\bar{U}_i} \neq \bar{U}_{i-1}$}
\State $\mathrm{go\_here} \gets s_{\bar{U}_i}$
\EndIf
\State $\tilde{\psi}(\bar{U}_i) \gets \mathrm{go\_here}$
\EndFor
\State {return} $\tilde{\psi}$
\end{algorithmic}
\end{algorithm}

\begin{algorithm}[h]
\caption{Get\_search\_states}\label{alg:get_search_states}
\begin{algorithmic}[1]
\Require $\bar{Z},s_\mathrm{init},i$
\Ensure States to search from $I$ and states to search too $D$ in $\bar{U}_i$
\State $I = []$ \Comment{States to search from} \label{alg:get_search_states:line1}
\If {$s(\bar{U}_i) \neq \bar{U}_{i-1}$}
\State Append $s_{\bar{U}_i}$ to $I$.
\EndIf
\If {$i=1$}
\State Append $s_\mathrm{init}$ to $I$
\Else
\For {$a \in \Sigma$\label{alg:get_search_states:line8}} 
\If {$\delta_{\bar{U}_i}(\bar{U}_{i-1},a) \neq \emptyset$} 
\State Append $\delta_{\bar{U}_i}(\bar{U}_{i-1},a)$ to $I$ \label{alg:get_search_states:line10}
\EndIf
\EndFor
\EndIf
\State $D = []$ \Comment{States to search to} \label{alg:get_search_states:line14}
\If {$i > 1$} 
\State Append $\bar{U}_{i-1}$ to $D$
\EndIf
\If {$i < n$} \Comment{Add exits as destinations}
\State Append $\{\mathrm{exit}_a\}_{a \in \Sigma}$ to $D$\label{alg:get_search_states:line19} 
\EndIf
\If {$i = \alpha+1$} 
\State Append $\bar{D}_{\beta}$ to $D$ \label{alg:get_search_states:line22}
\EndIf
\State {return} $(I,D)$
\end{algorithmic}
\end{algorithm}

\begin{algorithm}[h]
\caption{Add\_arcs}\label{alg:add_arcs}
\begin{algorithmic}[1]
\Require $G,s,\bar{Z},i, (c_d, z_d)_{d \in D},\tilde{\psi}$
\Ensure Graph $G$ with added arcs.
\If {$i>1$ and $c_{\bar{U}_{i-1}} < \infty$\label{alg:add_arcs:line1}}

\State Add arc $s \rightarrow \tilde{\psi}(\bar{U}_{i-1})$ to $G$ labelled $(c_{\bar{U}_{i-1}}, z_{\bar{U}_{i-1}})$ 
\EndIf
\For {$\mathrm{exit}_a \in D$\label{alg:add_arcs:line4}}
\If {$\tilde{\psi}(\bar{U}_{i},a) \neq \emptyset$ and $c_{\mathrm{exit}_a} < \infty$\label{alg:add_arcs:line5}}
\State Let $\bar{U}_{k}$ be the machine with state $\tilde{\psi}(\bar{U}_{i},a)$
\State Add arc $s \rightarrow \tilde{\psi}(\bar{U}_{i},a)$ to $G$ labelled $(c_{\mathrm{exit}_a},z_{\mathrm{exit}_a}+\gamma_{\bar{U}_{k}}(\bar{U}_{k-1},a))$ \label{alg:add_arcs:line7} 
\EndIf
\EndFor
\If {$i = \alpha+1$ and $c_{\bar{D}_{\beta}} < \infty$\label{alg:add_arcs:line10}} 
\State Add arc $s \rightarrow \bar{D}_{\beta}$ to $G$ labelled $(c_{\bar{D}_{\beta}}, z_{\bar{D}_{\beta}})$
\EndIf
\State {return} $G$
\end{algorithmic}
\end{algorithm}

\begin{proof}[Proof of Proposition \ref{proposition:solve_correct}]
Consider any optimal trajectory $w$ to $(\bar{Z},s_\mathrm{init},s_\mathrm{goal})$. We may assume that $w$ does not contain any loops (due to optimality). Let $\bar{U}$ be the subtree given by $\{\bar{U}_1,\dots,\bar{U}_n\}$ (with all states as leaves), and $\bar{D}$ be the subtree given by $\{\bar{D}_1,\dots,\bar{D}_{\beta}\}$ (with all states as leaves). By Lemma \ref{lemma:pass_through}, we can divide $w$ into $w=w_1w_2$ where $w_1$ is a trajectory from $s_\mathrm{init}$ to $\bar{D}_\beta$ in $\bar{U}$ and $w_2$ is a trajectory from $\mathrm{start}(\bar{D}_\beta)$ to $s_\mathrm{goal}$ in $\bar{D}$. Consider the output trajectory $z = z_1z_2$ from Algorithm \ref{alg:solve}. We want to show that $\bar{C}(z_1) \leq \bar{C}(w_1)$ and $\bar{C}(z_2) \leq \bar{C}(w_2)$, since then $\bar{C}(z) \leq \bar{C}(w)$ and $z$ is an optimal trajectory to $(\bar{Z},s_\mathrm{init},s_\mathrm{goal})$. We prove these two inequalities separately. 

\textbf{Part 1: $\bar{C}(z_1) \leq \bar{C}(w_1)$.} We can further divide $w_1$ into $w_1 = w_1^1 \dots w_1^K$ where each $w_1^j$ is a trajectory in some $\bar{U}_i$, where we may assume consecutive $w_1^{j}$ and $w_1^{j+1}$ are in different $\bar{U}_i$. Consider one $w_1^k$ and assume it is in $\bar{U}_i$. Note that $w_1^k$ must start from some state $s$ in $I$ given by Algorithm \ref{alg:get_search_states}. Furthermore, $w_1^k$ must end by entering $\bar{U}_{i-1}$ or $\bar{D}_\beta$, or exiting $\bar{U}_i$. All these ways of ending are captured by all the $d \in D$ in Algorithm \ref{alg:get_search_states}, with optimal trajectories given by $(z_d)_{d \in D}$. Hence, there exists an optimal trajectory $\tilde{w}_1^k$, computed by the local searches of Algorithm \ref{alg:solve_detailed_version}, that goes from the start state of $w_1^k$ to the state it reaches. Due to optimality, $\bar{C}(\tilde{w}_1^k) \leq \bar{C}(w_1^k)$. Note that also $\tilde{w}_1^k$ corresponds to an arc in $G$ labelled by $\tilde{w}_1^k$ and its corresponding cost $\bar{C}(\tilde{w}_1^k)$. Therefore, the trajectory 
\begin{equation*}
\tilde{w}_1 = \tilde{w}_1^1 \dots \tilde{w}_1^K
\end{equation*}
corresponds to a directed path of $G$ from $s_\mathrm{init}$ to $\bar{D}_\beta$ with cost 
\begin{equation*}
\bar{C}(\tilde{w}_1) = \sum_{k=1}^K \bar{C}(\tilde{w}_1^k) \leq \sum_{j=1}^K \bar{C}(w_1^k) = \bar{C}({w}_1).
\end{equation*}
Since $z_1$ is by construction an optimal directed path in $G$ from $s_\mathrm{init}$ to $\bar{D}_\beta$, we have $\bar{C}(z_1)  \leq \bar{C}(\tilde{w}_1)$, and thus
\begin{equation*}
\bar{C}(z_1) \leq \bar{C}({w}_1).
\end{equation*}
This completes the first part of the proof.

\textbf{Part 2: $\bar{C}(z_2) \leq \bar{C}(w_2)$.} Let $c$, $g$ and $\bar{D}_i$ be as in lines \ref{alg:solve_detailed_version:line14}--\ref{alg:solve_detailed_version:line16} in Algorithm \ref{alg:solve_detailed_version}. Note that $w_2$ starts at $c$ and can only reach $s_\mathrm{goal}$ by going to $g$. Hence, $w_2$ must have a sub-trajectory $w_2^1$ going from $c$ to $g$ in $\bar{D}_i$. Since $z$ in line \ref{alg:solve_detailed_version:line18} is an optimal trajectory from $c$ to $g$ in $\bar{D}_i$, we have that $\bar{C}(z) \leq \bar{C}(w_2^1)$. Thus, the sub-trajectory $z$ of $z_2$ going from $c$ to $g$ in $\bar{D}_i$ has cost $\leq$ the cost of the sub-trajectory $w_2^1$ of $w_2$. Note also that these sub-trajectories reaches the same state further down in the hierarchy, the new $c$ (line \ref{alg:solve_detailed_version:line20}) with a corresponding new $g$ (line \ref{alg:solve_detailed_version:line23}). Repeating this argument (but now for the while-loop), we get that the next sub-trajectory of $z_2$ also has a cost $\leq$ the cost of the next sub-trajectory of $w_2$. By induction (over the number of times the while-loop is called until $c = s_\mathrm{goal}$), we get that each sub-trajectory of $z_2$ has cost less than or equal to the corresponding sub-trajectory of $w_2$, and thus $\bar{C}(z_2) \leq \bar{C}(w_2)$.

By Part 1 and Part 2 above, we conclude that $\bar{C}(z) \leq \bar{C}(w)$. Thus, $z$ (line~\ref{alg:solve_detailed_version:line26}) is an optimal trajectory to $(\bar{Z},s_\mathrm{init},s_\mathrm{goal})$.
\end{proof}

\subsubsection{Expand: Algorithm} We now provide the full algorithm of $\mathrm{Expand}$ (line \ref{alg:hierarchical_planning_new:line10} of Algorithm \ref{alg:hierarchical_planning_new}), given by Algorithm \ref{alg:expand}. The procedure is almost identical to the optimal expansion, given by Algorithm \ref{alg:optimal_expansion}. Indeed, there are only two differences. The first difference is that line \ref{alg:optimal_expansion:line8} in Algorithm \ref{alg:optimal_expansion} is changed to return only the input $a$. In this way, Algorithm \ref{alg:expand} outputs just the input sequence of the optimal expansion, i.e., an optimal plan $u$. The other difference is that we construct the optimal plan using lists (lines \ref{alg:expand:line1} and \ref{alg:expand:line12}) that we sequentially append to (lines \ref{alg:expand:line4} and \ref{alg:expand:line15}) until we have the full optimal plan (lines \ref{alg:expand:line6} and \ref{alg:expand:line17}). This latter difference is just to get a lower time complexity. Note also that this version of $\mathrm{Expand}$ returns the whole optimal plan $u$ at once. For a sequential execution of $u$, obtaining the next optimal input of $u$ in a sequential manner,  one simply change [return $a$] in line \ref{alg:expand:line10} to [execute $a$]. In this way, we will sequentially execute the optimal plan $u$, which could be beneficial for time-critical applications where the next optimal input should be applied as fast as possible. That the algorithm indeed executes the optimal plan in the correct order is just a consequence of the order of the expansion.

\subsubsection{Time Complexity} Finally, we prove the time complexity of the Online Planner given by Proposition \ref{efficient_time_complexity}.
The proof is rather straightforward: We first compute the time complexities of all called algorithms of the Online Planner, then we combine these time complexities to infer the time complexity of the Online Planner. We start by proving the following lemma.
\begin{lemma}\label{lemma:time_complexity}
We have the following time complexities:
\begin{itemize}
\item Algorithm \ref{alg:reduce} ($\mathrm{Reduce}$) has time complexity $O( \mathrm{depth}(Z) b_s )$.
\item Algorithm \ref{alg:compute_alpha_and_beta} ($\mathrm{Compute\_paths}$) has time complexity~$O(\mathrm{depth}(Z))$.
\item Algorithm \ref{alg:add_nodes} ($\mathrm{Add\_nodes}$) has time complexity~$O(|\Sigma| \mathrm{depth}(Z))$.
\item Algorithm \ref{alg:compute_exit_transitions} ($\mathrm{Compute\_transitions}$) has time complexity~$O(|\Sigma| \mathrm{depth}(Z))$.
\item Algorithm \ref{alg:get_search_states} ($\mathrm{Get\_search\_states}$) has time complexity~$O(|\Sigma|)$.
\item Algorithm \ref{alg:add_arcs} ($\mathrm{Add\_arcs}$) has time complexity~$O(|\Sigma|)$.
\end{itemize}
\end{lemma}
\begin{proof}[Proof of Lemma \ref{lemma:time_complexity}]
We prove the time complexities separately:
\begin{itemize}
\item
It is easy to see that Algorithm~\ref{alg:compute_alpha_and_beta} has time complexity $O(\mathrm{depth}(Z))$ since the while-loops in it can at most iterate $O(\mathrm{depth}(Z))$ times. 
\item 
We now show the time complexity of Algorithm \ref{alg:reduce}, by walking through the algorithm. Line \ref{alg:reduce:line1} has time complexity $O(\mathrm{depth}(Z))$ (due to the time complexity of Algorithm~\ref{alg:compute_alpha_and_beta}), while lines \ref{alg:reduce:line2}--\ref{alg:reduce:line3} has time complexity $O(1)$. Concerning line \ref{alg:reduce:line4}, note first that $\bar{U}_1$ is identical to $U_1$ except from the modified cost function $\bar{\gamma}$ given by Equation \ref{eq:bar_gamma}. Moreover, given access to the original cost function ${\gamma}$, we can specify the function $\bar{\gamma}$ in constant time $O(1)$ and also obtain any specific value $\bar{\gamma}(q,a)$ in constant time $O(1)$ by just evaluating Equation \eqref{eq:bar_gamma}.\footnote{For the curious reader, we here provide the details that Equation \eqref{eq:bar_gamma} can be evaluated in constant time $O(1)$. To this end, note first that $\delta(q,a) \neq \emptyset$ can be evaluated in time $O(1)$ (assuming access to the values of $\delta$). Moreover, condition $({U}_i \xrightarrow{q} M) \in {T}, M \in \M$ translates into determining if $M \in \{ U_{i-1},D_\beta \}$ such that $({U}_i \xrightarrow{q} M) \in {T}$, and can therefore also be checked in time $O(1)$. Thus the conditions in Equation \eqref{eq:bar_gamma} can be checked in time $O(1)$. Moreover, given access to the values of $\gamma$ and $c_a^q$, the values in Equation \eqref{eq:bar_gamma} can be obtained in time $O(1)$. We conclude that Equation \eqref{eq:bar_gamma} can be evaluated in time $O(1)$ (given access to $\delta$, $\gamma$, $U_{i+1}$ and $D_\beta$).} Thus, we can construct $\bar{U}_1$ in time $O(1)$ by having a reference to $U_1$. Thus, line \ref{alg:reduce:line4} has time complexity $O(1)$. Continuing, we get that line \ref{alg:reduce:line5} can be done in time $O(b_s)$ (specifying all arcs), and line \ref{alg:reduce:line6} in time $O(1)$. Analogously, we get that line \ref{alg:reduce:line8} has time complexity $O(1)$, line \ref{alg:reduce:line9} has time complexity $O(b_s)$, line \ref{alg:reduce:line10} has time complexity $O(1)$ and lines \ref{alg:reduce:line11}--\ref{alg:reduce:line12} have time complexity $O(1)$. Thus, the for-loop in lines \ref{alg:reduce:line7}--\ref{alg:reduce:line12} has time complexity $O(\mathrm{depth}(Z) b_s)$ (since the loop has at most $\mathrm{depth}(Z)$ iterations). Finally, lines \ref{alg:reduce:line15}--\ref{alg:reduce:line16} have time complexity $O(1)$. We conclude that the $U_\mathrm{path}$-part (lines \ref{alg:reduce:line4}--\ref{alg:reduce:line16}) can be done in time $O(\mathrm{depth}(Z) b_s)$. A similar analysis shows that the $D_\mathrm{path}$-part (lines \ref{alg:reduce:line18}--\ref{alg:reduce:line27}) also has time complexity $O(\mathrm{depth}(Z) b_s)$. Hence, Algorithm \ref{alg:reduce} has time complexity~$O(\mathrm{depth}(Z) b_s)$.
\item
Algorithm  \ref{alg:add_nodes} has time complexity $O(|\Sigma| \mathrm{depth}(Z))$ due to its two for-loops (lines \ref{alg:add_nodes:line2} and \ref{alg:add_nodes:line6}) going through $O(|\Sigma| \mathrm{depth}(Z))$ elements. 
\item
Algorithm \ref{alg:compute_exit_transitions} has time complexity $O(|\Sigma| \mathrm{depth}(Z))$ due to the two for-loops at lines \ref{alg:compute_exit_transitions:line4} and \ref{alg:compute_exit_transitions:line5} going through $O(|\Sigma| \mathrm{depth}(Z))$ elements.
\item
Algorithm \ref{alg:get_search_states} has time complexity $O(|\Sigma|)$ since the computation bottlenecks are given by the for-loop on line \ref{alg:get_search_states:line8} and the append operation on line \ref{alg:get_search_states:line19}, both with time complexity $O(|\Sigma|)$.
\item
Algorithm \ref{alg:add_arcs} has time complexity $O(|\Sigma|)$ due to the for-loop at line \ref{alg:add_arcs:line4}. 

\end{itemize}
\end{proof}

Using Lemma \ref{lemma:time_complexity}, we now show the time complexity of $\mathrm{Solve}$, given by Algorithm \ref{alg:solve_detailed_version}:
\begin{lemma}\label{lemma:time_complexity_solve}
The time complexity of $\mathrm{Solve}$ given by Algorithm \ref{alg:solve_detailed_version} equals
\begin{align*}
O \Big ( \depth(Z) |\Sigma| \cdot \Big [ b_s |\Sigma| + (b_s+|\Sigma| ) \log \big ( (b_s+|\Sigma| )\big ) \Big]  \Big ) + \nonumber \\
O\Big (|\Sigma|^2 \depth(Z)+ |\Sigma| \depth(Z) \log \big ( |\Sigma| \depth(Z) \big ) \Big ) + \nonumber \\
O\Big (\depth(Z) \cdot \big ( b_s |\Sigma|+b_s\log(b_s) \big ) \Big )
\end{align*}
\end{lemma}

\begin{algorithm}[t]
\caption{Expand}\label{alg:expand}
\begin{algorithmic}[1]
\Require HiMM$ \; Z =(X,T)$, $\{z^{M}_a\}^{M \in X}_{a \in \Sigma}$ from Optimal Exit Computer, and trajectory $z = (q_i,a_i)_{i=1}^N$ of $\bar{Z}$.
\Ensure An optimal plan $\mathrm{optimal\_plan}$
\State $\mathrm{optimal\_plan} \gets []$. \label{alg:expand:line1}
\For{$(q_i,a_i)$ in $z$\label{alg:expand:line2}}
\State $u_i \gets \mathrm{Input\_expand}(q_i,a_i)$
\State Append $u_i$ to $\mathrm{optimal\_plan}$ \label{alg:expand:line4}
\EndFor
\State return $\mathrm{optimal\_plan}$ \label{alg:expand:line6} 
\State 
\State $\mathrm{Input\_expand}(q,a)$: \Comment{Expand recursively} 
\If {$q \in S_Z$}
\State return $a$. \label{alg:expand:line10}
\Else
\State $\mathrm{optimal\_exit} \gets []$ \label{alg:expand:line12}
\For {$(q_i,a_i)$ in $z^{N_q}_a$\label{alg:expand:line13}} \Comment{$z^{N_q}_a$ from Optimal Exit Computer}
\State $u_i  \gets  \mathrm{Input\_expand}(q_i,a_i)$
\State Append $u_i$ to $\mathrm{optimal\_exit}$ \label{alg:expand:line15}
\EndFor
\State return $\mathrm{optimal\_exit}$ \label{alg:expand:line17} 
\EndIf 
\end{algorithmic}
\end{algorithm}

\begin{proof}
 We start with Part 1. Line \ref{alg:solve_detailed_version:line2}--\ref{alg:solve_detailed_version:line3} has time complexity $O(|\Sigma| \depth(Z))$. Line \ref{alg:solve_detailed_version:line5} has time complexity $O(|\Sigma|)$. Line \ref{alg:solve_detailed_version:line7} has time complexity $O(E+V\log(V))$, where $E$ is the number of arcs and $V$ is the number of nodes in the Dijkstra search (considering the machine as a graph).\footnote{We assume the Dijkstra search uses a Fibonacci heap due to optimal time complexity, see \cite{DijkstraFibonacci} for details.} In our case, $V = |Q|+|\Sigma|$ due to the extra $|\Sigma|$ nodes $\{\mathrm{exit}_a\}_{a \in \Sigma}$ we add to $\bar{U}_i$ to represent the different exits (see Section \ref{solve_algorithm_and_correctness}).\footnote{Adding the extra nodes $\{\mathrm{exit}_a\}_{a \in \Sigma}$ to $\bar{U}_i$ can be done in time $O(|Q| |\Sigma|)$ by checking every state $q$ and input $a$ of $\bar{U}_i$ and adding an arc to $\mathrm{exit}_a$ if $\delta_{\bar{U}_i}(q,a) = \emptyset$. This is lower than the time complexity of the search itself, and will therefore not contribute to the time complexity of Line \ref{alg:solve_detailed_version:line7}.} Furthermore, $E = |Q| |\Sigma|$ since all nodes in $\bar{U}_i$ has $|\Sigma|$ arcs except $\{\mathrm{exit}_a\}_{a \in \Sigma}$ which does not have any arcs. Thus, line \ref{alg:solve_detailed_version:line7} has time~complexity 
\begin{equation*}
O \Big (b_s |\Sigma| + (b_s+|\Sigma| ) \log \big ( (b_s+|\Sigma| ) \big ) \Big ).
\end{equation*} 
Moreover, line \ref{alg:solve_detailed_version:line8} has time complexity $O(|\Sigma |)$. Since the number of elements of $I$ is bounded by $|\Sigma|+2$, we get that the for-loop line \ref{alg:solve_detailed_version:line6}--\ref{alg:solve_detailed_version:line8} has time complexity
\begin{align*}
O \Big ( [|\Sigma|+2] \Big [ b_s |\Sigma| + (b_s+|\Sigma| ) \log \big ( (b_s+|\Sigma| ) \big ) \Big]  \Big ) = \\
O \Big ( |\Sigma| \Big [ b_s |\Sigma| + (b_s+|\Sigma| ) \log \big ( (b_s+|\Sigma| ) \big ) \Big]  \Big )
\end{align*} 
The time complexity for line \ref{alg:solve_detailed_version:line4}--\ref{alg:solve_detailed_version:line8} is therefore
\begin{align}\label{eq:time_comp_1}
O \Big ( \depth(Z) |\Sigma| \Big [ b_s |\Sigma| + (b_s+|\Sigma| ) \log \big ( (b_s+|\Sigma| )\big ) \Big]  \Big ),
\end{align} 
which is also the time complexity for line \ref{alg:solve_detailed_version:line2}--\ref{alg:solve_detailed_version:line8}. Finally, the time complexity for bidirectional Dijkstra is the same as for Dijkstra's algorithm \cite{bast2016route}. Thus, line \ref{alg:solve_detailed_version:line11} has time complexity $O(E+V\log(V))$ where $E$ and $V$ is now the number of arcs and nodes of $G$. The number of nodes $V$ is bounded by $(|\Sigma|+2) \depth(Z)$ since we have at most $|\Sigma|+2$ nodes of $G$ in $\bar{U}_i$. Thus, $V \leq C |\Sigma| \depth(Z)$ for some constant $C>1$ and
\begin{align*}
V \log (V) \leq C |\Sigma| \depth(Z) \log \big ( C |\Sigma| \depth(Z) \big ) = \\
C |\Sigma| \depth(Z) \Big ( \log(C)+\log \big (|\Sigma| \depth(Z) \big ) \Big ) \leq \\
\tilde{C} |\Sigma| \depth(Z) \log \big ( |\Sigma| \depth(Z) \big ) 
\end{align*}
for some constant $\tilde{C}>1$ as long as $|\Sigma|>1$ (which is an implicit and reasonable  assumption in this paper since we are looking at planning problems and $|\Sigma|=1$ would imply no planning since there is then only one input to pick). Thus, we have that 
\begin{align*} 
V \log(V) = O \Big (  |\Sigma| \depth(Z) \log \big ( |\Sigma| \depth(Z) \big ) \Big ).
\end{align*}
Furthermore, the number of arcs is bounded by $|\Sigma| V \leq \Sigma| (|\Sigma|+2) \depth(Z)$. Hence, the time complexity for line \ref{alg:solve_detailed_version:line11} is
\begin{align}\label{eq:time_comp_2}
O(|\Sigma|^2 \depth(Z)+V\log(V)) = \nonumber \\
O\Big (|\Sigma|^2 \depth(Z)+ |\Sigma| \depth(Z) \log \big ( |\Sigma| \depth(Z) \big ) \Big )
\end{align}
The time complexity for Part 1 is given by adding Equation \eqref{eq:time_comp_1} and \eqref{eq:time_comp_2}:
\begin{align}\label{eq:time_comp_part1}
O \Big ( \depth(Z) |\Sigma| \cdot \Big [ b_s |\Sigma| + (b_s+|\Sigma| ) \log \big ( (b_s+|\Sigma| )\big ) \Big]  \Big ) + \nonumber \\
O\Big (|\Sigma|^2 \depth(Z)+ |\Sigma| \depth(Z) \log \big ( |\Sigma| \depth(Z) \big ) \Big )
\end{align}
We continue with the time complexity of Part 2. Line \ref{alg:solve_detailed_version:line13} has time complexity $O(1)$, line \ref{alg:solve_detailed_version:line14} $O(\depth(Z))$, line \ref{alg:solve_detailed_version:line15} $O(1)$ and line \ref{alg:solve_detailed_version:line16} $O(1)$. Let's now consider the while-loop line \ref{alg:solve_detailed_version:line17}--\ref{alg:solve_detailed_version:line23}. Line \ref{alg:solve_detailed_version:line18} has time complexity $O(E+V\log(V))$ where $E$ and $V$ are the number of arcs and nodes of $\bar{D}_i$ (seen as a graph). We have $V = |Q| \leq b_s$ and $E \leq |Q| |\Sigma| \leq b_s |\Sigma|$. Thus, line \ref{alg:solve_detailed_version:line18} has time complexity $O(b_s |\Sigma|+b_s\log(b_s))$. Since the while-loop can be called at most $\depth(Z)$ times, the total contribution line \ref{alg:solve_detailed_version:line18} has to the time complexity is $O\Big (\depth(Z) \cdot \big ( b_s |\Sigma|+b_s\log(b_s) \big ) \Big )$. Lines \ref{alg:solve_detailed_version:line19} and \ref{alg:solve_detailed_version:line21}--\ref{alg:solve_detailed_version:line23} all have $O(1)$ in time complexity, thus having total contribution $O(\depth(Z))$ over the while-loop. Finally, line \ref{alg:solve_detailed_version:line20} calls the function $\mathrm{start}$, which in turn recursively calls itself until it settles at a state further down in the hierarchy of machines. The total number of times $\mathrm{start}$ can recursively call itself over the whole while-loop is bounded by $\depth(Z)$, since this is the maximum number of steps one can go down in the hierarchy. Thus, line \ref{alg:solve_detailed_version:line20} has time complexity contribution $O(\depth(Z))$ over the whole while-loop. We conclude that the while-loop lines \ref{alg:solve_detailed_version:line17}--\ref{alg:solve_detailed_version:line23} has time complexity 
\begin{equation}\label{eq:time_comp_part2}
O\Big (\depth(Z) \cdot \big ( b_s |\Sigma|+b_s\log(b_s) \big ) \Big ),
\end{equation}
which is also the time complexity of Part 2. The time complexity for $\mathrm{Solve}$ is obtained by adding Equation \eqref{eq:time_comp_part1} and \eqref{eq:time_comp_part2}:
\begin{align*}
O \Big ( \depth(Z) |\Sigma| \cdot \Big [ b_s |\Sigma| + (b_s+|\Sigma| ) \log \big ( (b_s+|\Sigma| )\big ) \Big]  \Big ) + \nonumber \\
O\Big (|\Sigma|^2 \depth(Z)+ |\Sigma| \depth(Z) \log \big ( |\Sigma| \depth(Z) \big ) \Big ) + \nonumber \\
O\Big (\depth(Z) \cdot \big ( b_s |\Sigma|+b_s\log(b_s) \big ) \Big ).
\end{align*}
This completes the proof.
\end{proof}
We now continue with the the time complexity of $\mathrm{Expand}$:
\begin{lemma}\label{lemma:time_complexity_expand}
The time complexity of $\mathrm{Expand}$ given by Algorithm \ref{alg:expand} equals $O(\mathrm{depth}(Z))$ for obtaining the next input $u$ and $O(\mathrm{depth}(Z) |u|)$ for obtaining the whole optimal plan $u$ at once, where $|u|$ is the length of $u$.
\end{lemma}
\begin{proof}
Consider the recursive function $\mathrm{Input\_expand}$ that Algorithm \ref{alg:expand} uses. This function goes down the tree of $Z$ by recursively expanding trajectories from the Optimal Exit Computer in a depth-first manner. This depth-first recursion can happen at most $\depth(Z)$ times before one of the recursive calls of $\mathrm{Input\_expand}$ returns an input $a$ (line \ref{alg:expand:line10}). After an input $a$ (line \ref{alg:expand:line10}) has been returned, there is possibly several returns of $\mathrm{Input\_expand}$ function calls higher up in the hierarchy, maximally $\depth(Z)$ returns, before one of the $\mathrm{Input\_expand}$ function calls considers the next $(q_i,a_i)$ in its for-loop (line \ref{alg:expand:line13}), or the for-loop on line \ref{alg:expand:line2} considers the next $(q_i,a_i)$. This next $(q_i,a_i)$ is then considered by $\mathrm{Input\_expand}$, again with a depth-first recursion having at most $\depth(Z)$ function calls. Since going down recursively and going up consists of just $O(1)$ operations (e.g., appending $u_i$ to $\mathrm{optimal\_exit}$ at line \ref{alg:expand:line15} has time complexity $O(1)$), the time complexity for arriving (again) at line \ref{alg:expand:line10}, is bounded by the number of returns and calls we maximally can have before before arriving at line \ref{alg:expand:line10}. By above, this bound is $O(\depth(Z))$. Thus, the time complexity for obtaining the next optimal input is $O(\depth(Z))$. In particular, the time complexity for obtaining the whole optimal plan $u$ at once is $O(\mathrm{depth}(Z) |u|)$.
\end{proof}
We can now easily deduce Proposition \ref{efficient_time_complexity}:
\begin{proof}[Proof of Proposition \ref{efficient_time_complexity}]
Proposition \ref{efficient_time_complexity} follows by combining Lemma \ref{lemma:time_complexity} (stating the time complexity of $\mathrm{Reduce}$), with Lemma \ref{lemma:time_complexity_solve} (stating the time complexity of $\mathrm{Solve}$) and Lemma \ref{lemma:time_complexity_expand} (stating the time complexity of $\mathrm{Expand}$).
\end{proof}



\subsection{Reconfigurable Hierarchical Planning: Additional Details}
In this section, we provide details concerning the marking procedure and prove Proposition \ref{prop:optimal_exit_costs_correct_reconfig} and  \ref{prop:reconfigurability} given in Section \ref{reconfig_hier_planning}, starting with the former.

\begin{algorithm}[t]
\caption{Mark}\label{alg:mark}
\begin{algorithmic}[1]
\Require Modification $m$ and (modified) HiMM $Z$
\Ensure Appropriate marking of $Z$
\State Let $M$ be the MM targeted by $m$
\State $\mathrm{mark\_itself}(M,Z)$ \label{alg:mark:line2}
\State $\mathrm{mark\_ancestors}(M,Z)$ \label{alg:mark:line3}
\end{algorithmic}
\end{algorithm}

\subsubsection{Marking Procedure}
The details of the marking procedure for a modification $m$ is given by Algorithm~\ref{alg:mark}. Here, $M$ is the MM targeted by $m$, specified in Definition \ref{def:state_addition} to \ref{def:composition} for each type of modification. To understand the marking, note that when modifying $M$, the optimal exit costs of $M$ may change and thus needs to be recomputed (line \ref{alg:mark:line2}). Moreover, if $M$ is nested within other MMs, then these ancestor MMs of $M$ might also need to recompute their optimal exit costs (line \ref{alg:mark:line3}). All other MMs are not affected by the change.

\begin{convention}\label{marking_convention}
We use the convention that an HiMM has all MMs marked when initialised since no optimal exit costs have been computed yet. In this way, we can compute the optimal exit costs in the same way whether we have never computed them or just need to update some due to a change. Moreover, we always mark an HiMM correspondingly when modifying it, thereby always maintaining correctly marked HiMMs.
\end{convention}

\subsubsection{Proofs}
We continue with the proofs.

\begin{proof}[Proof of Proposition \ref{prop:optimal_exit_costs_correct_reconfig}]
Let HiMM $Z =(X,T)$ be given. The key observation for showing Proposition \ref{prop:optimal_exit_costs_correct_reconfig} is to note that the marked MMs of $Z$ will always either form a subtree of $T$ that includes the root of $T$, or be empty (where the latter is the case if all MMs in $Z$ have optimal exit costs that are up-to-date). Therefore, Algorithm \ref{alg:update_offline_step} (starting at the root of $T$) will always reach all MMs that are marked and since the computation for a specific MM is correct (Proposition \ref{th:compute_optimal_exits_correct}), we conclude that Algorithm \ref{alg:update_offline_step} indeed returns the correct optimal exit costs and corresponding trajectories. Thus, it remains to show the key observation (i.e., that the marked MMs of $Z$ will always either form a subtree of $T$ that includes the root of $T$, or be empty), which is intuitively clear, but needs a technical proof to account for the HiMMs that one adds to form $Z$. The proof is given below, where we for brevity say that a HiMM $Z$ that fulfils the key observation has the subtree-property.

We prove the key observation by induction. More precisely, note first that $Z$ is formed by starting with a set $\mathcal{Z}$ of initialised HiMMs that through modifications eventually form $Z$. Let $n$ be the number of times we add HiMMs in this procedure using either state addition or composition. We prove that $Z$ has the subtree-property by induction over $n$. 

\subsubsection*{Base Case}
Consider first the base case $n=0$. In this case, $\mathcal{Z}$ can only contain one element\footnote{Since if $\mathcal{Z}$ would contain at least two elements, then these HiMMs would eventually need to be composed or added to one another (to form $Z$ in the end), which is not possible if $n=0$.}, call it $Z_1$. Therefore, we only start with $Z_1$, which is first initialised and changed through modifications to eventually form $Z$, where we possibly also compute some of the optimal exit costs using Algorithm \ref{alg:update_offline_step} along the way (see Convention \ref{marking_convention}). Formally, this procedure is nothing but a sequence of changes on the form
\begin{equation}\label{eq:proof_sequence}
Z_1 \xrightarrow{o_1} Z^{(1)} \xrightarrow{o_2} Z^{(2)} \xrightarrow{o_3} \dots \xrightarrow{o_m} Z
\end{equation}
where $o_i$ is either a modification, or $o_i = \star$ specifying that we call Algorithm \ref{alg:update_offline_step}. More precisely, for a modification $o_i$, $Z \xrightarrow{o_i} Z'$ means that we changed $Z$ to $Z'$ through the modification $o_i$, while $Z \xrightarrow{\star} Z'$ means that we called Algorithm~\ref{alg:update_offline_step} on $Z$ with $Z'$ as a result (in the latter case, note that $Z'$ only differs by which MMs are marked). We will prove that all HiMMs in \eqref{eq:proof_sequence} fulfil the subtree-property, by induction over the sequence \eqref{eq:proof_sequence}. For the base case, note that $Z_1$ trivially fulfils the subtree-property since all MMs are marked (due to initialisation, see Convention \ref{marking_convention}). For the induction step, assume that $W$ is an HiMM in \eqref{eq:proof_sequence} that fulfils the subtree-property, and changed to $W'$ through $o$. We prove that $W'$ then also fulfils the subtree-property, separating the proof into cases:
\begin{enumerate}[(i)]
\item Consider first the case when $o = \star$. Then we call Algorithm \ref{alg:update_offline_step} on $W$ to get $W'$. Since $W$ has the subtree-property, Algorithm \ref{alg:update_offline_step} will reach all marked MMs and unmark them. Hence, $W'$ has no marked MMs, and therefore fulfils the subtree-property.
\item Consider now the case when $o$ is a modification. If $o$ is a state subtraction or arc modification, on MM $M$ in $W$ say, then we mark $M$ and all the ancestors of $M$ in $W'$ (due to Convention \ref{marking_convention}).\footnote{In particular, we mark the root MM of $W'$.} Therefore, $W'$ fulfils the subtree-property. If $o$ is a state addition, then $o$ cannot add any HiMM $Z_{\mathrm{add}}$ to an MM $M$ (since this would contract $n=0$). Thus, $o$ can only add a state to an MM $M$. In this case, $M$ and all ancestors of are marked, so $W'$ again fulfils the subtree-property. Finally, $o$ cannot be a composition since $n=0$.
\end{enumerate}
Combining the two cases, we conclude that $W'$ fulfils the subtree-property. By induction over the sequence \eqref{eq:proof_sequence}, all HiMMs in \eqref{eq:proof_sequence} fulfil the subtree-property. In particular, $Z$ fulfils the subtree-property. We have therefore proved the base case $n=0$.

\subsubsection*{Induction Step}
We continue with the induction step. Towards this, assume $n\geq1$. Let $Z_1$ be the last HiMM that is formed by either adding a HiMM using state addition, or using composition.\footnote{Note that $Z_1$ must indeed exist since we start with the HiMMs in $\mathcal{Z}$ but eventually are left with just $Z$.} Thus, from $Z_1$, $Z$ is formed by a sequence of changes as in \eqref{eq:proof_sequence}.\footnote{This follows from the fact that we cannot add another HiMM from $Z_1$ to $Z$, since $Z_1$ was the last HiMM that was formed by adding a HiMM.} Therefore, by just following the proof for the case $n=0$, we know that $Z$ fulfils the subtree-property if $Z_1$ fulfils the subtree-property. It remains to prove that $Z_1$ fulfils the subtree-property. By assumption, $Z_1$ is formed by adding a HiMM using state addition or composition. We consider the two cases separately:
\begin{enumerate}
\item Assume first that $Z_1$ is formed by state addition, adding an HiMM $Z_{\mathrm{add}}$ to an HiMM $Z_{\mathrm{to}}$. Let $n_{\mathrm{add}}$ ($n_{\mathrm{to}}$) be the number of times we add HiMMs to form $Z_{\mathrm{add}}$ ($Z_{\mathrm{to}}$) using either state addition or composition, similar to $n$. Then, $n_{\mathrm{add}} < n$ and $n_{\mathrm{to}} < n$.\footnote{This is because every time we add an HiMM to form e.g., $Z_{\mathrm{add}}$, this also counts as adding an HiMM to form $Z$, and we add at least one more HiMM to form $Z$, due to the formation of $Z_1$. Hence, $n_{\mathrm{add}} < n$. The argument for $n_{\mathrm{to}} < n$ is analogous.} By induction, $Z_{\mathrm{add}}$ and $Z_{\mathrm{to}}$ fulfils the subtree-property. Let $Z_{\mathrm{add}}$ be added to the MM $M$ in $Z_{\mathrm{to}}$. Then $M$ and its ancestors are marked\footnote{Note that this connects the marked subtree of $Z_{\mathrm{add}}$ with the marked subtree of $Z_{\mathrm{to}}$.} in $Z_1$ and since both $Z_{\mathrm{add}}$ and $Z_{\mathrm{to}}$ fulfils the subtree property, we conclude that $Z_1$ fulfils the~subtree-property.
\item Assume now that $Z_1$ is formed by composition using an MM $M$. By induction, analogous to the state addition case, all the HiMMs in the composition fulfils the subtree-property, and since $M$ (the root of $Z_1$) is marked, we conclude that $Z_1$ also fulfils the~subtree-property.
\end{enumerate}
Combining the two cases, we conclude that $Z_1$ fulfils the subtree-property, hence, $Z$ fulfils the subtree-property. 

We have showed both the base case and the induction step. Therefore, by induction over $n$, $Z$ always fulfils the subtree-property, proving the key observation. Therefore, the proof of Proposition \ref{prop:optimal_exit_costs_correct_reconfig} is~complete.
\end{proof}

We continue with the proof of Proposition \ref{prop:reconfigurability}.

\begin{proof}[Proof of Proposition \ref{prop:reconfigurability}]
Let HiMM $Z=(X,T)$ be given. Following the proof of Proposition \ref{th:compute_optimal_exits_time}, it is easy to conclude that if $K$ MMs are marked in $Z$, then the time complexity of Algorithm \ref{alg:update_offline_step} is $O(K [b_s |\Sigma|+(b_s+|\Sigma|) \log(b_s+|\Sigma|)])$. Since $Z$ has no marked MMs before the modification $m$, the number of MMs in $Z$ after the modification is at most $\mathrm{depth}(Z)$ if $m$ is a state addition, state subtraction, or arc modification, and at most 1 if $m$ is a composition. Hence, marking takes time $O(\mathrm{depth}(Z))$ for state addition, state subtraction, or arc modification, and $O(1)$ composition. Furthermore, using the observation above, running Algorithm \ref{alg:update_offline_step} therefore takes time $O(\mathrm{depth}(Z) [b_s |\Sigma|+(b_s+|\Sigma|) \log(b_s+|\Sigma|)])$ for state addition, state subtraction, or arc modification, and $O( 1 [b_s |\Sigma|+(b_s+|\Sigma|) \log(b_s+|\Sigma|)])$ for composition. We conclude that the time complexity for executing the Optimal Exit Computer if $m$ is a state addition, state subtraction, or arc modification is $O( \mathrm{depth}(Z) [b_s |\Sigma|+(b_s+|\Sigma|) \log(b_s+|\Sigma|)] )$, and $O([b_s |\Sigma|+(b_s+|\Sigma|) \log(b_s+|\Sigma|)] )$ if $m$ is a composition.
\end{proof}

\subsection{Hierarchical Planning with Identical Machines: Additional Details}
We here provide additional details concerning the Online Planner given an HiMM with identical machines, complementing Section \ref{optimal_planner_duplicate_case}. More precisely, the modified reduce step for an HiMM with identical machines is given by Algorithm \ref{alg:reduce_duplicate_case}, where the procedure for computing $\{ U_1,\dots,U_n \}$ and $\{ D_1, \dots, D_m\}$ is given by Algorithm \ref{alg:compute_alpha_and_beta_duplicate_case}. In particular, the modifications in Algorithm \ref{alg:reduce_duplicate_case} are coloured red. These modifications correspond to obtaining the state $q$ that ${U}_i$ (or ${D}_i$) corresponds to. For an HiMM without identical machines, this could be done by checking the parent of ${U}_i$ ($D_i$). For an HMM with identical machines, this is not possible due to ambiguity, and we hence get this information from $s_\mathrm{init}$ ($s_\mathrm{goal}$) instead.  Algorithm \ref{alg:compute_alpha_and_beta_duplicate_case} has been modified significantly, hence, we omit colouring. Here, the main difference is that we instead use $s_\mathrm{init}$ and $s_\mathrm{goal}$ directly to obtain $\{ U_1,\dots,U_n \}$ and $\{ D_1, \dots, D_m\}$, starting from the root of the tree and sequentially proceed~downwards.

\begin{algorithm}[t]
\caption{Reduce [for HiMM with identical machines]}\label{alg:reduce_duplicate_case}
\begin{algorithmic}[1]
\Require $(Z,s_{\mathrm{init}},s_{\mathrm{goal}})$ and $(c_a^M,z_a^M)_{a \in \Sigma, M \in X})$
\Ensure $(\bar{Z},U_{path},D_{path},\alpha,\beta)$
\State $(U_{path}, D_{path}, \alpha, \beta) \gets$ Compute\_paths($Z,s_\mathrm{init},s_\mathrm{goal}$)
\State Let $\bar{X}$ be an empty set \Comment{To be constructed}
\State Let $\bar{T}$ be an empty tree \Comment{To be constructed}
\State Construct $\bar{U}_1$ from $U_1$ and add to $\bar{X}$
\State Add node $\bar{U}_1$ to $\bar{T}$ with arcs out as $U_1$ in $T$ but no destination nodes
\State \textcolor{BrickRed}{$q \gets s_\mathrm{init}[-2]$} \Comment{The second last element of $s_\mathrm{init}$.}
\For {$i = 2,\dots,\mathrm{length}(U_{path})-1$}
\State Construct $\bar{U}_i$ from $U_i$ and add to $\bar{X}$
\State Add node $\bar{U}_i$ to $\bar{T}$ with arcs out as $U_i$ in $T$ but no destination nodes
\State Add destination node $\bar{U}_{i-1}$ to $\bar{U}_{i} \xrightarrow{q} \emptyset$ so that $(\bar{U}_{i} \xrightarrow{q}\bar{U}_{i-1}) \in \bar{T}$
\If{$i <\mathrm{length}(U_{path})-1$}
\State \textcolor{BrickRed}{$q \gets s_\mathrm{init}[-(i+1)]$}
\EndIf
\EndFor
\If {$\beta = 0$} \Comment{Then nothing more to do}
\State {return} $(\bar{Z} = (\bar{X},\bar{T}), U_{path})$
\EndIf 
\State Construct $\bar{D}_1$ from $D_1$ and add to $\bar{X}$
\State Add node $\bar{D}_1$ to $\bar{T}$ with arcs out as $D_1$ in $T$ but no destination nodes
\State \textcolor{BrickRed}{$q \gets s_\mathrm{goal}[-2]$} \Comment{The second last element of $s_\mathrm{goal}$.}
\For {$i = 2,\dots,\beta$}
\State Construct $\bar{D}_i$ from $D_i$ and add to $\bar{X}$
\State Add node $\bar{D}_i$ to $\bar{T}$ with arcs out as $D_i$ in $T$ but no destination nodes
\State Add destination node $\bar{D}_{i-1}$ to $\bar{D}_{i} \xrightarrow{q} \emptyset$ so that $(\bar{D}_{i} \xrightarrow{q}\bar{D}_{i-1}) \in \bar{T}$
\State \textcolor{BrickRed}{$q \gets s_\mathrm{goal}[-(i+1)]$}
\EndFor
\State Add destination node $\bar{D}_{\beta}$ to $\bar{U}_{\alpha+1} \xrightarrow{q} \emptyset$ so that $(\bar{U}_{\alpha+1} \xrightarrow{q} \bar{D}_{\beta}) \in \bar{T}$
\State {return} $(\bar{Z} = (\bar{X},\bar{T}),U_{path})$
\end{algorithmic}
\end{algorithm}

\begin{algorithm}
\caption{Compute\_paths [for HiMM with identical machines]}\label{alg:compute_alpha_and_beta_duplicate_case}
\begin{algorithmic}[1]
\Require $(Z,s_\mathrm{init},s_\mathrm{goal})$
\Ensure $(U_\mathrm{path}, D_\mathrm{path}, \alpha, \beta)$
\State $U_\mathrm{current} \gets Z.\mathrm{root}$ \Comment{Compute $U_\mathrm{path}$}
\State $U_\mathrm{path} \gets [U_\mathrm{current}]$
\State $ l \gets \mathrm{length}(s_\mathrm{init})$
\For{$i=1,\dots,l-1$}
\State $q \gets s_\mathrm{init}[i]$ \Comment{$s_\mathrm{init}[1]$ is the first entry of $s_\mathrm{init}$}
\State $U_\mathrm{current} \gets U_\mathrm{current}.\mathrm{children}[q]$
\State $U_\mathrm{path}.\mathrm{append}\_\mathrm{first}(U_\mathrm{current})$ \Comment{Append first in $U_\mathrm{path}$}
\EndFor
\State $U_\mathrm{path}.\mathrm{append}\_\mathrm{first}(s_\mathrm{init}[l])$ \Comment{Append $U_0$}
\State $D_\mathrm{current} \gets Z.\mathrm{root}$ \Comment{Compute $D_\mathrm{path}$}
\State $D_\mathrm{path} \gets [D_\mathrm{current}]$
\State $ l \gets \mathrm{length}(s_\mathrm{goal})$
\For{$i=1,\dots,l-1$}
\State $q \gets s_\mathrm{goal}[i]$
\State $D_\mathrm{current} \gets D_\mathrm{current}.\mathrm{children}[q]$
\State $D_\mathrm{path}.\mathrm{append}\_\mathrm{first}(D_\mathrm{current})$
\EndFor
\State $D_\mathrm{path}.\mathrm{append}\_\mathrm{first}(s_\mathrm{goal}[l])$ \Comment{Append $D_0$}
\State $(i,j) \gets (\mathrm{length}(U_\mathrm{path})-1,\mathrm{length}(D_\mathrm{path})-1)$ \Comment{Compute $\alpha$ and $\beta$}
\State $k \gets 1$
\While {$s_\mathrm{init}[k] = s_\mathrm{goal}[k]$}
\State $(i,j) \gets (i-1,j-1)$
\State $k \gets k+1$
\EndWhile
\State $(\alpha,\beta) \gets (i-1,j-1)$
\State {return} $(U_\mathrm{path}, D_\mathrm{path}, \alpha, \beta)$
\end{algorithmic}
\end{algorithm}

\subsection{Implementation Details}\label{implementation_details}
In this section, we provide some additional details concerning the implementation. All implementations are coded in Python.

\subsubsection{Dijkstra's algorithm and Bidirectional Dijkstra}
Concerning Dijkstra's algorithm and Bidirectional Dijkstra, the lowest time complexity is theoretically achieved by having Fibonacci heaps as priority queues. However, in practice, simple priority queues, i.e., using a binary heap, are often faster, hence, we use the latter in our implementation.

\subsubsection{Contraction Hierarchies}
To discuss our implementation of Contraction Hierarchies, we start with a general summary of the algorithm, referring to \cite{geisberger2012exact} for additional details. More precisely, Contraction Hierarchies is a shortest path algorithm for weighted graphs\footnote{We obtain an equivalent weighted graph of an HiMM by flattening the HiMM and treat the resulting MM as a weighted graph.} consisting of a preprocessing step and a query step. The preprocessing step removes (contracts) nodes iteratively adding additional arcs, called shortcuts, in the graph to compensate for the contracted nodes (whenever needed). The shortcuts are added in a way to preserve the shortest-path distances. Then, in the query step, the whole original graph plus the shortcuts are searched through using a modified version of Bidirectional Dijkstra to find a shortest path between any two nodes. The Bidirectional Dijkstra is modified in a way such that the backward and forward searches are constrained by the iterative order of when the nodes were contracted. Thus, the contraction order plays an important role for achieving efficient searches. 
There are several heuristics how to pick the contraction order. In our case, we based our implementation on the lecture notes by Hannah Bast, \url{https://ad-wiki.informatik.uni-freiburg.de/teaching/EfficientRoutePlanningSS2012}, and the paper \cite{geisberger2012exact}. In particular, for the preprocessing step determining the contraction order, we used lazy evaluations, selecting the next node $n$ to contract based on the edge difference of $n$ plus the number of neighbours to $n$ already contracted (to account for sparcity). We refer to the lecture notes for details. We stress that there might exist more efficient combinations of heuristics. 

\subsubsection{Differences With Our Planning Algorithm}

Finally, we highlight differences between our planning algorithm and Contraction Hierarchies, starting with the preprocessing step followed by the query step. 

\textbf{Preprocessing step.}
As opposed to Contraction Hierarchies, the preprocessing step of our planning algorithm does not work with a flat representation of the HiMM. Instead, it actively exploits the hierarchy of the HiMM in a principled manner, recursively going through the tree structure of the HiMM to compute exit costs of each MM. These exit costs is the reason our planning algorithm can ``contract'' (i.e., reduce) certain MMs in the query step (resulting in a significant speed-up), and can be seen as the analogue of shortcuts for Contraction Hierarchies (where the latter contract nodes instead of MMs). However, the contraction order of Contraction Hierarchies is instead based on heuristics and may therefore results in creating substantial unnecessary shortcuts (i.e., shortcuts that merely increase the computational load of the preprocessing step, but is not actively used in the query step). Consequently, the preprocessing step of Contraction Hierarchies may be unnecessary slow, and time complexity guarantees seems (to the best of our knowledge) be difficult to obtain in the HiMM-setting, whereas our planning algorithm have formal time complexity guarantees. The preprocessing step of Contraction Hierarchies is also agnostic to the hierarchy which makes it non-trivial to take into account modifications and identical machines. Furthermore, to make Contraction Hierarchies work efficiently, one needs in general some time to obtain a good combination of heuristics. In contrast, our planning algorithm contains nothing that needs to be~chosen.

\textbf{Query step.}
The query step of our planning algorithm uses Bidirectional Dijkstra and Dijkstra's algorithm, while Contraction Hierarchies uses a modified version of Bidirectional Dijkstra. More precisely, our planning algorithm constrain the search space of Bidirectional Dijkstra (and Dijkstra's algorithm) by reducing MMs in the HiMM, while Contraction Hierarchies instead enforces the contraction order onto Bidirectional Dijkstra (working with the flat representation of the HiMM). Depending on the which contraction heuristics are used in Contraction Hierarchies, these two approaches result in different search orders in general.


\fi

\end{document}